\documentclass[11pt]{article} 

\usepackage[margin=1in]{geometry}

\usepackage{etex}

\usepackage{xspace,enumerate}
\usepackage[dvipsnames]{xcolor}
\usepackage[T1]{fontenc}
\usepackage[full]{textcomp}
\usepackage[american]{babel}
\usepackage{mathtools}
\usepackage{amsmath,amssymb,amsthm,amsfonts}

\usepackage{newpxtext} 
\usepackage{textcomp} 
\usepackage[varg,bigdelims]{newpxmath}
\let\mathbb\varmathbb
\usepackage{bm} 
\usepackage[scr=rsfso]{mathalfa}

\usepackage{microtype}
\usepackage[
pagebackref,
colorlinks=true,
urlcolor=blue,
linkcolor=blue,
citecolor=OliveGreen,
]{hyperref}
\usepackage[capitalise,nameinlink]{cleveref}
\crefname{lemma}{Lemma}{Lemmas}
\crefname{fact}{Fact}{Facts}
\crefname{theorem}{Theorem}{Theorems}
\crefname{corollary}{Corollary}{Corollaries}
\crefname{claim}{Claim}{Claims}
\crefname{example}{Example}{Examples}
\crefname{algorithm}{Algorithm}{Algorithms}
\crefname{problem}{Problem}{Problems}
\crefname{definition}{Definition}{Definitions}
\usepackage{paralist}
\usepackage{turnstile}
\usepackage{tikz}
\usepackage{caption}
\DeclareCaptionType{Algorithm}
\usepackage{newfloat}
\usepackage{comment}

\newtheorem{theorem}{Theorem}[section]
\newtheorem*{theorem*}{Theorem}
\newtheorem{proposition}[theorem]{Proposition}
\newtheorem*{proposition*}{Proposition}
\newtheorem{lemma}[theorem]{Lemma}
\newtheorem*{lemma*}{Lemma}
\newtheorem{corollary}[theorem]{Corollary}
\newtheorem*{conjecture*}{Conjecture}
\newtheorem{fact}[theorem]{Fact}
\newtheorem*{fact*}{Fact}

\newtheorem*{hypothesis*}{Hypothesis}

\theoremstyle{definition}
\newtheorem{definition}[theorem]{Definition}
\newtheorem*{definition*}{Definition}

\newtheorem{example}[theorem]{Example}

\newtheorem{problem}[theorem]{Problem}

\theoremstyle{remark}
\newtheorem{claim}[theorem]{Claim}
\newtheorem*{claim*}{Claim}
\newtheorem{remark}[theorem]{Remark}
\newtheorem*{remark*}{Remark}

\newtheorem*{observation*}{Observation}

\usepackage{cleveref}
\usepackage[most]{tcolorbox}
\newtcbtheorem{Theorem}{Problem}{
  enhanced,
  sharp corners,
  attach boxed title to top left={
    yshifttext=-1mm
  },
  colback=white,
  colframe=blue!55,
  fonttitle=\bfseries,
  boxed title style={
    sharp corners,
    size=small,
    colback=blue!55,
    colframe=blue!55,
  } 
}{thm}

 \usepackage{xcolor}
\usepackage{todonotes}

\newcommand{\C}{\mathbb{C}}	                
\newcommand{\R}{\mathbb{R}}                     
               
\newcommand{\poly}{\text{poly}}
	
\newcommand{\eps}{\epsilon}
\newcommand{\RRR}{\mathcal{R}}

\newcommand{\CCC}{\mathcal{C}}

\newcommand{\ZZ}{\mathcal{Z}}
\newcommand{\pr}[1]{\text{\bf Pr}\normalfont\lbrack #1 \rbrack} 
\newcommand{\ex}[1]{\mathbb{E}\normalfont\lbrack #1 \rbrack}
\newcommand{\bpr}[1]{\text{\bf Pr}\normalfont \Big[#1 \Big]} 
\newcommand{\bex}[1]{\mathbb{E}\normalfont \Big[#1 \Big]}
\newcommand{\bvar}[1]{\mathbf{Var}\normalfont \Big(#1 \Big)}

\def\Tr{\text{Tr}}
\def\KF{\textsf{KF}}
\def\AA{\mathbf{A}}
\def\BB{\mathbf{B}}
\def\EE{\mathbf{E}}
\def\MM{\mathbf{M}}
\def\NN{\mathbf{N}}

\def\CC{\mathbf{C}}
\def\DD{\mathbf{D}}

\def\Y{\mathbf{Y}}
\def\X{\mathbf{X}}
\def\P{\mathbf{P}}
\def\U{\mathbf{U}}

\def\V{\mathbf{V}}

\def\S{\mathbf{S}}
\def\T{\mathbf{T}}
\def\W{\mathbf{W}}

\def\G{\mathbf{G}}
\def\M{\mathbf{M}}
\def\ZZ{\mathbf{Z}}

\usepackage{multirow}
\usepackage{hhline}
\usepackage{makecell}

\newcommand{\PP}{\mathbb{P}}

\newcommand{\wt}[1]{\widetilde{#1}}

\title{Testing Positive Semi-Definiteness via Random Submatrices}
\author{
Ainesh Bakshi\thanks{Ainesh Bakshi and Rajesh Jayaram would like to thank the partial support
from the Office of Naval Research (ONR) grant N00014-18-1-2562, and the National Science Foundation (NSF) under
Grant No. CCF-1815840.}\\
CMU \\
\texttt{abakshi@cs.cmu.edu}
\and 
Nadiia Chepurko \\
MIT \\
\texttt{nadiia@mit.edu} 
\and 
Rajesh Jayaram\footnotemark[1] \\
CMU \\
\texttt{rkjayara@cs.cmu.edu} 
}

\begin{document}
\maketitle
\begin{abstract}
    We study the problem of testing whether a matrix $\AA \in \R^{n \times n}$ with bounded entries ($\|\AA\|_\infty \leq 1$)  is positive semi-definite (PSD), 
    or $\eps$-far in Euclidean distance from the PSD cone, meaning that $\min_{\BB \succeq 0} \|\AA - \BB\|_F^2  > \eps n^2$, where $\BB \succeq 0$ denotes that $\BB$ is PSD.
    Our main algorithmic contribution is a non-adaptive tester which distinguishes between these cases using only $\tilde{O}(1/\eps^4)$ queries to the entries of $\AA$.\footnote{Throughout the paper, $\tilde{O}(\cdot)$ hides $\log(1/\eps)$ factors.} If instead of the Euclidean norm we considered the distance in spectral norm, we obtain the ``$\ell_\infty$-gap problem'', where $\AA$ is either PSD or satisfies $\min_{\BB\succeq 0} \|\AA - \BB\|_2  > \eps n$. For this related problem, we give a $\tilde{O}(1/\eps^2)$ query tester, which we show is optimal up to $\log(1/\eps)$ factors. 
  Both our testers randomly sample a collection of principal submatrices and check whether these submatrices are PSD. Consequentially,  our algorithms achieve \textit{one-sided error}: whenever they output that $\AA$ is not PSD, they return a certificate that $\AA$ has negative eigenvalues.
    
   We complement our upper bound for PSD testing with Euclidean norm distance by giving a $\tilde{\Omega}(1/\eps^2)$ lower bound for any non-adaptive algorithm. 
  Our lower bound construction is general, and can be used to derive lower bounds for a number of spectral testing problems. As an example of the applicability of our construction, we obtain a new $\tilde{\Omega}(1/\eps^4)$ sampling lower bound for testing the Schatten-$1$ norm with a $\eps n^{1.5}$ gap, 
   extending a result of Balcan, Li, Woodruff, and Zhang \cite{BalcanTesting}. In addition, our hard instance results in new sampling lower bounds for estimating the Ky-Fan Norm, and the cost of rank-$k$ approximations, i.e.  $\|\AA - \AA_k\|_F^2 = \sum_{i > k} \sigma_i^2(\AA)$.
\end{abstract}

\thispagestyle{empty}
\newpage
\tableofcontents
\thispagestyle{empty}
\newpage
\pagenumbering{arabic}

\section{Introduction}

Positive Semi-Definite (PSD) matrices are central objects of interest in algorithm design, and continue to be studied extensively in optimization, spectral graph theory, numerical linear algebra, 
statistics, and dynamical systems, among many others~\cite{vandenberghe1996semidefinite,wolkowicz2012handbook,goemans1995improved,arora2009expander, arora2005fast, steurer2010fast,spielman2004nearly,deza2009geometry,wainwright2019high,diakonikolas2019recent,slotine1991applied}. Specifically, a real-valued matrix $\AA \in \R^{n \times n}$ is said to be PSD if it defines a non-negative quadratic form: namely if $x^\top \AA x \geq 0$ for all $x \in \R^n$. If $\AA$ is symmetric, this is equivalent to the eigenvalues of $\AA$ being non-negative. Certifying whether a matrix is PSD often provides crucial insights into the structure of metric spaces~\cite{schoenberg1935remarks}, arises as a separation oracles in Semi-Definite Programming (SDP)~\cite{vandenberghe1996semidefinite}, leads to faster algorithms for solving linear systems and linear algebra problems~\cite{spielman2004nearly,kelner2013simple,musco2017sublinear,bakshi2019robust} detects existence of community structure in random graphs~\cite{saade2014spectral}, and is used to ascertain local convexity of functions. Furthermore, testing if a matrix is PSD is also used when studying the rate of dissipation in the heat equation~\cite{a1993heat} and 
the behavior of non-oscillatory, exponentially stable modes of linear differential equations~\cite{glendinning1994stability}. For these applications, in addition to testing the existence of negative eigenvalues, it is often important to provide a \textit{certificate} that the matrix is not PSD, by exhibiting a direction in which the quadratic form is negative.

While efficient, numerically stable algorithms for computing the spectrum of a matrix have been known since Turing \cite{turing1948rounding}, such algorithms require reading the entire matrix and incur a cubic running time in practice. Computing the eigenvalues of a matrix is often the bottleneck in applications, especially when just determining the existence of negative eigenvalues suffices. For instance, checking embeddability of a finite metric into Euclidean space, feasibility of a SDP, convexity of a function, and if specialized solvers are applicable for linear algebraic problems, all only require knowledge of whether a given matrix is PSD. The focus of this work is to study when the property of being PSD can be tested \textit{sublinear} time and queries, without reading the entire matrix. 

We approach the problem from the perspective of property testing \cite{goldreich1998property,goldreich2017introduction}, where the input matrix $\AA$ is promised to be either a PSD matrix, or ``$\eps$-far'' from PSD under an appropriate notion of distance (discussed below). Specifically, we work in the \textit{bounded-entry model}, proposed by Balcan, Li, Woodruff, and Zhang \cite{BalcanTesting}, where the input matrix has bounded entries: $\|\AA\|_\infty \leq 1$. Boundedness is often a natural assumption in practice, and  has numerous real world applications, such as recommender systems as in the Netflix Challenge \cite{koren2009matrix}, unweighted or bounded weight graphs \cite{goldreich2010introduction,goldreich1998property}, correlation matrices, distance matrices with bounded radius, and others \cite{li2014improved,kannan2016bounded,BalcanTesting}. Further, the boundedness of entries avoids degenerate instances where an arbitrarily large entry is hidden in $\AA$, thereby drastically changing the spectrum of $\AA$, while being impossible to test without reading the entire matrix.

Our starting point is a simple fact: a matrix $\AA$ is PSD if and only if all \textit{principal}\footnote{Recall that a principal submatrix $\AA_{T \times T}$ for $T \subseteq [n]$ is the restriction of $\AA$ to the rows and columns indexed by $T$.}  submatrices of $\AA$ are PSD.  However, a much more interesting direction is: if $\AA$ is not PSD, what can be said about the eigenvalues of the submatrices of $\AA$?  Specifically, if $\AA$ is far from PSD, how large of a submatrix must one sample in order to find a negative eigenvalue? 
Note that given a principal submatrix $\AA_{T \times T}$ with $x^\top \AA_{T \times T} x < 0$ for some $x \in \R^{|T|}$, this direction $x$ can be used as a certificate that the input matrix is not PSD, since $y^\top \AA y = x^\top \AA_{T \times T} x < 0$, where $y$ is the result of padding $x$ with $0$'s.  
Further, it leads us to a natural algorithm to test definiteness: sample multiple principal submatrices and compute their eigenvalues. If any are negative, then $\AA$ must not be PSD. Determining the query complexity of this task is the principal focus of this paper. Specifically, we ask: 
\begin{quote}
\centering 
\emph{Can the positive semi-definiteness of a bounded matrix be tested via the semi-definiteness of a small random submatrix?}
\end{quote}
\paragraph{The Testing Models.} 
The distance from $\AA$ to the PSD cone is given by $\min_{\BB \succeq 0} \|\AA - \BB\|$, where $\| \cdot \|$ is a norm, and $\BB \succeq 0$ denotes that $\BB$ is PSD. To instantiate $\|\cdot \|$, we consider two natural norms over $n \times n$ matrices: the spectral norm ($\| \cdot \|_2$) and the Euclidean norm ($\|\cdot\|_F$). 
Perhaps surprisingly, the distance of a symmetric matrix $\AA$ to the PSD cone under these norms can be characterized in terms of the eigenvalues of $\AA$. In particular, let $\lambda \in \R^n$ be the vector of eigenvalues of $\AA$. Then, the spectral norm distance corresponds to the $\ell_\infty$ distance between $\lambda$ and the positive orthant. Similarly, the squared Frobenius distance corresponds to the $\ell^2_2$ distance between $\lambda$ and the positive orthant.

Therefore, we will refer to the two resulting gap problems as the $\ell_\infty$-gap and the $\ell_2^2$-gap, respectively.  This connection between matrix norms of $\AA$ and vector norms of eigenvalues $\lambda$ will be highly useful for the analysis of random submatrices. Next, we formally define the testing problems:

\begin{problem}[PSD Testing with Spectral norm/$\ell_\infty$-gap]\label{prob:linf}
Given $\eps\in (0,1]$ and a symmetric matrix $\AA \in \R^{n \times n}$ such that $\|\AA \|_{\infty} \leq 1$, distinguish whether $\AA$ satisfies: 
\begin{itemize}
\itemsep0em 
    \item[(1)] $\AA$ is PSD. 
    \item[(2)] $\AA$ is $\eps$-far from the PSD cone in Spectral norm:  $\min_{\BB \succeq 0} \|\AA - \BB\|_2 = \max_{i : \lambda_i < 0} |\lambda_{i}(\AA)| \geq \eps n$.
\end{itemize}
\end{problem} 

The fact that the spectral norm distance from $\AA$ to the PSD cone $(\min_{\BB \succeq 0} \|\AA - \BB\|_2)$ is equivalent to the magnitude of the smallest negative eigenvalue of $\AA$ is a consequence of the variational principle for eigenvalues. For general non-symmetric matrices $\AA$, one can replace $(2)$ above with the condition $x^\top \AA x < - \eps n$ for some unit vector $x \in \R^n$, which is equivalent to $(2)$ if $\AA$ is symmetric (again by the variational principle). 
We note that our results for the $\ell_\infty$-gap hold in this more general setting.\footnote{Also note that given query access to any  $\AA \in \R^{n \times n}$, one can always run a tester on the symmetrization $\BB = (\AA + \AA^\top)/2$, which satisfies $x^\top \AA x =x^\top \BB x$ for all $x$, with at most a factor of $2$ increase in query complexity.}

Next, if we instantiate $\|\cdot\|$ with the (squared) Euclidean norm, we obtain the $\ell_2^2$ gap problem.

\begin{problem}[PSD Testing with $\ell_2^2$-gap]\label{prob:l2}
Given $\eps\in (0,1]$ and a symmetric matrix $\AA \in \R^{n \times n}$ such that $\|\AA \|_{\infty} \leq 1$, distinguish whether $\AA$ satisfies:
\begin{itemize}
\itemsep0em 
    \item[(1)] $\AA$ is PSD.
    \item[(2)] 
    $\AA$ is $\eps$-far from the PSD cone in squared Euclidean norm: 
    \begin{equation}\label{eqn:1}
        \min_{\BB \succeq 0}\|\AA - \BB \|^2_F  = \sum_{i : \lambda_i(\AA) < 0} \lambda_i^2(\AA) \geq \eps n^2
    \end{equation}
\end{itemize}

\end{problem} 
Note that the identity $\min_{\BB \succeq 0}\|\AA - \BB \|^2_F  = \sum_{i : \lambda_i(\AA) < 0} \lambda_i^2(\AA) $ in equation \ref{eqn:1} also follows from the variational principle for eigenvalues (see Appendix \ref{sec:appendixA}). Similarly to the $\ell_\infty$-gap, if $\AA$ is not symmetric one can always run a tester on the symmetrization $(\AA + \AA^\top)/2$.
Also observe that $\|\AA\|_F^2 \leq n^2$ and $\|\AA\|_2 \leq n$ for bounded entries matrices, hence the respective scales of $n,n^2$ in the two gap instances above. 
Notice by definition, if a symmetric matrix $\AA$ is $\eps$-far from PSD in $\ell_\infty$ then $\AA$ is $\eps^2$-far from PSD in $\ell_2^2$. However, the converse is clearly not true, and as we will see the complexity of PSD testing with $\eps^2$-$\ell_2^2$ gap is strictly harder than testing with $\eps$-$\ell_\infty$ gap.\footnote{The difference in scaling of $\eps$ between the $\ell_\infty$ and $\ell_2^2$ gap definitions ($\eps$ is squared in the latter) is chosen for the sake of convenience, as it will become clear the two problems are naturally studied in these respective paramaterizations.} 

In fact, there are several important examples of matrices which are far from the PSD cone in $\ell_2^2$, but which are not far in $\ell_\infty$. For instance, if $\AA$ is a random matrix with bounded moments, such as a matrix with i.i.d.  Rademacher ($\{1,-1\}$) or Gaussian entries, then as a consequence of Wigner's Semicircle Law $\AA$ will be $\Omega(1)$-far in $\ell_2^2$ distance. 
However, $\|\AA\|_2  = O(\sqrt{n})$ with high probability, so $\AA$ will only be $O(1/\sqrt{n})$-far in $\ell_\infty$ distance. Intuitively, such random instances should be very ``far'' from being PSD, and the $\ell_2^2$ distance captures this fact. 
\begin{remark}
A previous version of this work defined the gap in Problem \ref{prob:linf} in full generality (without symmetry assumed) as $x^\top \AA x < - \eps n$ for a unit vector $x$. We have since changed the presentation, as this more general definition does not clearly emphasize the connection between the Problem \ref{prob:linf} and the spectral norm. 
The authors would like to thank an anonymous reviewer for this suggestion. We note that the results themselves remain unaffected. 
\end{remark}

\subsection{Our Contributions}
\label{sec:contri}

We now introduce our main contributions. Our algorithms for PSD testing randomly sample principal submatrices and check if they are PSD. Thus, all our algorithms have one-sided error; when $\AA$ is PSD, they always return \textsf{PSD}, and whenever our algorithms return \textsf{Not PSD}, they output a certificate in the form of a principal submatrix which is not PSD. 
In what follows, $\omega < 2.373$ is the exponent of matrix multiplication, and $\tilde{O},\tilde{\Omega}$ notation only hide $\log(1/\eps)$ factors (and $\log(s)$ factors for Ky-Fan-$s$ and residual error bounds), thus our bounds have no direct dependency on the input size $n$.
We first state our result for the $\ell_\infty$ gap problem in its most general form, which is equivalent to Problem \ref{prob:linf} in the special case when $\AA$ is symmetric.

\smallskip 		\smallskip \smallskip 
\noindent \textbf{Theorem \ref{thm:inftymain} }($\ell_\infty$-gap Upper Bound)\textit{  There is a non-adaptive sampling algorithm which, given $\AA  \in \R^{n \times n} $ with $\|\AA\|_\infty \leq 1$ and $\eps \in (0,1)$, returns \textsf{PSD} if $x^\top \AA x \geq 0$ for all $x \in \R^n$, and with probability $2/3$ returns \textsf{Not PSD} if $x^\top\AA x \leq - \eps n$ for some unit vector $x \in \R^n$. The algorithm make $\tilde{O}(1/\eps^2)$ queries to the entries of $\AA$, and runs in time $\tilde{O}(1/\eps^{\omega})$.\smallskip		\smallskip \smallskip 
	}
	
We demonstrate that the algorithm of Theorem \ref{thm:inftymain} is optimal up to $\log(1/\eps)$ factors, even for adaptive algorithms with two-sided error. Formally, we show:

\smallskip 	\smallskip \smallskip 
\noindent \textbf{Theorem \ref{thm:linftyLB} }($\ell_\infty$-gap Lower Bound)\textit{ Any adaptive or non-adaptive algorithm which solves the PSD testing problem with $\eps$-$\ell_\infty$ gap with probability at least $2/3$, even with two-sided error and if $\AA$ is promised to be symmetric, must query $\wt{\Omega}( 1/\eps^2)$ entries of $\AA$.\smallskip		\smallskip \smallskip 
	}



Next, we present our algorithm for the $\ell_2^2$-gap problem. Our algorithm crucially relies on first running our tester for the $\ell_\infty$-gap problem, which allows us to demonstrate that if $\AA$ is far from PSD in $\ell_2^2$ but close in $\ell_\infty$, then 
	it must be far, under other notions of distance such as Schatten norms or residual tail error, from any PSD matrix. 
	
\smallskip 		\smallskip \smallskip 
\noindent \textbf{Theorem \ref{thm:l2Main} } ($\ell_2^2$-gap Upper Bound)\textit{ 
There is a non-adaptive sampling algorithm which, given a symmetric matrix $\AA \in \R^{n \times n}$ with $\|\AA\|_\infty \leq 1$ and $\eps \in (0,1)$, returns \textsf{PSD} if $\AA$ is PSD, and with probability $2/3$ returns \textsf{Not PSD} if $ \min_{\BB \succeq 0}\|\AA - \BB \|^2_F\geq \eps n^2$. The algorithm make $\tilde{O}(1/\eps^4)$ queries to $\AA$,  and runs in time $\tilde{O}(1/\eps^{2\omega})$.\smallskip		\smallskip \smallskip 
	}

We complement our upper bound by a  $\wt{\Omega}( \frac{1}{\eps^2})$  lower bound for PSD-testing with $\eps$-$\ell_2^2$ gap, which holds even for algorithms with two sided error. Our lower bound demonstrates a separation between the complexity of PSD testing with $\sqrt{\eps}$-$\ell_\infty$ gap and PSD testing with $\eps$-$\ell_2^2$-gap, showing that the concentration of negative mass in large eigenvalues makes PSD testing a strictly easier problem.
	
	\smallskip 		\smallskip \smallskip 
\noindent \textbf{Theorem \ref{thm:lbmain} }($\ell_2^2$-gap Lower Bound)  \textit{ 
Any non-adaptive algorithm which solves the PSD testing problem with $\eps$-$\ell_2^2$ gap with probability at least $2/3$, even with two-sided error, must query $\wt{\Omega}( 1/\eps^2)$ entries of $\AA$. \smallskip		\smallskip \smallskip 
	}

	Our lower bound is built on discrete hard instances which are ``locally indistinguishable'', in the sense that the distribution of any small set of samples is completely identical between the PSD and $\eps$-far cases.
	At the heart of the lower bound is a key combinatorial Lemma about arrangements of paths on cycle graphs (see discussion in Section \ref{sec:techlb}). 
	Our construction is highly general, and we believe will likely be useful for proving other lower bounds for matrix and graph property testing problems. Exemplifying the applicability of our construction, we obtain as an immediate corollary a new lower bound for testing the Schatten-$1$ norm of $\AA$. Recall, that the Schatten-$1$ norm is defined via $\|\AA\|_{\mathcal{S}_1} = \sum_i \sigma_i(\AA)$, where $\sigma_1(\AA) \geq \dots \geq \sigma_n(\AA)$ are the singular values of $\AA$. 
	
			\smallskip 	\smallskip \smallskip 
\noindent \textbf{Theorem \ref{thm:schattenlb} } (Schatten-$1$ Lower Bound)  \textit{ Fix any $1/\sqrt{n} \leq \eps \leq 1$.
Then any non-adaptive algorithm in the bounded entry model that distinguishes between 
\begin{enumerate}
    \item $\|\AA\|_{\mathcal{S}_1} \geq \eps n^{1.5}$,
    \item $\|\AA\|_{\mathcal{S}_1} \leq (1-\eps_0)\eps n^{1.5}$
\end{enumerate}
with probability $2/3$, where $\eps_0 = 1/\log^{O(1)}(1/\eps)$, must make at least $\tilde{\Omega}( 1/\eps^4)$ queries to $\AA$.
	}			\smallskip 		\smallskip \smallskip 
	
Note that one always has $\|\AA\|_{\mathcal{S}_1}  \leq n^{1.5}$ in the bounded entry model ($\|\AA\|_\infty \leq 1$), which accounts for the above scaling.
		Theorem \ref{thm:schattenlb} extends a lower bound of Balcan et. al. \cite{BalcanTesting}, which is $\Omega(n)$ for the special case of $\eps ,\eps_0= \Theta(1)$. Thus, for the range $\eps = \tilde{O}(n^{-1/4})$, our lower bound is an improvement. To the best of our 
			knowledge, Theorem \ref{thm:schattenlb} gives the first $\tilde{\Omega}(n^2)$ sampling lower bound for testing Schatten-$1$ in a non-degenerate range (i.e., for $\|\AA\|_{\mathcal{S}_1} > n$). 
			\begin{remark}
			We note that the lower bound of \cite{BalcanTesting} is stated for a slightly different version of gap (a ``$\eps$-$\ell_0$''-gap), where either $\|\AA\|_{\mathcal{S}_1} \geq c_1 n^{1.5}$ for a constant $c_1$, or at least $\eps n^2$ of the entries of $\AA$ must be changed (while respecting $\|\AA\|_\infty \leq 1$) so that the Schatten-$1$ is larger than $c_1 n^{1.5}$. However, their lower bound construction itself satisfies the ``Schatten-gap'' version as stated in Theorem \ref{thm:schattenlb}, where either  $\|\AA\|_{\mathcal{S}_1} \geq c_1 n^{1.5}$, or $\|\AA\|_{\mathcal{S}_1} \leq c_2 n^{1.5}$ and $c_1 > c_2$ are constants. From here, it is easy to see that this gap actually implies the $\ell_0$-gap (and this is used to obtain the $\ell_0$-gap lower bound in \cite{BalcanTesting}), since if  $\|\AA\|_{\mathcal{S}_1} \leq c_2 n^{1.5}$ then for any $\EE$ with $\|\EE\|_\infty \leq 2$ and $\|\EE\|_0 \leq  \epsilon n^2$ for a small enough constant $\epsilon < c_2^2$, we have $\|\AA + \EE \|_{\mathcal{S}_1} \leq \|\AA\|_{\mathcal{S}_1} + \|\EE\|_{\mathcal{S}_1}  \leq n^{1.5} (c_2 + 2 \sqrt{\epsilon} ) < c_1 n^{1.5}$. So Theorem \ref{thm:schattenlb} implies a lower bound of $\tilde{\Omega}(1/\eps^2)$ for distinguishing $\|\AA\|_{\mathcal{S}_1} \geq \sqrt{\eps} n^{1.5}$ from the case of needing to change at least $\tilde{\Omega}(\eps n^2)$ entries of $\AA$ so that $\|\AA\|_{\mathcal{S}_1} \geq \sqrt{\eps} n^{1.5}$. Thus, our lower bound \textit{also} extends the $\ell_0$-gap version of the results of \cite{BalcanTesting} for the range $\eps = \tilde{O}(1/\sqrt{n})$.
			\end{remark}
			
			In addition to Schatten-$1$ testing, the same lower bound construction and techniques from Theorem \ref{thm:lbmain} also result in new lower bounds for testing the Ky-Fan $s$ norm $\|\AA\|_{\KF(s)} = \sum_{i=1}^s \sigma_i(\AA)$, as well as the cost of the best rank-$s$ approximation $\|\AA - \AA_s \|_{F}^2 = \sum_{i> s} \sigma_i^2(\AA)$, stated below. In the following, $s$ is any value $1 \leq s \leq n/(\poly \log n)$, and $c$ is a fixed constant. 
			
				\smallskip \smallskip \smallskip 
\noindent \textbf{Theorem \ref{thm:kflb} }(Ky-Fan Lower Bound)   \textit{ 
Any non-adaptive algorithm in the bounded entry model which distinguishes between 
\begin{enumerate}
    \item $\|\AA\|_{\KF(s)} >  \frac{c}{\log s}n$
    \item $\|\AA\|_{\KF(s)} < (1-\eps_0) \cdot \frac{c}{\log s}n $
\end{enumerate}
with probability $2/3$, where $\eps_0 = \Theta(1/\log^2(s))$, must query at least $\tilde{\Omega}(s^2)$ entries of $\AA$.
}\smallskip \smallskip \smallskip 

	\smallskip \smallskip \smallskip 
\noindent \textbf{Theorem \ref{thm:taillb} }(Residual Error Lower Bound) \textit{ 
  Any non-adaptive algorithm in the bounded entry model which distinguishes between  
  \begin{enumerate}
      \item $\|\AA - \AA_s \|_{F}^2 >   \frac{c}{s \log s}n $
      \item $\|\AA - \AA_s \|_{F}^2 <  (1-\eps_0) \cdot \frac{c}{s \log s}n$
  \end{enumerate}
with probability $2/3$,  where $\eps_0 = 1/\log^{O(1)}(s)$,  must query at least $\tilde{\Omega}(s^2)$ entries of $\AA$.
 }\smallskip \smallskip \smallskip 
 

 Our lower bound for the Ky-Fan norm complements a Ky-Fan testing lower bound of \cite{li2016tight}, which is $\Omega(n^2/s^2)$ for distinguishing \textbf{1)} $\|\AA\|_{KF(s)} < 2.1 s \sqrt{n}$ from \textbf{1)} $\|\AA\|_{KF(s)} > 2.4s  \sqrt{n}$ when $s = O(\sqrt{n})$. Note their bound decreases with $s$, whereas ours increases, thus the two bounds are incomparable (although they match up to $\log(s)$ factors at $s = \Theta(\sqrt{n})$).\footnote{The bound from \cite{li2016tight} is stated in the sketching model, however the entries of the instance are bounded, thus it also applies to the sampling model considered here.}
 We also point out that there are (not quite matching) upper bounds for both the problems of Ky-Fan norm  and $s$-residual error testing in the bounded entry model, just based on a standard application of the Matrix Bernstein Inequality.\footnote{See Theorem 6.1.1 of \cite{tropp2015introduction}, applied to $S_k = a_{(k)}(a_{(k)})^\top$, where $a_{(k)}$ is the $k$-th row sampled in $\AA$; for the case of residual error, one equivalently applies matrix Bernstein inequality to estimate the head $\sum_{i \leq k} \sigma_i^2(\AA)$. These bounds can be tightened via the usage of interior Chernoff bounds \cite{gittens2011tail}.} 
 We leave the exact query complexity of these and related testing problems for functions of singular values in the bounded entry model as subject for future work.
			
			

	



\paragraph{A Remark on the $\ell_2^2$-Gap.} We note that
there appear to be several key barriers to improving the query complexity of PSD testing with $\ell_2^2$-gap beyond $O(1/\eps^4)$, which we briefly discuss here. 
First, in general, to preserve functions of the squared singular values of $\AA$ up to error $\eps n^2$, such as $\|\AA\|_F^2 = \sum_i \sigma_i^2(\AA)$ or $\|\AA\|_2^2 = \sigma_1^2(\AA)$, any algorithm which samples a submatrix must make  $\Omega(1/\eps^4)$ queries (see Lemma \ref{lem:lowerboundAA} for estimating $\sum_{i \leq k} \sigma_i^2$ for any $k$). In other words, detecting $\eps n^2$-sized perturbations in the spectrum of a matrix in general requires $\Omega(1/\eps^4)$ sized submatrices. This rules out improving the query complexity by detecting the $\eps n^2$ negative mass in $\AA$ via, for instance, testing if the sum of squares of top $k=1/\eps$ singular values has $\Theta(\eps n^2)$ less mass than it should if $\AA$ were PSD (even this may require $\Omega(k^2/\eps^4)$ queries, see the discussion in Section \ref{sec:techl2}).

The key issue at play in the above barrier appears to be the requirement of sampling submatrices.  Indeed, notice for the simplest case of $\|\AA\|_F^2 $, we can easily estimate $\|\AA\|_F^2$ to additive $\eps n^2$ via $O(1/\eps^2)$ queries to random entries of $\AA$. On the other hand, if these queries must form a submatrix, then it is easy to see that $\Omega(1/\eps^4)$ queries are necessary, simply from the problem of estimating $\|\AA\|_F^2$ whose rows (or columns) have values determined by a coin flip with bias either equal to $1/2$ or $1/2 + \eps$. On the other hand, for testing positive semi-definiteness, especially with one-sided error, the requirement of sampling a principal submatrix seems unavoidable.

 
In addition, a typical approach when studying spectral properties of submatrices is to first pass to a random row submatrix $\AA_{S \times [n]}$, argue that it preserves the desired property (up to scaling), and then iterate the process on a column submatrix $\AA_{S \times T}$. Unfortunately, these types of arguments are not appropriate when dealing with eigenvalues of $\AA$, since after passing to the rectangular matrix $\AA_{S \times [n]}$, any notion of negativity of the eigenvalues has now been lost. This forces one to argue indirectly about functions of the singular values of $\AA_{S \times [n]}$, returning to the original difficulty described above.
We leave it as an open problem to determine the exact non-adaptive query complexity of PSD testing with $\ell_2^2$-gap. For a further discussion of these barriers and open problems, see Section \ref{sec:conclusion}.





\subsection{Connections to Optimization, Euclidean Metrics and Linear Algebra}

We now describe some explicit instances where our algorithms may be useful for testing positive semi-definiteness. We emphasize that in general, the distance between $\AA$ and the PSD cone may be too small to verify via our testers. However, when the input matrices satisfy a non-trivial gap from the PSD cone, we can speed up some basic algorithmic primitives.
The first is testing feasibility of the PSD constraint in a Semi-Definite Program (SDP) with sublinear queries and time, so long as the variable matrix has bounded entries. Importantly, our algorithms also output a separating hyperplane to the PSD cone.
\begin{corollary}[Feasibility and Separating Hyperplanes for SDPs]
Given a SDP $\mathcal{S}$, let $\X \in \R^{n \times n}$ be a symmetric matrix that violates the PSD constraint for $\mathcal{S}$. Further, suppose $\|\X \|_{\infty} \leq 1$ and $\X$ is $\eps n^2$-far in entry-wise $\ell_2^2$ distance to the PSD cone. Then, there exists an algorithm that queries $\widetilde{O}(1/\epsilon^4)$ entries in $\X$ and runs in $\wt{O}(1/\epsilon^{2\omega})$  time, and with probability $9/10$, outputs a vector $\tilde{v}$ such that $\tilde{v}^{T} \X \tilde{v} < 0$. 
Moreover, if $\lambda_{\min}(\X) <  -\epsilon n$, then there is an algorithm yielding the same guarantee, that queries $\widetilde{O}(1/\epsilon^2)$ entries in $\X$ and runs in $\wt{O}(1/\epsilon^{\omega})$  time. 
\end{corollary}

While in the worst case, our algorithm may need to read the whole matrix to exactly test if $\X$ is PSD, there may be applications where relaxing the PSD constraint to the convex set of matrices which are close to the PSD cone in Euclidean distance is acceptable. 
Moreover, our algorithm may be run as a preliminary step at each iteration of an SDP solver to check if the PSD constraint is badly violated, resulting in speed-ups by avoiding an expensive eigendecomposition of $\X$ whenever our algorithm outputs a separating hyperplane~\cite{vandenberghe1996semidefinite}.

Next, we consider the problem of testing whether an arbitrary finite metric $d$ over $n$ points, $x_1, \dots x_n \in \mathbb{R}^d$ is embeddable into Euclidean Space. Testing if a metric is Euclidean has a myriad of applications, such as determining whether dimensionality reduction techniques such as  Johnson-Lindenstrauss can be used \cite{parnas2003testing}, checking if efficient Euclidean TSP solvers can be applied \cite{arora1998polynomial}, and more recently, computing a low-rank approximation in sublinear time \cite{bakshi2018sublinear,indyk2019sample}. It is well known (Schoenberg's criterion~\cite{dattorro2010convex}) that given a distance matrix $\DD \in \R^{n \times n}$ such that $\DD_{i,j} = d(x_i, x_j)$, the points are isometrically embeddable into Euclidean space if and only If $\G = \mathbf{1} \cdot \DD_{1,*} + \DD_{1,*}^\top \cdot \mathbf{1}^\top - \DD  \succeq 0$, where $ \DD_{1,*}$ is the first row of $\DD$. Notice that embeddability is scale invariant, allowing one to scale distances to ensure boundedness. Furthermore, since our algorithms sample submatrices and check for non-positive semi-definiteness, the tester need not know this scaling in advance, and gives guarantees for distinguishing definiteness if the necessary gap is satisfied after hypothetically scaling the entries.

\begin{corollary}[Euclidean Embeddability of Finite Metrics]
Given a finite metric $d$ on $n$ points $\{x_1,\dots,x_n\}$, let $\DD \in \R^{n \times n}$ be the corresponding distance matrix, scaled so that $\|\DD\|_\infty \leq 1/3$, and let $\G= \mathbf{1}\DD_{1,*} + \DD_{1,*}^\top\mathbf{1}^\top - \DD$. Then if $\min_{\BB \succeq 0} \|\G - \BB\|_F^2 \geq \eps n^2$, there exists an algorithm that queries $\wt{O}(1/\epsilon^4)$ entries in $\AA$ and with probability $9/10$, determines the non-embeddability  of $\{x_1,\dots,x_n\}$ into Euclidean space.  Further, the algorithm runs in time $\wt{O}(1/\epsilon^{2\omega})$.
\end{corollary}
\begin{remark}
An intriguing question is to characterized geometric properties of finite metrics based on the $\ell_2^2$-distance of the Schoenberg matrix $\G$ from the PSD cone. For instance, given a finite metric with Schoenberg matrix $\G$ that is close to being PSD in $\ell^2_2$-distance, can we conclude that the metric has a low worst or average case distortion embedding into Euclidean space?
\end{remark}
\begin{remark}
Since rescaling entries to be bounded only affects the gap parameter $\eps$, in both of the above cases, so long as the magnitude of the entries in $\X,\DD$ do not scale with $n$, the running time of our algorithms is still sublinear in the input. 
\end{remark}

Finally, several recent works have focused on obtaining sublinear time algorithms for low-rank approximation when the input matrix is PSD \cite{musco2017sublinear, bakshi2019robust}. However, such algorithms only succeed when the input is PSD or close to PSD (in $\ell_2^2$), and it is unknown how to verify whether these algorithm succeeded in sublinear time. Therefore, our tester can be used as a pre-processing step to determine input instances where the aforementioned algorithms provably will (or will not) succeed.

\subsection{Related work}
Property testing in the bounded entry model was first considered in \cite{BalcanTesting} to study the query complexity of testing spectral properties of matrices, such as stable rank (the value $\|\AA\|_F^2/\|\AA\|_2^2)$ and Schatten $p$ norms. A related model, known as the \textit{bounded row model}, where rows instead of entries are required to be bounded, was studied by Li, Wang, and Woodruff \cite{li2014improved}, who gave tight bounds for testing stable rank in this model. In addition, the problem of testing the rank of a matrix from a small number of queries has been well studied \cite{parnas2003testing,krauthgamer2003property,li2014improved,barman2016testing}, as well the problem of estimating the rank via a random submatrix \cite{balcan2011learning,balcan2016noise}. Notice that since rank is not a smooth spectral property, hiding an unbounded value in a single entry of $\AA$ cannot drastically alter the rank. Thus, for testing rank, the condition of boundedness is not required.

More generally, the bounded entry model is the natural sampling analogue for the \textit{linear sketching model}, where the algorithm gets to choose a matrix $\S \in \R^{t \times n^2}$, where $t$ is the number of ``queries'', and then observes the product $\S \cdot \textsf{vec}(\AA)$, where $\textsf{vec}(\AA)$ is the
vectorization of $\AA$~\cite{li2017embeddings,li2016tight, braverman2018matrix,li2016approximating,li2014improved,li2014sketching,braverman2019schatten,li2019approximating}. The model has important applications to streaming and distributed algorithms. Understanding the query complexity of sketching problems, such as estimating spectral norms and the top singular values \cite{andoni2013eigenvalues,li2014sketching,li2016tight}, estimating Schatten and Ky-Fan norms \cite{li2016tight, li2017embeddings,li2016approximating, braverman2019schatten}, estimating $\ell_p$ norms \cite{alon1996space, indyk2006stable, kane2010exact,  jayaram2019towards,ben2020framework}, and $\ell_p$ sampling \cite{monemizadeh20101, Jowhari:2011,jayaram2018perfect,jayaram2019weighted}, has been a topic of intense study. For the problem of sketching \emph{eigenvalues} (with their signs), perhaps the most related result is \cite{andoni2013eigenvalues}, which gives point-wise estimates of the top eigenvalues. Notice that linear sketching can simulate sampling by setting the rows of $\S$ to be standard basis vectors, however sketching is in general a much stronger query model. 
Note that to apply a linear sketch, unlike in sampling, one must read all the entries of $\AA$, which does not yield sublinear algorithms.

A special case of the sketching model is the \textit{matrix-vector product} model, which has been studied extensively in the context of compressed sensing~\cite{candes2006stable, eldar2012compressed} and sparse recovery~\cite{gilbert2010sparse}. Here, one chooses vectors $v_1,\dots,v_k$ and observes the products $\AA v_1,\dots,\AA v_k$. Like sketching, matrix-vector product queries are a much stronger access model than sampling. Recently, in the matrix-vector product model, Sun et. al. considered testing various graph and matrix properties~\cite{sun2019querying}, 
and Han et. al. considered approximating spectral sums and testing positive semi-definiteness~\cite{han2017approximating}. 

Lastly, while there has been considerable work on understanding concentration of norms and singular values of random matrices, not as much is known about their \textit{eigenvalues}. 
Progress in understanding the behavior of singular values of random matrices includes concentration bounds for spectral norms of submatrices \cite{rudelson2007sampling,tropp2008norms}, concentration bounds for extreme singular values \cite{gittens2011tail,tropp2015introduction,vershynin2010introduction,garg2018matrix,kyng2018matrix}, non-commutative Khintchine inequalities for Schatten-$p$ norms \cite{lust1991non,pisier2009remarks,pisier2017non}, as well as Kadison-Singer type discrepancy bounds \cite{marcus2015interlacing,kyng2020four,song2020hyperbolic}. 
These random matrix concentration bounds have resulted in improved algorithms for many fundamental problems, such as low-rank approximation and regression \cite{clarkson2017low,  musco2017sublinear, bakshi2018sublinear, indyk2019sample, diao2019optimal} and spectral sparsification
\cite{spielman2011spectral,batson2013spectral,andoni2016sketching}. However, in general, understanding behavior of negative eigenvalues of random matrices and submatrices remains largely an open problem. 

\subsection{Technical Overview} 
In this section, we describe the techniques used in our non-adaptive testing algorithms for the $\ell_\infty$ and more general $\ell^2_2$ gap problem, as well as the techniques involved in our lower bound construction for the $\ell^2_2$-gap. 

\subsubsection{PSD Testing with  \texorpdfstring{$\ell_\infty$}{L-inf}  Gap}\label{sec:techlinf}

Recall in the 
general statement of the $\ell_\infty$-gap problem, our task is to distinguish between $\AA \in \R^{n \times n}$ satisfying $x^\top \AA x \geq 0$ for all $ x \in \R^n$, or $x^\top \AA x \leq - \eps n$ for some unit vector $x \in \R^n$. 
Since if  $x^\top \AA x \geq 0$ for all $ x \in \R^n$ the same holds true for all principal submatrices of $\AA$, it suffices to show that in the $\eps$-far case we can find a $k \times k$ principal submatrix $\AA_{T \times T}$ such that $y^\top \AA_{T \times T} y < 0$ for some $y \in \R^{k}$.\footnote{This can be efficiently checked by computing the eigenvalues of $\AA_{T \times T} + \AA_{T \times T}^\top$.} 

\paragraph{Warmup: A $O(1/\eps^3)$ query algorithm.} 
Since we know $x^\top \AA x \leq - \eps n$ for some fixed $x$, one natural approach would be to show that the quadratic form with the \textit{same} vector $x$, projected onto to a random subset $T \subset [n]$ of its coordinates, is still negative. Specifically, we would like to show that the quadratic form $\mathcal{Q}_T(x) = x^\top_T\AA_{T \times T} x_T$, of $x$ with a random principal submatrix $\AA_{T \times T}$ for $T \subset [n]$ will continue to be negative. If $\mathcal{Q}_T(x) < 0$, then clearly $\AA_{T \times T}$ is not PSD. Now while our algorithm does not know the target vector $x$, we can still analyze the concentration of the scalar random variable $\mathcal{Q}_T(x)$ over the choice of $T$, and show that it is negative with good probability. 


\smallskip \smallskip \smallskip 
\noindent \textbf{Proposition \ref{prop:exp} and Lemma \ref{lem:var} (informal)}  \textit{ 
	Suppose $\AA \in \R^{n \times n}$ satisfies $\|\AA\|_\infty \leq 1$ and $x^\top \AA x \leq - \eps n$  where $\|x\|_2 \leq 1$. Then if $k \geq 6/\eps$, and if $T \subset [n]$ is a random sample of expected size $k$, we have  $\ex{\mathcal{Q}_T(x)} \leq - \frac{\eps k^2}{4n}$ and $\text{Var}(\mathcal{Q}_T(x)) \leq O(\frac{k^3}{n^2})$.
	}
\smallskip \smallskip \smallskip

	
By the above Proposition, after setting $k = \Theta(1/\eps^2)$, we have that $|\ex{\mathcal{Q}_T(x)}|^2  = \Omega( \text{Var}(\mathcal{Q}_T(x))$, and so by Chebyshev's inequality, with constant probability we will have $\mathcal{Q}_T(x)< 0$. This results in a $k^2 = O(1/\eps^4)$ query tester. To improve the complexity, we could instead set $k=\Theta(1/\eps)$ and re-sample $T$ for $k$ times independently to reduce the variance. Namely, one can sample submatrices $T_1,T_2,\dots,T_k$, and analyze  $\frac{1}{k}\sum_{i=1}^k\mathcal{Q}_{T_i}(x)$. The variance of this sum goes down to $O(\frac{k^2}{n^2})$, so, again by Chebyshev's inequality, the average of these quadratic forms will be negative with constant probability. If this occurs, then at least one of the quadratic forms must be negative, from which we can conclude that at least one of $\AA_{T_i \times T_i}$ will fail to be PSD, now using only $O(1/\eps^3)$ queries.

	\paragraph{A Family of Hard Instances}
	One could now hope for an even tighter analysis of the concentration of $\mathcal{Q}_T(x)$, so that $O(1/\eps^2)$ total queries would be sufficient. Unfortunately, the situation is not so simple, and in fact the two aforementioned testers are tight in the query complexity for the matrix dimensions they sample. Consider the hard instance $\AA$ in the left of Figure \ref{fig:Matrix1}, which is equal to the identity on the diagonal, and is zero elsewhere except for a small subset $S \subset [n]$ of $|S| = \eps^2 n$ rows and columns, where we have $\AA_{S \times \overline{S}}=\AA_{\overline{S}\times S} = - \mathbf{1}$, where $\overline{S}$ is the complement of $S$. Notice that if we set $x_i^2 = 1/(2n)$ for $i \notin S$ and $x_i^2 = 1/(2\eps^2 n)$ for $i \in S$, then $x^\top \AA x \leq - \eps n/4$. However, in order to even see a single entry from $S$, one must sample from at least $\Omega(1/\eps^2)$ rows or columns. In fact, this instance itself gives rise to a $\Omega(1/\eps^2)$ lower bound for any testing algorithm, even for adaptive algorithms (Theorem \ref{thm:linftyLB}).

\begin{figure}
\centering
\includegraphics[scale = .17 ]{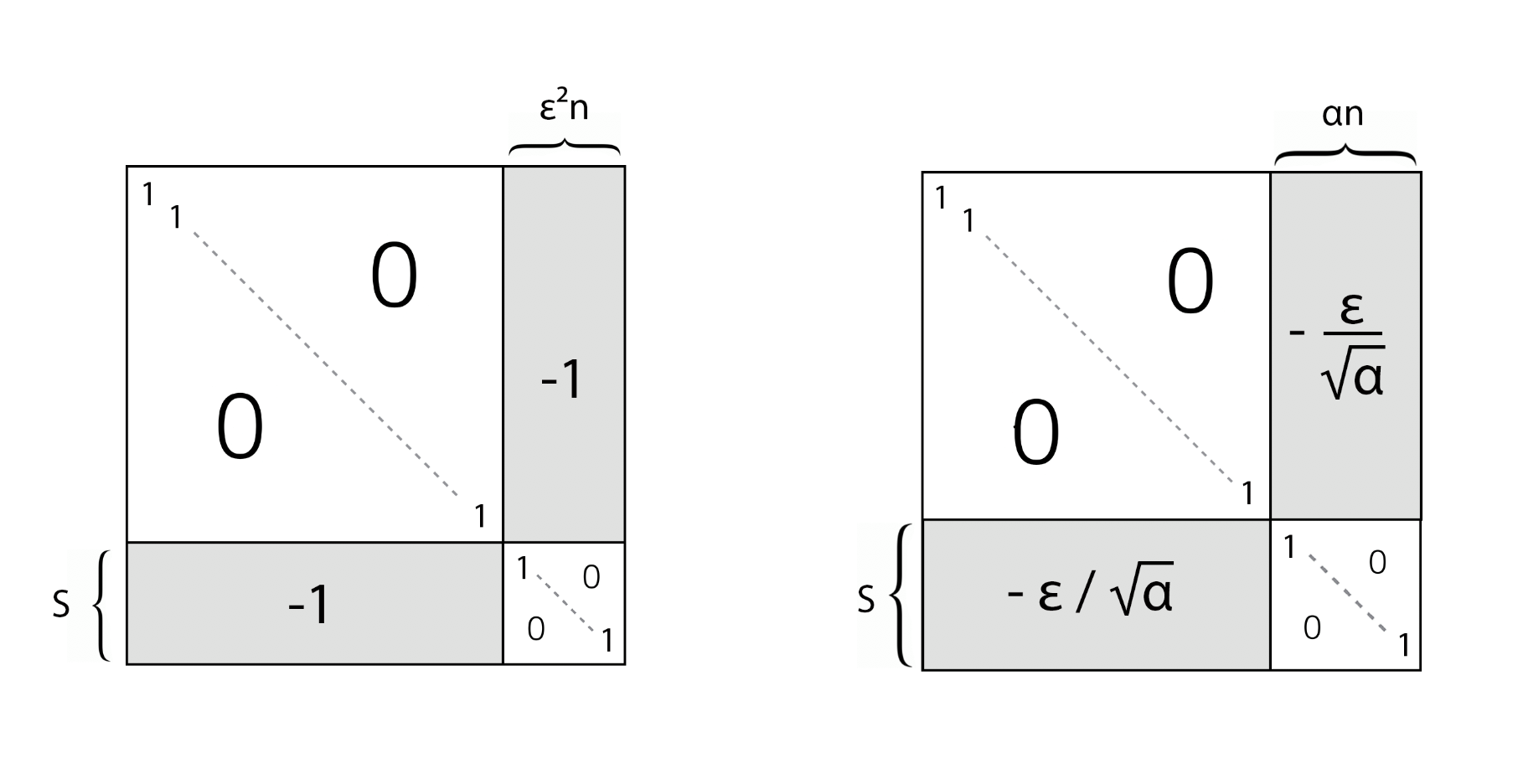}

	\caption{Hard instances for $\ell_\infty$ testing. On the left, the negative mass is highly concentrated in $|S| = \eps^2 n$ rows and columns, and on the right it more spread out over $|S| = \alpha n$, where $\eps^2 \leq \alpha \leq \eps$. }
	\label{fig:Matrix1}
	
\end{figure}
	
	The difficulty of the above instance is that the negative mass of $x^\top \AA x$ is hidden in only a $\eps^2$-fraction of $\AA$. On the other hand, since the negative entries are so large and concentrated, one need only  sample $O(1)$ entries from a single row $i \in S$ in order for $\AA_{T \times T}$ to be non-PSD in the prior example. Thus, an algorithm for such instances would be to sample $O(1/\eps^2)$ principal submatrices, each of \textit{constant} size. 
	On the other hand, the set $S$ could also be more spread out; namely, we could have $|S| = \alpha n$ for any $\eps^2 \leq \alpha \leq \eps$,	but where each entry in $\AA_{S \times \overline{S}}$ is set to $-\eps/\sqrt{\alpha}$ (see the matrix in the right side of Figure \ref{fig:Matrix1}). If instead, we define $x_i^2 = 1/(2\alpha n)$ for $i \in S$, we still have $x^\top \AA x < - \eps n/4$. However, now any submatrix $\AA_{T \times T}$ with $|T \cap S| = 1$ must have at least $|T| \geq \alpha/\eps^2$ rows and columns, otherwise $\AA_{T \times T}$ would be PSD due to the identity on the diagonal.

	

	The aforementioned instances suggest the following approach: query matrices at $O(\log \frac{1}{\eps} )$ different scales of subsampling. Specifically, for each $\eps^2 \leq \alpha = 2^i \leq \eps$, we sample $\tilde{O}(\frac{\eps^2}{\alpha^2})$ independent $k \times k$ submatrices, each of size $k = \tilde{O}(\alpha/\eps^2)$, giving a total complexity of $\tilde{O}(\frac{1}{\eps^2})$. The analysis now proceeds by a complete characterization of the ways in which $x^\top \AA x$ can be negative. Specifically, we prove the following: either a substantial fraction of the negative mass is hidden inside of a small set of rows and columns $S$ with $|S| < \eps n$, or it is the case that $\text{Var}(\mathcal{Q}_T(x))$ is small enough so that a single $k \times k$ submatrix will already be non-PSD with good probability when $k \gtrsim 1/\eps$. 
	Given this classification, it suffices to demonstrate a level of subsampling which will find a non-PSD submatrix when the negative mass is concentrated inside inside a small set $S$.


	\paragraph{Eigenvector Switching.}
	
To analyze this case, ideally, one would like to demonstrate that conditioned on $T$ intersecting $S$ at some level of subsampling, we will have $\mathcal{Q}_T(x) < 0$ with good probability. Unfortunately, the approach of analyzing the quadratic form with respects to $x$ will no longer be possible; in fact, $\mathcal{Q}_T(x)$ may never be negative conditioned on $|T \cap S|  = 1$ (unless $|T| > 1/\eps$, which we cannot afford in this case). The complication arises from the fact that the coordinates of $x_i$ in the small set $S$ can be extremely large, and thus the diagonal contribution of $x_i^2 \AA_{i,i}$ will dominate the quadratic form of a small submatrix. For instance, if $\AA_{T \times T}$ is a sample with $k =|T| = O(1)$ which intersects the set $S$ in the leftmost matrix in Figure \ref{fig:Matrix1}, where $x_i = 1/(\eps \sqrt{n})$ for $i \in S$ and $x_i = 1/ \sqrt{n}$ otherwise, then $\mathcal{Q}_T(x) \approx k/n -(k/\sqrt{n}) x_i  + \AA_{i,i} x_i^2$, which is dominated by the diagonal term $\AA_{i,i} x_i^2 = 1 /(\eps^2n)$. Thus, while $\AA_{T \times T}$ itself is not PSD, we have that $\mathcal{Q}_T(x)> 0$.


	To handle this, we must and analyze the quadratic form $\mathcal{Q}_T(\cdot)$ with respect to \textit{another} direction $y$. The vector $y$ may not even satisfy $y^\top \AA y < 0$, however conditioned on $|T \cap S| \geq 1$, we will have $\mathcal{Q}_T(y) < 0$ with good probability. Clearly, we must scale down the large coordinates $x_i$ for $i \in S$. However, one cannot scale too low, otherwise the negative contribution of the rows $i \in S$ would become too small.
	The correct scaling is then a careful balancing act between the contributions of the different portions of $\AA_{T \times T}$. Informally, since the $x_i$'s for $i \in S$ make up a $|S|/n$ fraction of all coordinates, they can be as large as $x_i^2 \geq( n/|S| )\cdot(1/n)$. However, inside of the smaller submatrix $\AA_{T \times T}$, then conditioned on $i \in T$, since $|T|$ is small $x_i$ now makes up a larger $1/|T|$ fraction of the submatrix, thus we should scale down $x_i$ to only be $x_i^2 \approx |T|/n $.  With this scaling in mind, we (very roughly) set $y_i^2 = (|S|/n) \cdot (|T|/n)$ if $i \in S$, and set $y_i = x_i$ otherwise. The remaining argument then requires a careful analysis of the contribution of entries of $\AA$ outside of $S$ to show that the target vector $y$ indeed satisfies $\mathcal{Q}_T(y) <0$ with good probability conditioned on $T$ intersecting $S$.

\subsubsection{PSD Testing with \texorpdfstring{$\ell_2$}{L-2}  Gap}\label{sec:techl2}
Recall in the $\ell_2$ gap problem, our task is to distinguish between $\AA$ being PSD, and  $\AA$ being $\eps$-far in $\ell_2^2$ distance from any PSD matrix, namely that $\sum_{i : \lambda_i(\AA) < 0} \lambda_i^2(\AA) > \eps n^2$. In what follows, we refer to the quantity $\sum_{i : \lambda_i(\AA) < 0} \lambda_i^2(\AA)$ as the \textit{negative mass} of $\AA$. First observe that in the special case that we had a ``large'' negative eigenvalue, say $\lambda_{\min}(\AA) = - \sqrt{\eps} n$, then by applying our testing algorithm for $\ell_\infty$-gap, we could find a non-PSD submatrix with only $\tilde{O}(1/\eps)$ queries. However, in general the negative mass of $\AA$ may be spread out over many smaller eigenvalues. 
Thus, we cannot hope to apply our earlier approach for the $\ell_\infty$-gap, which preserved the quadratic form $\mathcal{Q}_T(x) = x^\top_T \AA_{T \times T} x_T$ with respects to a fixed direction $x$. Instead, our approach will be to show that if $\AA$ is $\eps$-far from PSD in $\ell_2^2$, then the singular values of $\AA$ must be ``far'' from PSD, in some other notion of distance, allowing us to indirectly infer the existence of negative eigenvalues in submatrices. 


\paragraph{PSD matrices are top-heavy, and a reduction to estimating the tail.} Our first step is to show that if $\AA \in \R^{n \times n}$ is PSD, then the $t$-``tail'' of $\AA$, defined as $\sum_{i > t} \sigma_i^2(\AA)$, cannot be too large. This can be derived from the following fact: if $\AA$ is PSD then we can bound the Schatten-$1$ norm of $\AA$ by $ \|\AA\|_{\mathcal{S}_1} = \sum_i \sigma_i(\AA) = \text{Tr}(\AA)$, which is at most $n$ if $\|\AA\|_\infty \leq 1$. This simple fact will prove highly useful, since whenever we can demonstrate that the Schatten-$1$ norm of a submatrix $\AA_{T \times T}$ is larger than $|T|$, we may immediately conclude the that $\AA_{T \times T}$ is not PSD. 
In addition, it implies:

\smallskip
\noindent \textbf{Proposition \ref{prop:topheavy}} (PSD matrices are top-heavy) \textit{
Fix any $n  \in \mathbb{N}$, $1 \leq t \leq n$, and $\DD \in \R^{n \times n}$. Then if $\DD$ is PSD, we have
	\[\sum_{i > t} \sigma_i(\DD)^2 \leq \frac{1}{t}\left( \Tr(\DD)\right)^2\] 
In particular, if $\DD$ has bounded entries $\|\DD\|_\infty \leq 1$, we have $\sum_{i > t} \sigma_i(\DD)^2 \leq \frac{1}{t}n^2$.
}\smallskip

On the other hand, suppose that $\AA$ is $\eps$-far from PSD, and let $t > 10 / \eps$. Then if no eigenvalue is smaller than $- \eps n/100$, a condition which can be checked with $\tilde{O}(1/\eps^2)$ queries by first running our $\ell_\infty$-gap tester, then the negative mass must be spread out, and it must be the case that a substantial fraction of the negative mass of $\AA$ is contained in the bottom $n-t$ singular values. Specifically, we must have $\sum_{i > t} \sigma_i(\AA)^2 > (\eps/2) n^2$, whereas any PSD matrix $\DD$ would have to satisfy  $\sum_{i > t} \sigma_i^2(\DD) \leq (\eps/10) n^2$ by the above Proposition. Thus, after first running our $\ell_\infty$ tester, it will suffices to estimate the tail  $\sum_{i > t} \sigma_i^2(\AA)$. Equivelantly, since  $\|\AA\|_F^2 = \sum_{i  } \sigma_i^2(\AA)$ can be efficiently estimated, it also suffices to estimate the ``head'' $\sum_{i \leq t} \sigma_i^2(\AA)$  to additive $O(\eps n^2)$.

In order to accomplish this, one could utilize the tools from random matrix concentration, such as Matrix Bernstein's inequality \cite{tropp2015introduction}, which allows one to estimate each $\sigma_i^2$  to error $\eta n^2 $ by taking a random rectangular $O(1/\eta^2) \times O(1/\eta^2)$ sized submatrix. The error in estimating $\sum_{i \leq t} \sigma_i^2(\AA)$ is then $t \eta n^2$, thus one needs to set $\eta = O(\eps/t)$, giving a $O(1/\eps^8)$ tester with two-sided error. Using a careful bucketing analysis on the error, along with the more powerful Interior Matrix Chernoff bounds of Gittens and Tropp \cite{gittens2011tail}, one can improve this to $O(t^2/\eps^4) = O(1/\eps^6)$. However, substantial improvements on unconditional estimation of $\sum_{i \leq t} \sigma_i^2(\AA)$ seem unlikely. In fact, we demonstrate that event for $t=1$ (spectral norm estimation), tools such as matrix concentration inequalities which sample submatrices of $\AA$, must make $\Omega(1/\eps^4)$ queries (Lemma \ref{lem:lowerboundAA}), which rules out, for instance, a $o(t^2/\eps^4)$ upper bound for general $t$. 
Thus, instead of unconditional estimation, our main insight is to demonstrate conditions under which  $\sum_{i \leq t} \sigma_i^2(\AA)$ can be efficiently estimated. When these conditions do not hold, we show that it is because the Schatten-$1$ norm of our sampled submatrix must be too large, from which we can deduce the existence of negative eigenvalues in our query.

In the first case, if the $t$-th singular value is not too large, say $\sigma_{t+1}(\AA) \leq 10n/t$, we show that the (re-scaled) tail $\frac{n^2}{k^2}\sum_{i > t} \sigma_i^2(\AA_{S \times T})$ of a random rectangular matrix, where $|S| = |T| = k = O(1/\eps^2)$, approximates the tail of $\AA$ to error $O(\eps n^2)$. Our argument relies on splitting $\AA$ into head and tail pieces $\AA = \AA_t + \AA_{-t}$, where $\AA_t$ is $\AA$ projected onto the top-$t$ eigenvectors of $\AA$. We demonstrate that the spectral mass of each is preserved after passing to a random row submatrix, and additionally demonstrate that $ \sigma_{\max}(\AA_{-t}) = \sigma_{t+1}(\AA)$ does not grow too much using spectral decay inequalities for random submatrices \cite{rudelson2007sampling}. This forces the spectrum of $(\AA_{-t})_{S \times [n]}$ to be well spread out, allowing us to apply interlacing inequalities to demonstrate that after adding $(\AA_{t})_{S \times [n]}$ back in, the resulting tail is still sufficiently large, and then iterate the argument when sampling columns to obtain $\AA_{S \times T}$.

On the other hand, if $\sigma_{t+1}(\AA)$ is too large, then after moving to a random row submatrix the spectral norm of $\AA_{-t}$ can concentrate highly in its top eigenvalues, which can then be absorbed by the top $t$ eigenvalues of $\AA_t$, stealing too much mass from the tail. Instead, note that if $\sigma_{t+1}(\AA) \geq 10n/t$, then the Schatten norm of $\AA$ must be large, namely $\sum_i \sigma_i(\AA) > 10 n$, which cannot occur if $\AA$ is PSD. We show that by applying Interior Eigenvalue Matrix Chernoff bounds (mentioned above), we can preserve this fact, obtaining $\frac{n}{k}\sigma_{t+1} (\AA_{S \times T}) > 10n/t$ with good probability when $k =  \Omega(1/\eps^2)$. If this is the case, then the Schatten norm of the submatrix will be too large:  $\|\AA_{S \times T}\|_{\mathcal{S}_1} \geq t (10k/t) > 10k$. To obtain a certificate from this fact, we move to the larger principal submatrix $\AA_{(S\cup T) \times (S\cup T)}$, which we show must still have large Schatten norm, from which we can infer the existence of negative eigenvalues. Similarly, in the earlier case, we show that the large tail of $\AA_{S \times T}$ implies that the principal submatrix  $\AA_{(S\cup T) \times (S\cup T)}$ also has too large of a tail, meaning it must not be PSD.



\subsubsection{Lower Bounds}\label{sec:techlb}
As seen above, the distribution of negative mass in the matrix $\AA$ plays an important role in the complexity of testing if $\AA$ is PSD. Specifically, the problem becomes easier the more concentrated the negative mass is within a few eigenvalues. So in order to avoid a $o(1/\eps^2)$ upper bound from the $\ell_\infty$-testing algorithm, our hard instance must have $|\lambda_{\min}(\AA)|  = O( \eps n)$ in the $\eps$-far case. On the other hand, we cannot allow the negative mass to be extremely spread out, otherwise we would have to add many more \textit{positive} eigenvalues to avoid violating the trace constraint $|\text{Tr}(\AA)| = |\sum_i \lambda_i(\AA)| \leq n$ implied by the boundedness, creating further spectral differences between the instances. With this in mind, our hard distribution will have $1/\eps$ negative eigenvalues, each roughly equal to $\lambda_i(\AA) = - \eps n$. 

\paragraph{The Hard Instance.}
Our first insight is to construct a discrete instance, with the property that the distribution induced by observing a small sample of the ``meaningful'' entries of $\AA$ is \textit{identical} in both cases.
Specifically, we construct two distribtuions: $\mathcal{D}_{\text{YES}}$ and $\mathcal{D}_{\text{NO}}$ over $n \times n$ matrices. In both cases, $\AA$ will be block diagonal, with $k$ disjoint blocks $B_1,B_2,\dots,B_k \subset [n]$, each of size $|B_i| = n/k$, for some parameter $k$; we will later set $k=\Theta(1/\eps)$, so our target lower bound is $\Omega(k^2)$. In $\mathcal{D}_{\text{YES}}$, each block $\AA_{B_i \times B_i}$ will be PSD, whereas in $\mathcal{D}_{\text{NO}}$ we will have $\lambda_{\min}(\AA_{B_i \times B_i}) = - \tilde{\Theta}(n/k) \approx - \eps n$. The partition $B_1 \cup B_2 \cup \dots \cup B_k = [n]$ is chosen randomly, so that for any fixed set of samples, only a small fraction them will be contained inside any block $\AA_{B_i \times B_i}$.
The diagonal entries will always be fixed to $1$, and all off-diagonal entries are either $\{0,1,-1\}$. 
The samples $a_1,a_2,\dots,a_s  \in [n] \times [n]$ of any algorithm can then be interpreted as a graph $H$ (possibly with self-loops), where for each edge $a_r= (i,j) \in E(H)$, the algorithm learns the value $\AA_{i,j} \in \{0,1,-1\}$. 

Now consider the algorithm which just samples a $t \times t$ principal submatrix $T \subset [n]$, so that $H$ is a $t$-clique. Now in expectation $\ex{|T \cap B_i|} = \frac{t}{k}$ for each $i$, however, by a balls and bins argument, as $t$ approaches $k$ we will obtain some blocks $i$ with  $|T \cap B_i| = \Omega(\log k/\log\log k)$. Thus, to fool this query, we must be able to ``fool'' cliques of size roughly $\log k$ within a block $B_i$. On the other hand, an algorithm could find many more entries in a block by lop-sided sampling: for instance, it could sample $k^2$ entries in a single column of $\AA$ ($H$ is a $k^2$-star), getting $k$ entries inside a column of a block $\AA_{B_i \times B_i}$. Thus we must also fool large star queries. It turns out that the right property to consider is the \textit{matching number} $\nu(H)$ of the query graph $H$, i.e. the size of a maximum matching. Notice for a star $H$, we have $\nu(H) = 1$. We prove (roughly) that if within each block $B_i$, one can ``fool'' every query graph $H$ inside $\AA_{B_i \times B_i}$ with matching number $\nu(H) < \ell$, one would obtain a lower bound of $\Omega(k^{\frac{2(\ell -1)}{\ell}})$. Thus, it will suffice to fool all query graphs $H$ within a block $B_i$ with $\nu(H) \leq \log k$. 



For a first step towards this, suppose that in  $\mathcal{D}_{\text{YES}}$, we set each block independently to $\AA_{B_i \times B_i} = v v^\top$, where $v \in \{1,-1\}^{|B_i|}$ is a random sign vector, and in $\mathcal{D}_{\text{NO}}$, we set $\AA_{B_i \times B_i} = - v v^\top$ (except we fix the diagonal to be $1$ in both cases). Now notice that the distribution of any individual entry $(\AA_{B_i \times B_i})_{a,b}$ is symmetric, and identical in both  $\mathcal{D}_{\text{YES}}$ and $\mathcal{D}_{\text{NO}}$. Furthermore, it is not difficult to check that the distribution of a path or star query $H$ inside of $\AA_{B_i \times B_i}$ is also identical in both cases. On the other hand, if $H$ contained a \textit{triangle}, then this would not be the case, since in $\mathcal{D}_{\text{YES}}$ one could never have a negative cycle $(x,y,z)$ where $v_x v_y = v_y v_z  = v_z v_x = -1$, whereas this could occur in $\mathcal{D}_{\text{NO}}$, since we could have that $-v_x v_y = -v_y v_z  = -v_z v_x = -1$. Thus, roughly, to distinguish between these distributions $\mathcal{D}_{\text{YES}}$ from $\mathcal{D}_{\text{NO}}$, an algorithm must sample a triangle within one of the blocks $\AA_{B_i \times B_i}$, which one can show requires $\Omega(k^{4/3})$ queries, yielding a first lower bound.\footnote{Note that $\nu(H) = 1$ for a triangle $H$, so the $\Omega(k^{2(\ell - 1)/\ell})$ lower bound when $\nu(H) < \ell$ is actually loose here.}

\paragraph{Boosting to $\Omega(k^{2})$.}
Given the above example, we would now like to construct instances which fool $H$ with larger and larger $\nu(H)$. In fact, our next insight is to have an even simpler structure on $\mathcal{D}_{\text{YES}}$ and $\mathcal{D}_{\text{NO}}$: each of them will be a random permutation of one of two \textit{fixed} matrices $\DD_1,\DD_2$ respectively. We now formalize the ``fooling'' condition we need. For a matrix $\BB$ and a query graph $H$, let $(\BB)_H$  denote the result of setting all entries of $\BB$ not in $H$ equal to zero. Then the matrices $\DD_1,\DD_2$ must have the property that for any graph $H$ with $\nu(H) \leq \log k$, if $\sigma:[m] \to [m]$ is a random permutation and $\P_\sigma \in \R^{m \times m}$ is the row permutation matrix corresponding to $\sigma$, then the distribution of $(\P_\sigma \DD_1 \P_\sigma^\top)_H$ is identical to the distribution $(\P_\sigma \DD_2 \P_\sigma^\top)_H$. We call this property $H$-\textit{subgraph equivalence}. This implies that any algorithm which queries the edges in $H$ inside of $\P_\sigma \DD_1 \P_\sigma^\top$ or $\P_\sigma \DD_2 \P_\sigma^\top$ will be unable to distinguish between them with any advantage. To obtain a lower bound, we must also have a gap between $\lambda_{\min}(\DD_1)$ and $\lambda_{\min}(\DD_2)$, so that their spectrum can be shifted to make one PSD and the other far. Furthermore, neither $\lambda_{\min}(\DD_1)$ or $\lambda_{\min}(\DD_2)$ can be too negative, otherwise by shifting we would lose boundedness of the entries.



A priori, it is not even clear that such matrices $\DD_1,\DD_2$ exist, even for fixed values of $\nu(H)$, such as $\nu(H) = 5$. Our main contribution now is to demonstrate their existence for every $\nu(H)$. Our construction is simple, but perhaps surprisingly so. Both $\DD_1,\DD_2$ will be adjacency matrices; in the PSD case, we set $\DD_1$ to be the cycle graph $C_{2m+1}$ on $2m+1 = \Theta(\log k)$ vertices, and in the $\eps$-far case we set $\DD_2$ to be the disjoint union of two cycles $C_{m+1} \oplus C_m$. Since one of $m$ and $m+1$ is even, while $2m+1$ is odd, we will have that $\lambda_{\min}(C_{m+1} \oplus C_m) = -2$, but $\lambda_{\min}(C_{2m+1} ) > -2$.\footnote{To intuitively see why this is true, note that if $m$ is even and $v \in \{-1,1\}^m$ is the vector that assigns opposite signs to adjacent vertices of $C_m$, then we have $C_m v = -2 v$. However, if $m$ is odd, this assignment $v$ is no longer possible.} 
To show subgraph equivalence, it suffices to show a slightly more general version of the following: for any graph $H$ with $\nu(H) < m/4$, the number of subgraphs of $C_{2m+1}$ isomorphic to $H$ is the same as the number of subgraphs of $C_{m+1} \oplus C_m$ isomorphic to $H$.\footnote{A more general statement is needed since $H$ can also query for edges which do not exist in $C_{2m+1}$.}  Note that if $\nu(H) < m/4$, then $H$ is just a disjoint collection of paths.

Our proof of this fact is by a construction of a bijection from arrangements of $H$ in $C_{2m+1}$ to $H$ in $C_{m+1} \oplus C_m$. While a seemingly simple property, some care must be taken when designing a bijection. 
Our mapping involves first ``swapping'' two paths (whose length depends on $H$) in $C_{2m+1}$, before ``splitting'' $C_{2m+1}$ into two cycles of length $m$ and $m+1$. We direct the reader to Section \ref{sec:CnLemma} for further details.


\vspace{-.1 in}

\paragraph{Amplifying the Gap.}
The subgraph equivalence between $C_{2m+1}$ and $C_{m+1} \oplus C_m$ prevents any algorithm from distinguishing between them with a small number of samples, however the gap in the minimum eigenvalue shrinks at the rate of $\Theta(1/m^2)$. Meaning, if we set  $\gamma =  \lambda_{\min}(C_{2m+1}) = 2-\Theta(1/m^2)$, while the matrix $\gamma \mathbf{I} + C_{2m+1}$ is PSD and has constant sized entries, we only have $\lambda_{\min}(\gamma \mathbf{I} + C_{m+1} \oplus C_m) = - \Theta(1/m^2)$, which is not far enough from PSD. 
Instead, recall that we only need $m = \Omega(\log k)$ to fool all $H$ with $\nu(H) \leq \log k$, but the block size which we must fill is much larger: $\AA_{B_i \times B_i}$ has size $|B_i| = n/k$. Thus, instead of setting $m = \Theta(n/k)$ and filling all of $\AA_{B_i \times B_i}$ with the cycles, we set $m = \Theta(\log k)$, and we amplify the spectral gap by taking the tensor product of the small graphs $C_{2m+1}$ and  $C_{m+1} \oplus C_m$ with a large, fixed matrix $\MM$, so that $(\gamma \mathbf{I} + C_{2m+1}) \otimes \MM$ has $|B_i|$ rows and columns. We prove that taking the tensor product with any fixed $\MM$ preserves the subgraph equivalence properties of the original matrices.  From here, our lower bounds for testing PSD with $\ell_2$ gap, Schatten norms, Ky fan, and the cost of the best rank-$k$ approximation, all follow by a proper choice of $\MM$.  For PSD testing, we can choose $\MM = \mathbf{1}$ to be the all $1$'s matrix, and to amplify the gap in Schatten $1$ norm, we can choose $\MM$ to be a random Rademacher matrix. Since $\MM = \mathbf{1}$ is PSD and $\|\MM\|_2 = \wt{\Omega}(n/k)$, the gap is amplified to the desired $- \wt{\Omega}(n/k)$.  Finally, we remark that to obtain a lower bound for another norm, any matrix $\MM$ which is large in that norm may be suitable, so long as the original sub-graph equivalent matrices also have a gap in that norm. We pose it as an interesting open problem to design other pairs of matrices $\DD_1,\DD_2$ with different spectral gaps which have good sub-graph equivalence properties.

\section{Preliminaries}
We now introduce the notation and definitions that will be used consistently throughout the paper. Additional, specialized notation will be introduced as needed in their respective sections. Specifically, our lower bound construction in Section \ref{sec:lb} utilizes several additional pieces of notation, such as those concerning signed graphs, which are introduced at the beginning of that section. 

\paragraph{Singular Values and Eigenvalues.}We use boldface $\AA$ notation to denote matrices.
For a $n \times d$ matrix $\AA$,  let $\sigma_{\max}(\AA) = \sigma_1(\AA) \geq \sigma_2(\AA) \geq \dots \geq \sigma_{\min\{n,d\}}(\AA) = \sigma_{\min}(\AA)$ denote the singular values of $\AA$. If $\AA$ is rank $r$, let $\AA = \U\mathbf{\Sigma} \V^\top$ be its singular value decomposition, where $\U \in \R^{n \times r},\V\in \R^{d \times r}$ have orthonormal columns, and $\mathbf{\Sigma} \in  \R^{r \times r} $ is a diagonal matrix with the (non-zero) singular values $\sigma_i$ on the diagonal. We use $\mathbf{\Sigma}_k$ to denote the matrix $\mathbf{\Sigma}$ but with all entries but the $k$ largest singular values removed and use $\mathbf{\Sigma}_{-k}$ to denote the matrix $\mathbf{\Sigma}$ but with all entries but the $n-k$ smallest singular values removed. Let $\AA_k = \U \Sigma_k \V^\top$ and  $\AA_{-k} = \U \Sigma_{-k}\V^\top$. The matirx $\AA_k$ is referred to as the \textit{truncated SVD} of $\AA$, and is the best rank-$k$ approximation to $\AA$: $\| \AA - \AA_k\|_F^2 = \sum_{i > k}\sigma^2(\AA) = \min_{\BB \text{ rank-k}} \|\AA - \BB\|_F^2$.
For the special case when $\AA \in \R^{n \times n}$ is symmetric, we use $\U \mathbf{\Lambda} \U^\top$ to denote the Eigenvalue Decomposition of $\AA$, where $\lambda_{\max}(\AA) = \lambda_1(\AA) \geq \lambda_2(\AA) \geq \dots \geq \lambda_n(\AA) = \lambda_{\min}(\AA)$ denote the eigenvalues of $\AA$.  A real-symmetric matrix $\AA$ is said to be \textit{Positive Semi-Definite} (PSD) if $\lambda_{\min} \geq 0$, which is equivalent to having $x^\top \AA x \geq 0$ for all $x \in \R^n$.   We will utilize the Loewner ordering on symmetric matrices.

\begin{definition}[Loewner Ordering]
For symmetric matrices $\BB,\DD$, we write $\BB \succeq \DD$ if $\BB - \DD$ is PSD. 
\end{definition}
\noindent
Notice that if $\BB \succeq \DD$, then by definition $x^\top \BB x \geq x^\top \DD x$ for all $x \in \R^n$. Then by an application of the Courant-Fischer variational principle for eigenvalues, we have that $\lambda_i(\BB) \geq \lambda_i(\DD)$ for all $i \in [n]$.

\paragraph{Matrix Norms and Submatrices.}
We use the notation $\|\AA\|_2 = \sigma_{\max}(\AA)$ to denote the spectral norm of $\AA$, $\|\AA\|_F = (\sum_{i,j} A_{i,j}^2 )^{1/2}= (\sum_{i=1}^n \sigma_i^2(\AA))^{1/2}$ to denote the Frobenius norm of $\AA$.  For $p \geq 1$, we write
$\|\AA\|_{\mathcal{S}_p} = (\sum_{i=1}^n \sigma_{i}^p(\AA))^{1/p}$ to denote the Schatten $p$-norm of $\AA$, and $\|\AA\|_{\KF(p,k)} = (\sum_{i=1}^k \sigma^p_{i}(\AA))^{1/p}$ to denote the $(p,k)$-Ky-Fan norm of $\AA$. If $p$ is not specified for a Ky-Fan norm, it is assumed to be $1$, namely $\|\AA\|_{\KF(k)} = \|\AA\|_{\KF(1,k)}$. 
For subsets $S,T\subseteq [n]$, we denote the matrix $\AA_{S \times T} \in \R^{|S| \times |T|}$ as the matrix $\AA$ restricted to the submatrix of the rows in $S$ and the columns in $T$. If $S = T$, then the square submatrix $\AA_{S \times T} = \AA_{S \times S}$ is called a \textit{principal submatrix} of $\AA$. For a vector $x \in \R^n$ and subset $S \subset [n]$, we write $x_S \in \R^n$ to denote the vector obtained after setting equal to zero all coordinates $x_i$ with $i \notin S$. Finally, we use the notation $\AA_{i,*}$ to denote the $i$-th row of $\AA$, and $\AA_{*,i}$ to denote the $i$-th column of $\AA$.

\section{PSD Testing with \texorpdfstring{$\ell_\infty$}{L-inf} Gap}\label{sec:Linfty}


In this section, we introduce our algorithm for the PSD testing problem with $\ell_\infty$-gap. As discussed earlier, we consider a more general version of the $\ell_\infty$ gap than the definition presented in Problem \ref{prob:linf}, which allows one to test a notion of positive semi-definitness which applies to non-symmetric matrices as well. Specifically, we define the \emph{PSD} case as when $x^\top \AA x \geq 0$ for all $x \in \R^n$, and the \emph{far} case as when $x^\top \AA x < - \eps n$ for a unit vector $x$. We note that if $\AA$ is symmetric, this definition is equivalent to  Problem \ref{prob:linf}. In fact, as we will see shortly, one can always reduce the non-symmetric case to the symmetric case, so this distinction will not matter algorithmically. 
Formally, we solve the following problem:

\begin{definition}[General PSD Testing with $\ell_{\infty}$-Gap.]\label{def:linfty}
Fix, $\eps \in (0,1]$ and let $\AA \in \R^{n \times n}$ be any matrix satisfying $\|\AA\|_\infty \leq 1$. The goal is to distinguish between the following two cases:
\begin{itemize}
	\item \textbf{YES Instance}: $\AA$ satisfies $x^\top \AA x  \geq 0$, for all $x \in \R^n$. 
	\item \textbf{NO Instance}: There exists a unit vector $x \in \R^n$ such that $x^\top \AA x < - \eps n$.
\end{itemize}
with probability at least $2/3$.
\end{definition}

\paragraph{Reducing to the symmetric case} In the case where $\AA$ is symmetric, as in Problem \ref{prob:linf}, the above gap instance can be restated in terms of the minimum eigenvalue of $\AA$. Specifically, we are promised that either  $\lambda_{\min}(\AA) \geq 0$ or $\lambda_{\min}(\AA) \leq -\eps n$. 
However, we now observe that one can reduce to the symmetric case with only a factor of $2$ loss in the query complexity, by simply querying the symmetrization $(\AA + \AA^\top)/2$. First note, that for any $x \in \R^{n}$ and any matrix $\AA \in \R^{n \times n}$, we have $x^\top \AA x = x^\top \AA^\top x$, thus for any $x$ we have $x^\top \AA x = x^\top \frac{\AA + \AA^\top}{2}x$. Thus $x^\top \AA x \geq 0$ for all $x$ if and only if $x^\top \frac{\AA + \AA^\top}{2} x \geq 0$ for all $x$, which occurs if and only if the matrix $\frac{\AA + \AA^\top}{2}$ is PSD. Similarlly, we have that $x^\top \AA x \leq - \eps n$ for some unit vector $x$ if and only if $x^\top \frac{\AA + \AA^\top}{2} x \leq - \eps n$ for some unit vector $x$, which occurs if and only if $\lambda_{\min}(\frac{\AA + \AA^\top}{2} )\leq - \eps n$. Note also that the matrix $\frac{\AA + \AA^\top}{2} $ has bounded entries $\|\frac{\AA + \AA^\top}{2} \|_\infty \leq 1$ if $\|\AA\|_\infty \leq 1$. Moreover, query access to $\frac{\AA + \AA^\top}{2} $ can be simulated via query access to $\frac{\AA + \AA^\top}{2}$ with a loss of at most a factor of $2$ in the query complexity, by symmetrizing the queries. In fact, our algorithms will not even incur this factor of $2$ loss, since all queries our algorithms make will belong to principal submatrices of $\AA$. Thus, in what follows, we can restrict ourselves to the original formulation as specified in Problem \ref{prob:linf}, and assume our input $\AA$ is symmetric.


The goal of this section is now to prove the following theorem, which demonstrate the existence of a $\tilde{O}(1/\eps^2)$ query one-sided error tester for the above problem. In Section \ref{sec:lb}, we demonstrate that this complexity is optimal (up to $\log(1/\eps)$ factors), even for testers with two sided error (Theorem \ref{thm:linftyLB}).

\begin{theorem}[Query Optimal One-Sided Tester for $\ell_{\infty}$ Gap (see Theorem \ref{thm:inftymain})]
	There is an algorithm which, given $\AA$ with $\|\AA\|_\infty \leq 1$ such that either $x^\top \AA x \geq 0$ for all $x$ (\texttt{YES} case), or such that $x^\top\AA x \leq - \eps n$ for some $x \in \R^n$ with $\|x\|_2 \leq 1$  (\texttt{NO} case), distinguishes the two cases with probability at least $3/4$, while making at most $\tilde{O}(\frac{1}{\eps^2})$ queries to the entries of $\AA$, and runs in time $\tilde{O}(1/\eps^{\omega})$, where $\omega < 2.373$ is the exponent of matrix multiplication. Moreover, in the first case when  $x^\top \AA x \geq 0$ for all $x$, is PSD, the algorithm always correctly outputs \texttt{YES}.
\end{theorem}

\paragraph{Algorithmic Setup} 

First recall that if $\AA$ is PSD, then then every \textit{principal} submatrix $\AA_{T \times T}$ of $\AA$ for $T \subseteq [n]$ is also PSD. Thus, it will suffice to query a collection of principal submatrices $\AA_{T_1 \times T_1} , \AA_{T_2 \times T_2} , \dots, \AA_{T_t \times T_t} $ of $\AA$, and return \texttt{Not PSD} if any one of them fails to be PSD. Such a algorithm then naturally has one-sided error, since if $\AA$ is PSD it will always return PSD. Thus, in the remainder of the section, it suffices to consider only the \texttt{NO}  case, and demonstrate that, if $\AA$ satisfies $x^\top\AA x \leq - \eps n$ for some unit vector $x \in \R^n$, then with good probability at least one of the sampled principal submatrices will fail to be PSD. 

Moreover, as shown above, it suffices to consider the case where $\AA$ is symmetric. In this case, we will fix $x$ to be the eigenvector associated with the smallest eigenvalue of $\AA$. Thus, in what follows, we can fix $x$ so that $\min_{z \in \R^n : \|z\|_2 = 1} z^\top \AA z = x^\top \AA x = \lambda_{\min}(\AA) =  -\eps n$. 
Notice here we define $\eps$ to satisfy the \emph{equality} $\lambda_{\min}(\AA)= -\eps n$, however our algorithm need only know a lower bound $\eps_0 < \eps$ on $\eps$. The reason for this is that the input parameter $\eps_0$ will only effect the sizes of the random submatrices being sampled (smaller $\eps_0$ increases the size). Thus, an algorithm run with parameter $\eps_0$ can be simulated by first running the algorithm with parameter $\eps > \eps_0$, and then randomly adding rows and columns to the sampled submatrices from the correct marginal distribution. Thus, there is a coupling such that the submatrices chosen by an algorithm with any input $\eps_0 < \eps$ will always contain the submatrices sampled by an algorithm given the input exactly $\eps$, so if the latter algorithm sampled a non-PSD submatrix, so would the former.
Thus, for the sake of analysis, we can assume that the value $\eps$ is known.

Throughout the following section, we will assume $1/\eps < c \cdot n$ for some sufficiently small constant $c$. Notice that if this was not the case, we would have $1/\eps^2 = \Omega(n^2)$, and we would be permitted to read  the entire matrix $\AA$, as this is within our target budget of $\tilde{O}(1/\eps^2)$.

\subsection{Warm-up: a \texorpdfstring{$O(1/\eps^3)$}{epsilon cube}  algorithm}\label{sec:warmup}
We first describe a $O(1/\eps^3)$ query algorithm for the problem of PSD testing with $\ell_\infty$-gap. The general approach and results of this algorithm will be needed for the more involved $\wt{O}(1/\eps^2)$ query algorithm which we shall develop in the sequel. As noted above, it suffices to analyze the NO case, where we have $x^\top \AA x = \lambda_{\min}(\AA) = - \eps n$ for a unit vector $x \in \R^n$. 
Our goal will be to analyze the random variable $Z = x_T^\top \AA_{T \times T} x_T$, where $T \subset [n]$ is a random subset, where each $i \in [n]$ is selected independently with some probability $\delta$. Notice that if $\delta_i \in \{0,1\}$ is an indicator variable indicating that we sample $i \in T$, then we have $Z = x_T^\top \AA_{T \times T} x_T= \sum_{i,j} x_i \AA_{i,j} x_j \delta_i \delta_j$. 

Now, algorithmically, we do not know the vector $x$. However, if we can demonstrate concentration of $Z$, and show that $Z < 0$ for our sampled set $S$, then we can immediately conclude that $\AA_{T \times T}$ is not PSD, a fact which \textit{can} be tested. Thus, the analysis will proceed by pretending that we did know $x$, and analyzing the concentration of $x_T^\top \AA_{T \times T} x_T$. In the following section, however, we will ultimately analyze the concentration of this random variable with respects a slightly different vector than $x$.

We first remark that we can assume, up to a loss in a constant factor in the value of $\eps$, that the diagonal of $\AA$ is equal to the identity.

 \begin{proposition}\label{prop:1diag}
We can assume $\AA_{i,i} = 1$, for all $i \in [n]$. Specifically, by modifying $\AA$ so that $\AA_{i,i} = 1$, for all $i \in [n]$, the completeness (PSD) case is preserved and the soundness (not PSD) case is preserved up to a factor of $1/2$ in the parameter $\eps$.
\end{proposition}
\begin{proof}
	Every time we observe an entry $\AA_{i,i}$ we set it equal to $1$. In this new matrix, if $\AA$ was PSD to begin with, then $\AA$ will still be PSD, since this modification corresponds to adding a non-negative diagonal matrix to $\AA$. If $x\AA x \leq - \eps n$ originally for some $x \in \R^n$, then $x^\top\AA x \leq - \eps n + 1$ after this change, since the diagonal contributes at most $\sum_i \AA_{i,i} x_i^2 \leq \|x\|_2^2 \leq 1$ to the overall quadratic form. Note this additive term of $1$ is at most $(\eps n)/2$ since we can assume $\eps = \Omega(1/n)$. 
\end{proof}

We now notice that if $x$ is the eigenvector associated with a a large enough eigenvalue, the $\ell_2$ mass of $x$ cannot be too concentrated.

\begin{proposition}\label{prop:spread}
Let $\AA \in \R^{n \times n}$ be a symmetric matrix with $\lambda_{\min}(\AA) = - \eps n$, and let $x \in \R^n$ be the (unit) eigenvector associated with $\lambda_{\min}(\AA)$. Then we have that  $\|x\|_\infty \leq \frac{1}{\eps\sqrt{ n}}$.
\end{proposition}
\begin{proof}
By Cauchy-Schwartz, for any $i \in [n]$:
\[ |\lambda_{\min}| \cdot | x_i| = |\langle \AA_{i,*} , x\rangle| \leq  \|  \AA_{i,*} \|_2 \leq \sqrt{n}   \]
from which the proposition follows using $\lambda_{\min}(\AA) = - \eps n$.
\end{proof}

Recall that our goal is to analyze the random variable $Z = x_T^\top \AA_{T \times T} x_T= \sum_{i,j} x_i \AA_{i,j} x_j \delta_i \delta_j$. To proceed, we bound the moments of $Z$. Our bound on these moments can be tightened as a function of the \textit{row and column contributions} of the target vector $x$, which we now define.

\begin{definition}\label{def:RRCC1}
Fix any $y \in \R^n$. Then for any $i \in [n]$, define the total row and column contributions of $i$ as $\RRR_i(y) = \sum_{j \in [n] \setminus i} y_i \AA_{i,j}y_j$ and $\CCC_i(y) = \sum_{j \in [n] \setminus i} y_j \AA_{j,i} y_i$ respectively.  
\end{definition}\noindent
Notice from the above definition, we have $\sum_i \RRR_i(y) + \CCC_i(y) = 2\left(y^\top \AA y - \sum_i \AA_{i,i} y_i^2\right) $.
 \begin{fact}\label{fact:case1OPT}
 Let $x \in \R^n$ be the eigenvector associated with $\lambda_{\min}(\AA)$. Then we have $\RRR_i(x) + \CCC_i(x) \leq 0$ for all $i \in [n]$.
 \end{fact}
 \begin{proof}
     Suppose there was an $i$ with $\RRR_i(x) + \CCC_i(x) > 0$. Then setting $z = x_{[n] \setminus i}$ we have $z^\top \AA z = \langle x , \AA x\rangle - (\RRR_i(x) + \CCC_i(x))  - \AA_{i,i} (x_i)^2$. Recall from Proposition \ref{prop:1diag} that we can assume $\AA_{i,i}  = 1$ for all $i$, thus it follows that $z^\top \AA z < \langle x , \AA x\rangle$, which contradicts the optimality of $x$.
     \end{proof}

\noindent
	We now bound the expectation of the random quadratic form.
	\begin{proposition}[Expectation Bound]\label{prop:exp} Let $\AA \in \R^{n \times n}$ be a matrix with $\|\AA\|_\infty \leq 1$, and let $y \in \R^n$ be any vector with $\|y\|_2 \leq 1$ and $y^\top \AA y <  - \eps n$.
		Let $Z = \sum_{i,j} y_i \AA_{i,j} y_j \delta_i \delta_j$, where $\delta_1,\dots,\delta_n \sim \text{Bernoulli}(\frac{k}{n})$. Then if $k \geq 8/\eps$, we have $\ex{Z} \leq - \frac{\eps k^2}{ 4n }$.
	\end{proposition}
\begin{proof} Let $c_{i,j} = \AA_{i,j}y_i y_j$.
	 First note, for any $i \in [n]$, the term $c_{i,j}$ is included in $T$ with probability $k/n$ if $i = j$, and with probability $k^2/n^2$ if $i \neq j$. So 
	\begin{equation}
	\begin{split}
	\bex{Z} &= \sum_{i \neq j}  \frac{k^2}{n^2} c_{i,j} + \sum_{i \in [n]} \frac{k}{n} c_{i,i}\\
	&=   \frac{k^2}{n^2} \left( \langle y ,\AA y\rangle - \sum_{i \in [n]}\AA_{i,i} y_i^2 \right) +  \frac{ k}{n}\sum_{i \in [n]}\AA_{i,i} y_i^2 \\
	& \leq - \frac{\eps  k^2}{ 2n } + \left(\frac{ k}{n} +  \frac{ k^2}{n^2}\right) \sum_{i \in [n]} y_i^2\\
		& \leq - \frac{\eps k^2}{ 2n } + \frac{2k}{n} \leq - \frac{\eps k^2}{ 4n }\\
	\end{split}
\end{equation}
Where in the last inequality, we assume $k \geq 8 /\eps$.	
\end{proof}\noindent
Next, we bound the variance of $Z$. We defer the proof of the following Lemma to Section \ref{sec:variance}.

	\begin{lemma}[Variance Bound]\label{lem:var}
		Let $\delta_1,\dots,\delta_n \sim \text{Bernoulli}(\frac{k}{n})$. Let $y \in \R^n$ be any vector such that $\|y\|_2 \leq 1, \|y\|_\infty \leq \frac{1}{\eps \sqrt{n}}$, and $y^\top \AA y = - \eps n$, where $\AA \in \R^{n \times n}$ satisfies $\|\AA\|_\infty \leq 1$. Further suppose that $\RRR_i(y) + \CCC_i(y) \leq 0$ for each $i \in [n]$. Then, assuming $k \geq 6/\eps$, we have
		 \[\mathbf{Var}\left[\sum_{i,j} y_i \AA_{i,j} y_j \delta_i \delta_j	\right] \leq O\left(\frac{k^3}{n^2}\right)  \]
		  Moreover, if the tighter bound $\|y\|_\infty \leq \frac{\alpha}{\eps \sqrt{n}}$ holds for some $\alpha \leq 1$, we have
\[\mathbf{Var}\left[\sum_{i,j} y_i \AA_{i,j} y_j \delta_i \delta_j\right] \leq  O\left(\frac{k^2}{n^2}+ \frac{\alpha k^3}{n^2} \right) 	\]
\end{lemma}

We note that the variance of the random quadratic form can be improved if we have tighter bounds on certain properties of the target vector $y$. We demonstrate this fact in the following Corollary, which we will use in Section \ref{sec:improving}. Note that the assumptions of Corollary \ref{cor:var} differ in several minor ways from those of Lemma \ref{lem:var}. For instance, we do not require $k \geq 6/\eps$ (we note that this assumption was required only to simply the expression in Lemma \ref{lem:var}). Also notice that we do not bound the diagonal terms in Corollary \ref{cor:var}.   
We defer the proof of Corollary \ref{cor:var} to Section \ref{sec:variance}.

\begin{corollary}[Tighter Variance Bound]\label{cor:var}		Let $\delta_1,\dots,\delta_n \sim \text{Bernoulli}(\frac{k}{n})$.
	Let $\AA \in \R^{n \times n}$ with $\|\AA\|_\infty \leq 1$ be any matrix and $y$ a vector such that $|y^\top\AA y| \leq c_1\eps n$ for some value $c_1>0$, and such that $\|y\|_\infty \leq \frac{\alpha}{\eps \sqrt{n}}$ for some $\alpha > 0$. 
	Let $\ZZ \in \R^n$ be defined by $\ZZ_i = \RRR_i(y) + \CCC_i(y)$ for $i \in [n]$, and suppose we have $\|\ZZ\|_2^2 \leq c_2 \eps n$. 
	Then we have	
	\[\mathbf{Var}\left[\sum_{i\neq j} y_i \AA_{i,j} y_j \delta_i \delta_j	\right] \leq  O\left(\frac{k^2}{n^2} +  \frac{c_1^2 k^4 \eps^2}{n^2} +  \frac{(c_1 +c_2 )\eps  k^3}{n^2}	+ \frac{\alpha^2 k^3}{n^2} \right) 	\]
\end{corollary}

We now observe that the variance computations from Lemma \ref{lem:var} immediately gives rise to a $O(1/\eps^3)$ algorithm.

\begin{theorem}
	There is a non-adaptive sampling algorithm which queries $O(\eps^{-3})$ entries of $\AA$, and distinguishes the case that $\AA$ is PSD from the case that $\lambda_{\min}(\AA) < - \eps n$ with probability $2/3$. 
\end{theorem}	\begin{proof}
Let $x \in \R^n$ be the eigenvector associated with $\lambda_{\min}(\AA) = - \eps n$ (recall that we can assume equality), and let $Z_1,\dots,Z_d$ be independent repetitions of the above process, with $k = 10/\eps $ and $d = 3840/\eps$. Let $Z = \frac{1}{d}\sum_{i=1}^d Z_i$. Then $\bvar{Z} \leq \frac{6}{d}\frac{k^3}{n^2}$ by Lemma \ref{lem:var}, where we used the bounds on $\|x\|_\infty$ from Proposition \ref{prop:spread} and the property that $\RRR_i(x) + \CCC_i(x) \leq 0$ for all $i$ from Fact \ref{fact:case1OPT} to satisfies the assumptions of Lemma \ref{lem:var}. By Chebysev's inequality:
\begin{equation}
\begin{split}
\bpr{ Z \geq -\frac{\eps k^2}{4n} +\frac{\eps k^2}{8n} }	&\leq \left(\frac{64 n^2}{\eps^2 k^4}\right)\left(\frac{6k^3}{dn^2}\right) \\
& \leq \frac{1}{10 \eps k}  \\ 
& \leq \frac{1}{100}
\end{split}
\end{equation}
It follows that with probability $99/100$, the average of the $Z_i$'s will be negative. Thus at least one of the $Z_i$'s must be negative, thus the submatrix corresponding to this $Z_i$ will not be PSD. The total query complexity is $O(k^2 d) = O(\eps^{-3})$.

\end{proof}

\subsection{Variance Bounds}
\label{sec:variance}

In this section, we provide the proofs of the variance bounds in Lemma \ref{lem:var} and Corollary \ref{cor:var}. For convenience, we restate the Lemma and Corollary here before the proofs.

	\smallskip \smallskip \smallskip 
\noindent \textbf{Lemma \ref{lem:var}}  \textit{ 
	Let $\delta_1,\dots,\delta_n \sim \text{Bernoulli}(\frac{k}{n})$. Let $y \in \R^n$ be any vector such that $\|y\|_2 \leq 1, \|y\|_\infty \leq \frac{1}{\eps \sqrt{n}}$, and $y^\top \AA y = - \eps n$, where $\AA \in \R^{n \times n}$ satisfies $\|\AA\|_\infty \leq 1$. Further suppose that $\RRR_i(y) + \CCC_i(y) \leq 0$ for each $i \in [n]$. Then, assuming $k \geq 6/\eps$, we have
		 \[\mathbf{Var}\left[\sum_{i,j} y_i \AA_{i,j} y_j \delta_i \delta_j	\right] \leq O\left(\frac{k^3}{n^2}\right)  \]
		  Moreover, if the tighter bound $\|y\|_\infty \leq \frac{\alpha}{\eps \sqrt{n}}$ holds for some $\alpha \leq 1$, we have
\[\mathbf{Var}\left[\sum_{i,j} y_i \AA_{i,j} y_j \delta_i \delta_j\right] \leq  O\left(\frac{k^2}{n^2}+ \frac{\alpha k^3}{n^2} \right) 	\]
}\smallskip \smallskip \smallskip 

\begin{proof}
Let $c_{i,j} = \AA_{i,j}y_i y_j$. We have
	\begin{equation} \label{eqn:varbig}
\begin{split}
&\mathbf{Var}\left[\sum_{i,j} y_i \AA_{i,j} y_j \delta_i \delta_j	\right] \leq  \frac{k}{n}\sum_{i  }  c_{i,i}^2 + \frac{k^2}{n^2} \sum_{i \neq j} c_{i,j}^2 + \frac{k^2}{n^2} \sum_{i \neq j} c_{i,j} c_{j,i} + \frac{k^2}{n^2 }\sum_{i \neq j } c_{i,i} c_{j,j} + \frac{k^2}{n^2} \sum_{i \neq j}c_{i,i}c_{i,j}   \\
 +& \frac{k^2}{n^2} \sum_{i \neq j}c_{i,i}c_{j,i} + \frac{k^3}{n^3}\sum_{i\neq j \neq u} c_{i,j} c_{u,j}  + \frac{k^3}{n^3}\sum_{j\neq i \neq u} c_{i,j} c_{i,u}  + \frac{k^3}{n^3}\sum_{i\neq j \neq u} c_{i,j} c_{j,u}+ \frac{k^3}{n^3}\sum_{j\neq i \neq u} c_{i,j} c_{u,i} \\ 
 & +  \frac{2k^3}{n^3} \sum_{i \neq j \neq u} c_{i,i} c_{j,u} + \frac{k^4}{n^4}\sum_{i \neq j \neq v \neq u} c_{i,j} c_{u,v}  - \left( \frac{k^2}{n^2}\sum_{i \neq j} y_i \AA_{i,j} y_j - \frac{k}{n} \sum_i \AA_{i,i} y_i^2 \right)^2 \\
\end{split}
\end{equation}

We first consider the last term $\frac{k^4}{n^4}\sum_{i \neq j \neq v \neq u} c_{i,j} c_{u,v} = \sum_{i \neq j} y_i \AA_{i,j} y_j \sum_{u \neq v \neq i \neq j} y_u \AA_{u,v} y_v$. Here $i \neq j \neq v \neq u$ means all $4$ indices are distinct.  Note that this term is canceled by a subset of the terms within $ \left( \frac{k^2}{n^2}\sum_{i \neq j} y_i \AA_{i,j} y_j - \frac{k}{n} \sum_i \AA_{i,i} y_i^2 \right)^2$.  Similarly, the term $\frac{k^2}{n^2 }\sum_{i \neq j } c_{i,i} c_{j,j}$ cancels. Moreover, after expanding $ \left( \frac{k^2}{n^2}\sum_{i \neq j} y_i \AA_{i,j} y_j - \frac{k}{n} \sum_i \AA_{i,i} y_i^2 \right)^2$, every remaining term which does not cancel with another term exactly is equal to another term in the variance above, but with an additional one (or two) factors of $\frac{k}{n}$ attached. Thus, if we can bound the remaining terms in Equation \ref{eqn:varbig} by some value $B$, then an overall variance bound of $2\cdot B$ will follow.

We now consider $\mathcal{T} = \left(\sum_{j\neq i \neq u} c_{i,j} c_{i,u} +\sum_{i\neq j \neq u} c_{i,j} c_{u,j} +\sum_{j\neq i \neq u} c_{i,j} c_{u,i} +\sum_{i\neq j \neq u} c_{i,j} c_{j,u} \right)$. We have
 \[\sum_{i\neq j \neq u} c_{i,j} c_{i,u} = \sum_i \sum_{j \neq i} y_i \AA_{i,j} y_j \sum_{u \neq i \neq j} y_i \AA_{i,u} y_u\]
  \[\sum_{i\neq j \neq u} c_{i,j} c_{u,j} = \sum_j \sum_{i \neq j} y_i \AA_{i,j} y_j \sum_{u \neq i \neq j} y_u \AA_{u,j} y_j\]
  \[\sum_{j\neq i \neq u} c_{i,j} c_{u,i} = \sum_i \sum_{j \neq i} y_i \AA_{i,j} y_j \sum_{u \neq i \neq j} y_u \AA_{u,i} y_i\]
  \[\sum_{i\neq j \neq u} c_{i,j} c_{j,u} = \sum_j \sum_{i \neq j}  y_i \AA_{i,j} y_j \sum_{u \neq i \neq j} y_j \AA_{j,u} y_u\]

Now for simplicity, we write $\RRR_i = \RRR_i(y)$ and $\CCC_i = \CCC_i(y)$ for $i \in [n]$. Then by assumption, we have 
$\RRR_i + \CCC_i \leq 0$ for each $i$, thus $|\sum_i (\RRR_i + \CCC_i) | = \sum_i |(\RRR_i + \CCC_i)|$. 
Also note that we have $|\sum_i (\RRR_i + \CCC_i)| = |2 y^\top \AA y -2\sum_i \AA_{i,i} y_i^2| \leq  4\eps n $. Now observe  
\[\left|\left( \sum_i\sum_{j \neq i} y_i \AA_{i,j} y_j \sum_{u \neq i \neq j} y_i \AA_{i,u} y_u\right)  - \sum_i\RRR_i^2 \right| =  \sum_i\sum_{u \in [n] \setminus i} y_i^2 \AA_{i,u}^2 y_u^2\leq  \sum_iy_i^2 \leq 1\]
And similarly
\[\left|\left(\sum_j \sum_{i \neq j} y_i \AA_{i,j} y_j \sum_{u \neq i \neq j} y_u \AA_{u,j} y_j\right)  -\sum_j\CCC_i^2 \right| = \sum_j\sum_{u \in [n] \setminus j} y_u^2 \AA_{u,j}^2 y_j^2\leq \sum_jy_j^2 \leq 1\]
\[\left|\left( \sum_i \sum_{j \neq i} y_i \AA_{i,j} y_j \sum_{u \neq i \neq j} y_u \AA_{u,i} y_i\right)  - \sum_i\RRR_i\CCC_i \right| =  \sum_i\sum_{u \in [n] \setminus i} y_i^2 \AA_{i,u} \AA_{u,i} y_u^2\leq  \sum_iy_i^2 \leq 1\]
\[\left|\left( \sum_j \sum_{i \neq j}  y_i \AA_{i,j} y_j \sum_{u \neq i \neq j} y_j \AA_{j,u} y_u\right)  -\sum_j\RRR_j\CCC_j \right| = \sum_j\sum_{u \in [n] \setminus j} y_u^2 \AA_{u,j} \AA_{j,u} y_j^2\leq \sum_jy_j^2 \leq 1\]
Taking these four equations together, we obtain $\left| \mathcal{T} - \sum_i (\RRR_i + \CCC_i)^2 \right| \leq 4$, so it will suffice to upper bound the value $ \sum_i (\RRR_i + \CCC_i)^2$ instead. First note that since $|y_i| \leq \frac{1}{\eps \sqrt{n}}$ for all $i$, so for any $i \in [n]$ we have  \[|(\RRR_i + \CCC_i)| \leq  |\sum_{j \neq i} y_i \AA_{i,j} y_j| + |\sum_{j \neq i} y_j \AA_{j,i} y_i| \leq \frac{1}{\eps \sqrt{n}}(\sum_{j }  2y_j) \leq \frac{2}{\eps \sqrt{n}} \|y\|_1  \leq  \frac{2}{\eps} \] 
Combining this bound with the fact that $\sum_i |(\RRR_i + \CCC_i)|  \leq 4 \eps n$ from earlier, it follows that the sum $\sum_i (\RRR_i + \CCC_i)^2$ is maximized by setting $2 \eps^2 n$ of the terms $(\RRR_i + \CCC_i)$ equal to the largest possible value of $(2/\eps)$, so that $\sum_i (\RRR_i + \CCC_i)^2 \leq 2 \eps^2 n (2/\eps)^2 = O(n)$. This yields an upper bound of $\frac{k^3}{n^3}\mathcal{T} = O(\frac{k^3}{n^2})$. Note, that in general, given the bound $\|y\|_\infty \leq \frac{\alpha}{\eps \sqrt{n}}$ for some value $\alpha \leq 1$, then each term $|(\RRR_i + \CCC_i)| \leq \frac{2\alpha }{\eps}$. On the other hand, $\sum_i |(\RRR_i + \CCC_i)| \leq 4 \eps n$.
Thus, once again, $\sum_i |(\RRR_i + \CCC_i)|^2$ is maximized by setting $\Theta(\eps^2 n/\alpha )$ inner terms equal to $\Theta((\frac{\alpha }{\eps})^2)$, giving $\mathcal{T} \leq \alpha n$ for general $\alpha < 1$. 
Thus, for general $\alpha \leq 1$, we have $\frac{k^3}{n^3}\mathcal{T} = O(\frac{\alpha k^3}{n^2})$.

Next, we bound $\frac{k^2}{n^2} \sum_{i \neq j}c_{i,i}c_{i,j} + \frac{k^2}{n^2} \sum_{i \neq j}c_{i,i}c_{j,i}$ by $\frac{k^2}{n^2}\sum_i y_i^2( \RRR_i + \CCC_i) $. As shown above, $ | \RRR_i + \CCC_i| \leq 2y_i\sqrt{n}$, thus altogether we have
\begin{equation}\label{eqn:secondorderterms}
  \frac{k^2}{n^2}\left( \sum_{i \neq j}c_{i,i}c_{i,j} +\sum_{i \neq j}c_{i,i}c_{j,i} \right)\leq \frac{k^2}{n^2}\sum_i y_i^3 \sqrt{n}   
\end{equation}
 Using that $\|y\|_\infty \leq \frac{\alpha}{\eps \sqrt{n}}$ for $\alpha \leq 1$, and the fact that $\|y\|_2^2 \leq 1$, it follows that $\|y\|_3^3$ is maximized by having $\frac{n \eps^2}{\alpha^2}$ terms equal to $\|y\|_\infty \leq\frac{\alpha}{\eps \sqrt{n}}$, which gives an upper bound of $\|y\|_3^3 \leq \frac{\alpha}{\eps \sqrt{n}}$. Thus, we can bound the right hand side of Equation \ref{eqn:secondorderterms} by $\frac{k^2 \alpha}{n^2 \eps}$, which is $O(k^3/n^2)$ when $\alpha = 1$ using that $k = \Omega(1/\eps)$.

Now, we bound $\frac{k^3}{n^3} \sum_{i \neq j \neq u} c_{i,i} c_{j,u} $ by 

\begin{equation}
\begin{split}
\frac{k^3}{n^3} \sum_{i \neq j \neq u} c_{i,i} c_{j,u} &\leq \frac{k^3}{n^3} \sum_i y_i^2 \AA_{i,i} \sum_{j \neq u \neq i} y_j\AA_{j,u} y_u  \\ 
& \leq \frac{k^3}{n^3} \sum_i y_i^2 \AA_{i,i} \sum_{j \neq u \neq i} y_j\AA_{j,u} y_u    \\
& \leq \frac{k^3}{n^3} \sum_i y_i^2 \AA_{i,i} \left(\eps n + O(1)\right) \\
&\leq \frac{\eps k^3}{n^2} \\
& = O( \frac{k^2}{n^2})
\end{split}
\end{equation}

 Also observe that $\sum_{i ,j }  c_{i,j}^2 \leq \sum_{i,j} y_i^2 y_j^2 = \|y|_2^4 \leq 1$, so $\sum_{i \neq j} c_{i,j}^2 \leq   \sum_{i,  j} c_{i,j}^2  \leq 1$, and also $\sum_{i \neq j} c_{i,j} c_{j,i} \leq  \sum_{i,j} y_i^2 y_j^2  \leq 1$, which bounds their corresponding terms in the variance by $O(k^2/n^2)$.
Finally, we must bound the last term $\frac{k}{n}\sum_{i  }  c_{i,i}^2  = \frac{k}{n}\sum_i y_i^4 \AA_{i,i}^2 \leq\frac{k}{n} \sum_i y_i^4 $. Note that $|y_i| \leq 1/(\eps \sqrt{n})$ for each $i$, and $\|y\|_2 \leq 1$. Thus $\sum_i y_i^4$ is maximized when one has $\eps^2 n$ terms equal to  $1/(\eps \sqrt{n})$ , and the rest set to $0$. So 
 $\sum_i y_i^4 \leq \eps^2 n(\frac{1}{\eps \sqrt{n}})^4 \leq \frac{1}{\eps^2 n}$. 
 In general, if $\|y\|_\infty \leq \frac{\alpha}{\eps \sqrt{n}}$, we have $\sum_i y_i^4 \leq \frac{\eps^2 n}{\alpha^2 }(\frac{\alpha}{\eps \sqrt{n}})^4 \leq \frac{\alpha^2}{\eps^2 n}$. Thus we can bound $\frac{k}{n}\sum_i c_{i,i}^2$ by $O(\frac{k^3 \alpha^2}{n^2})$

Altogether, this gives 

	\begin{equation}
\begin{split}
&\mathbf{Var}\left[\sum_{i,j} y_i \AA_{i,j} y_j \delta_i \delta_j	\right] \leq  O(\frac{k^2}{n^2}+ \frac{\alpha k^3}{n^2} + \frac{\alpha k^2  }{n^2 \eps}  + \frac{\alpha^2 k^3}{n^2}) \\ 
& = O(\frac{k^2}{n^2}+ \frac{\alpha k^3}{n^2}  + \frac{\alpha^2 k^3}{n^2})\\
	 \end{split}
	 \end{equation}
	 which is $O(k^3/n^2)$ in general (where $\alpha \leq 1$), 
	 where we assume $k \geq 6/\eps$ throughout. 
 \end{proof}

	\smallskip \smallskip \smallskip 
\noindent \textbf{Corollary \ref{cor:var}}  \textit{ 
Let $\delta_i \in \{0,1\}$ be an indicator random variable with $\ex{\delta_i} = k/n$. 
	Let $\AA \in \R^{n \times n}$ with $\|\AA\|_\infty \leq 1$ be any matrix and $y$ a vector such that $|y^\top\AA y| \leq c_1\eps n$ for some value $c_1>0$, and such that $\|y\|_\infty \leq \frac{\alpha}{\eps \sqrt{n}}$ for some $\alpha > 0$. 
	Let $\ZZ \in \R^n$ be defined by $\ZZ_i = \RRR_i(y) + \CCC_i(y)$ for $i \in [n]$, and suppose we have $\|\ZZ\|_2^2 \leq c_2 \eps n$. 
	Then we have	
	\[\mathbf{Var}\left[\sum_{i\neq j} y_i \AA_{i,j} y_j \delta_i \delta_j	\right] \leq  O\left(\frac{k^2}{n^2} +  \frac{c_1^2 k^4 \eps^2}{n^2} +  \frac{(c_1 +c_2 )\eps  k^3}{n^2}	+ \frac{\alpha^2 k^3}{n^2} \right) 	\]
}\smallskip \smallskip \smallskip

\begin{proof}
	We proceed as in Lemma \ref{lem:var}, except that we may remove the terms with $c_{i,j}$ for $i=j$, yielding
		\begin{equation} 
	\begin{split}
	&\mathbf{Var}\left[\sum_{i\neq j} y_i \AA_{i,j} y_j \delta_i \delta_j	\right] \leq \frac{k^2}{n^2} \sum_{i \neq j} c_{i,j}^2 + \frac{k^2}{n^2} \sum_{i \neq j} c_{i,j} c_{j,i}
 + \frac{k^3}{n^3}\sum_{i\neq j \neq u} c_{i,j} c_{u,j}  \\ 
 & + \frac{k^3}{n^3}\sum_{j\neq i \neq u} c_{i,j} c_{i,u}  + \frac{k^3}{n^3}\sum_{i\neq j \neq u} c_{i,j} c_{j,u}+ \frac{k^3}{n^3}\sum_{j\neq i \neq u} c_{i,j} c_{u,i} \\ 
	& +  \frac{k^3}{n^3} \sum_{i \neq j \neq u} c_{i,i} c_{j,u} + \frac{k^4}{n^4}\sum_{i \neq j \neq v \neq u} c_{i,j} c_{u,v}- \left( \frac{k^2}{n^2}\sum_{i \neq j} y_i \AA_{i,j} y_j \right)^2 \\
	\end{split}
	\end{equation}
As in Lemma \ref{lem:var}, we can cancel the term $\frac{k^4}{n^4}\sum_{i \neq j \neq v \neq u} c_{i,j} c_{u,v}$ with a subterm of $-\left( \frac{k^2}{n^2}\sum_{i \neq j} y_i \AA_{i,j} y_j \right)^2$, and bound the remaining contribution of $-\left( \frac{k^2}{n^2}\sum_{i \neq j} y_i \AA_{i,j} y_j \right)^2$ by individually bounding the other terms in the sum.

First, we can similarly bound the last term by $c_1^2 \eps^2  k^4/n^2$ as needed. Now when bounding  $$\mathcal{T} = \left(\sum_{j\neq i \neq u} c_{i,j} c_{i,u} +\sum_{i\neq j \neq u} c_{i,j} c_{u,j} +\sum_{j\neq i \neq u} c_{i,j} c_{u,i} +\sum_{i\neq j \neq u} c_{i,j} c_{j,u} \right)$$ we first observe that in the proof of  Lemma \ref{lem:var}, we only needed a bound on $\|\ZZ\|_2^2$ to give the bound on $\mathcal{T}$. 
So by assumption, $\|\ZZ\|_2^2 \leq c_2 \eps n$, which gives a total bound of $\frac{c_2 k^3 \eps }{n^2}$ on $\frac{k^3}{n^3}\mathcal{T}$.  

	Also, we bound we bound $\frac{k^3}{n^3} \sum_{i \neq j \neq u} c_{i,i} c_{j,u} $ by 
	
	\begin{equation}
	\begin{split}
	\frac{k^3}{n^3} \sum_{i \neq j \neq u} c_{i,i} c_{j,u} &\leq \frac{k^3}{n^3} \sum_i y_i^2 \AA_{i,i} \sum_{j \neq u \neq i} y_j\AA_{j,u} y_u  \\ 
	& \leq \frac{k^3}{n^3} \sum_i y_i^2 \AA_{i,i} \sum_{j \neq u \neq i} y_j\AA_{j,u} y_u    \\
	& \leq \frac{k^3}{n^3} \sum_i y_i^2 \AA_{i,i} \left(c_1\eps n + O(1)\right) \\
	&\leq \frac{\eps c_1 k^3}{n^2} \\
	\end{split}
	\end{equation}
	which is within our desired upper bound. Finally observe that $\sum_{i ,j }  c_{i,j}^2 \leq \sum_{i,j} y_i^2 y_j^2 = \|y|_2^4 \leq 1$, so $\sum_{i \neq j} c_{i,j}^2 \leq   \sum_{i,  j} c_{i,j}^2  \leq 1$, and also $\sum_{i \neq j} c_{i,j} c_{j,i} \leq  \sum_{i,j} y_i^2 y_j^2  \leq 1$, which bounds their corresponding terms in the variance by $O(k^2/n^2)$, which completes the proof.

\end{proof}

\subsection{Improving the complexity to \texorpdfstring{$\tilde{O}(1/\eps^{2})$}{epsilon square} }\label{sec:improving}

We now demonstrate how to obtain an improved sample complexity of $\tilde{O}(1/\eps^2)$ using different scales of sub-sampling, as well as a careful ``eigenvector switching'' argument. As before, we can assume that $\AA$ is symmetric, and $x = \arg \min_{v \in \R^n, \|v\|_2 \leq 1} v^\top \AA v$ is the smallest eigenvector of $\AA$, so that that $\langle x, \AA x \rangle = \lambda_{\min}(\AA) = - \eps n$. Also recall that our algorithms will not need to explicitly know the value $\eps = \min_{v \in \R^n, \|v\|_2 \leq 1} v^\top \AA v/n$, only a lower bound on it, since when run on smaller $\eps$ our algorithm only samplers larger submatrices. 
By Proposition \ref{prop:spread}, we have $\|x\|_\infty \leq \frac{1}{\eps \sqrt{n}}$. We now partition the coordinates of $x$ into \textit{level-sets}, such that all the coordinates $x_i^2$ within a level set have magnitudes that are close to each other.
\begin{definition}
Given $(\AA,x)$, where $x$ is as defined above, define the base level set $S$ as $S = \{ i \in [n] \; : \; |x_i|^2 \leq \frac{100}{\eps n } \}$, and let $T_a  = \{ i \in [n] \; : \; \frac{100\cdot 2^{a-1}}{\eps n}  \leq |x_i|^2 \leq \frac{ 100 \cdot2^{a}}{\eps n } \}$ for an integer $a \geq 1$.
\end{definition}

 We now break the analysis into two possible cases. In the first case, the coordinates in one of the sets $T_a$ contributed a substantial fraction of the ``negativeness'' of the quadratic form $x^\top_T \AA x$, for some $a$ sufficiently large. Since the sets $T_a$ can become smaller as $a$ increases while still contributing a large fraction of the negative mass, this case can be understood as the negativeness of $x^\top \AA x$ being highly concentrated in a small fraction of the matrix, which we must then find to determine that $\AA$ is not PSD. In the second case, no such contributing $T_a$ exists, and the negative mass is spread out more evenly across the terms in the quadratic form $x^\top \AA x$. If this is the case, we will show that our variance bounds from the prior section can be made to obtain a proof that a single, large sampled principal submatrix $T \subset [n]$ will satisfies $x^\top_T \AA_{T \times T} x_{T}$ will be negative with non-negigible probability. 
 
Formally, we define the two cases as follows:

\paragraph{Case 1:} We have $x_S \AA x_{T_a} + x_{T_a} \AA x_S  \leq -\frac{\eps n }{10\log(1/\eps)}$ for some $a$ such that $2^a  \geq 10^6 \zeta^3$, for some $\zeta = \Theta(\log^2(1/\eps))$ with a large enough constant.

\paragraph{Case 2:} The above does not hold; namely,  we have $x_S \AA x_{T_a} + x_{T_a} \AA x_S  > -\frac{\eps n }{10\log(1/\eps)}$ for every $2^a  \geq 10^6 \zeta^3$.

\subsubsection{Case 1: Varied Subsampling and Eigenvector Switching}
In this section, we analyze the \textbf{Case 1}, which specifies that $x_S \AA x_{T_a} + x_{T_a} \AA x_S  \leq -\frac{\eps n }{10\log(1/\eps)}$ for some $T_a$ such that $2^a  \geq 10^6 \zeta^3$, where $\zeta = \Theta(\log^2(1/\eps))$ is chosen with a sufficiently large constant. Recall here that $x \in \R^n$ is the (unit) eigenvector associated with $\lambda_{\min}(\AA) = - \eps n$. We now fix this value $a$ associated with $T_a$. 
In order to find a principal submatrix $\AA_{T \times T}$ that is not PSD for some sampled subset $T \subset [n]$, we will need to show that $T \cap T_a$ intersects in at least one coordinate.

 As discussed in Section \ref{sec:techlinf}, we will need to switch our analysis from $x$ to a different vector $y$, in order to have $y^\top_T \AA_{T \times T} y_T <  0$ with non-negligible probability conditioned on $|T \cap T_a| \geq 1$.
To construct the appropriate vector $y$, we will first proceed by proving several propositions which bound how the quadratic form $x^\top \AA x$ changes as we modify or remove some of the coordinates of $x$.
For the following propositions, notice that by definition of $T_b$, using the fact that $\|x\|_2^2 \leq 1$, we have that $|T_b| \leq \frac{\eps n}{100 2^{b-1}}$ for any $b \geq 1$, which in particular holds for $b=a$. 

  \begin{proposition}\label{prop:rectangularcontribution}
  Let $\AA \in \R^{n \times n}$ satisfy $\|\AA\|_\infty \leq 1$.	Let $S,T \subset [n]$, and let $v \in \R^n$ be any vector such that $\|v\|_2 \leq 1$. Then $|v_S^\top \AA v_T| \leq \sqrt{|S|\cdot | T|}$
  \end{proposition}
\begin{proof} 
	We have $|v_S^\top \AA v_T| = |\sum_{i \in S} \sum_{j \in T} v_i \AA_{i,j}v_j|\leq  \sum_{i \in S}|v_i| \sum_{j \in T} |v_j| \leq  \sum_{i \in S}|v_i| \|v_T\|_1 \leq \|v_S\|_1\|v_T\|_1\allowbreak \leq \sqrt{|S| |T|}$ as needed.
\end{proof}

  \begin{proposition}\label{prop:rectangularcontribution2}
  Let $\AA \in \R^{n \times m}$ satisfy $\|\AA\|_\infty \leq 1$ for any $n,m$. and let $v \in \R^{n}, u \in \R^{m}$ satisfy $\|u\|_2^2, \|v\|_2^2 \leq 1$. Then 
  \[\sum_{j=1}^m \left(\sum_{i=1}^{n} v_i \AA_{i,j} u_j \right)^2 \leq n\]
  \end{proposition}
  \begin{proof}
  We have $\left(\sum_{i=1}^{n} v_i \AA_{i,j} u_j \right)^2  \leq u_j^2\left(\sum_{i=1}^{n} 
 |v_i|  \right)^2 \leq u_j^2 \|v\|_1^2$, so the sum can be bounded by $\sum_{j=1}^m  u_j^2 \|v\|_1^2 = \|v\|_1^2 \|u\|_2^2 \leq \|v\|_1^2 \leq n$ as needed.

  \end{proof}

  \begin{proposition}\label{prop:Sbound}
  Let $x$ be as defined above.  Then we have  $|\langle x_S, \AA x_S\rangle| \leq 10 \eps n$.
  \end{proposition}
  \begin{proof}
  	Suppose $\langle x_S, \AA x_S\rangle= C\eps n$ for a value $C$ with $|C|>10$. 
  	Note that $|\langle x_{[n] \setminus S}, \AA x_{[n] \setminus S} \rangle| \leq \frac{\eps n}{100}$ by Proposition \ref{prop:rectangularcontribution}, using that $|[n] \setminus S| = |\cup_{b \geq 1} T_b| \leq \frac{\eps n}{100}$ (here we use the fact that at most $\frac{\eps n}{100}$ coordinates of a unit vector can have squared value larger than $\frac{100}{\eps n}$). If $C>0$, then we must have that $(\langle x_S ,\AA x_{[n] \setminus S} \rangle +\langle x_{[n] \setminus S}, \AA x_S\rangle ) \leq - ( C +99/100) \eps n$ for us to have that $\langle x , \AA x\rangle = - \eps n$ exactly. 
  	Thus if $C$ is positive and larger than $10$, it would follow that by setting $v=x_S/2 + x_{[n] \setminus S}$, we would obtain a vector $v$ with $\|v\|_2 \leq 1$ such that $v$ has smaller quadratic form with $\AA$ than $x$, namely with $v^\top \AA v \leq - ( C +99/100) \eps n /2 +  \eps C n/4 + n \eps / 100 < -\eps n$ using that $C > 10$, which contradicts the optimality of $x$ as the eigenvector for $\lambda_{\min}(\AA)$. Furthermore, if $C < -10$, then $x_S^\top \AA x_S < - 10 \eps$, which again contradicts the optimality of $x$. 
  	
  \end{proof}

 Now recall that the total row and column contributions of $i$ are defined as $\RRR_i(x) = \sum_{j \in [n] \setminus i} x_i \AA_{i,j} x_j$ and $\CCC_i(x) = \sum_{j \in [n] \setminus i} x_j \AA_{j,i} x_i$ respectively. In the remainder of the section, we simply write $\RRR_i = \RRR_i(x)$ and $\CCC_i = \CCC_i(x)$. We now define the contribution of $i$ within the set $S \subset [n]$.
\begin{definition}
Let $S \subset [n]$ be as defined above. Then for any $i \in [n]$, define the row and column contributions of $i$ within $S$ as $\RRR_{i}^S = \sum_{j  \in S \setminus i} x_i \AA_{i,j} x_j$ and $\CCC_{i}^S = \sum_{j \in S \setminus i} x_j \AA_{j,i} x_i$ respectively. 
\end{definition}\noindent
Observe from the above definition, we have $\sum_{i \in T_a} (\RRR_{i}^S + \CCC_{i}^S)= (x_S \AA x_{T_a} + x_{T_a} \AA x_S) \leq  - \frac{\eps n}{10 \log(1/\eps)}$, where the inequality holds by definition of Case 1.

     \begin{proposition}\label{propLl1rowcolumnbound}
         We have $\sum_{i \in S} (\RRR_{i}^S + \CCC_{i}^S)^2 \leq 1601 \eps n$. 
     \end{proposition}
     \begin{proof}
         Let $z^S ,z,z^-\in \R^{|S|}$ be vectors defined for $i \in S$ as $z_i^S = \RRR_{i}^S+ \CCC_{i}^S$, $z_i = \RRR_i   + \CCC_i$, and $z^- = z - z^S$. Notice that our goal is to bound $\|z^S\|_2^2 $, which by triangle inequality satisfies $\|z^S\|_2^2 \leq 2\left(\|z\|_2^2 + \|z^-\|_2^2\right)$. First note that 
         \begin{equation}
             \begin{split}
                 \|z^-\|_2^2 &= \sum_{i \in S} \left(\sum_{j \notin S} x_i \AA_{i,j} x_j + \sum_{j \notin S}  x_j \AA_{j,i} x_i \right)^2 \\
                 & \leq 2  \sum_{i \in S} \left(\sum_{j \notin S} x_i \AA_{i,j} x_j \right)^2 + 2  \sum_{i \in S} \left(\sum_{j \notin S}  x_j \AA_{j,i} x_i \right)^2\\
             \end{split}
         \end{equation}
      Using that $[n] \setminus S < \eps n/100$, we have by Proposition \ref{prop:rectangularcontribution2} that $\sum_{i \in S} \left(\sum_{j \notin S} x_i \AA_{i,j} x_j \right)^2 \leq \eps n / 100$, so $\|z^-\|_2^2 \leq \eps n /25$. 
         
         We now bound $\|z\|_1 = \sum_{i \in S} | \RRR_i + \CCC_i|$. By Fact \ref{fact:case1OPT}, we have $\RRR_i + \CCC_i \leq 0$ for all $i \in [n]$, which means that $\|z\|_1  \leq \sum_{i \in [n]} | \RRR_i + \CCC_i| = |2 \langle x , \AA x \rangle - 2\sum_{i \in [n]} \AA_{i,i} (x_i)^2| \leq 2 \eps n$. 
         Next, we bound $\|z\|_\infty$. Notice that  $|\RRR_i + \CCC_i |= 2|  x_i \langle \AA_{i,*}, x \rangle - \AA_{i,i}(x_i)^2 |  \leq 2\eps n (x_i)^2 + 2\AA_{i,i}(x_i)^2 < 4 \eps n  (x_i)^2 $,  using that $\AA x = - \eps n x$, since $x$ is an eigenvector of $\AA$. Since for $i\in S$, it also follows that $(x_i)^2 \leq \frac{100}{\eps n}$, thus $\|z\|_\infty \leq 400$ as needed. It follows that $\|z\|_2^2$ is maximized by having $ 2 \eps n/400$ coordinates equal to $400$, giving  $\|z\|_2^2 \leq  2 \eps n/400 (400)^2 = 800 \eps n$. It follows then that $\|z^S\|_2^2 \leq 1601 \eps n$ as needed.
     \end{proof}

\paragraph{Eigenvector Setup: } We now define the ``target'' direction $y \in \R^n$ which we will use in our analysis for \textbf{Case 1}. First, we will need the following definitions. Let 
\[ D_a^p = \left\{ t \in T_a \; : \;   -\frac{ 2^{p+1} 2^{a}}{  \log(1/\eps)}  \leq \RRR_t^S + \CCC_t^S \leq -\frac{ 2^p 2^{a}}{  \log(1/\eps)} \right\}\]
Define the \textit{fill} $\beta$ of $T_a$ as the value such that $\beta = 2^{-p}$ where $p \geq 1$ is the smallest value of $p$ such that $-|D_a^p|(\frac{ 2^p 2^{a}}{  \log(1/\eps)}) \leq - \frac{\eps n}{40 \log^2(1/\eps)}$. Note that at least one such $p$ for $1 \leq p \leq \log(1/\eps)$ must exist. let $T_a^* = D_a^p$ where $\beta = 2^{-p}$. Observe that $x_S \AA x_{T_a^*} + x_{T_a^*} \AA x_S \leq- \frac{\eps n}{40 \log^2(1/\eps)}$.
Finally, we define our target ``eigenvector'' $y$ as
\begin{equation}\label{eqn:eigen}
   y = x_S + \zeta\beta \left(2^{-a} \cdot  x_{T_a}\right) 
\end{equation}
 where $\zeta  = \Theta(\log^2(1/\eps))$ is as above,  and we also define our target submatrix subsampling size as $\lambda  = \frac{2000\beta^2 \zeta^2 \log(1/\eps)}{2^a \eps}$. First, we prove that for a random submatrix $\AA_{T \times T}$, where $i \in [n]$ is sampled and added to $T$ with probability $\lambda/n$, we have that $y^\top \AA_{T \times T} y$ is negative in expectation conditioned on $|T \cap T_a^*| \geq 1$.

 \begin{lemma}\label{lem:conditionalex}
 	Suppose we are in case $1$ with $T_a$ contributing such that $2^a \geq 10^6 \zeta^3$.
 	Let $\delta_i$ be an indicator variable that we sample coordinate $i$, with $\ex{\delta_i} = \frac{\lambda}{n}$ and $\lambda  = \frac{2000\beta^2 \zeta^2 \log(1/\eps)}{2^a \eps}$. Then if $y = x_S + \zeta\beta (2^{-a} x_{T_a})$ where $\zeta  \geq 100 \log^2(1/\eps)$, and if $t \in T_a^*$, then
 	\[\mathbb{E} \left[ \sum_{i,j \in S \cup \{t\}} y_i \AA_{i,j} y_j \delta_i \delta_j \; \big| \; \delta_t = 1 \right] \leq - \frac{50 \zeta \lambda }{ n} \]
 \end{lemma}
\begin{proof}
 First observe $\ex{ \sum_{i \in S} y_i^2 \AA_{i,i} \delta_i} \leq \frac{\lambda}{n}\|x\|_2^2 \leq \frac{\lambda}{n}$ since $x$ is a unit vector. Note that  $y_S = x_S$ by construction, so we can use Proposition \ref{prop:Sbound} to bound $|\langle y_S , \AA y_S\rangle |$ by $10 \eps n$, which gives	
	\begin{equation}
	\begin{split}
	\mathbb{E} \left[ \sum_{i,j \in S \cup \{t\}} y_i \AA_{i,j} y_j \delta_i \delta_j \; \big| \; \delta_t = 1 \right] &\leq \frac{\lambda^2}{n^2}\left(\langle y_S , \AA y_S  \rangle- \sum_{i \in S} y_i^2 \AA_{i,i} \delta_i\right)+ \frac{\lambda}{n} + \frac{\lambda}{n}\sum_{i \in S} (\AA_{t,i} + \AA_{i,t}) y_i y_t  +y_t^2 \\ 
		& \leq \frac{20 \lambda^2 \eps }{n} +\frac{\lambda }{n}\left(1 +  \frac{ \zeta \beta }{2^{a}} \left(\RRR_t^S + \CCC_t^S\right) \right)+ (\frac{\zeta \beta}{2^a})^2 (\frac{100 2^a}{\eps n})\\
			& \leq \frac{20 \lambda^2 \eps }{n} +\frac{\lambda }{n}\left(1 +  \frac{ \zeta \beta }{2^{a}} \left(\RRR_t^S + \CCC_t^S\right) \right)+ \frac{100\zeta^2 \beta^2}{2^a \eps n}\\
	\end{split}
	\end{equation}
	Now by definition of $i \in T_a^*$, we have $\left(\RRR_t^S + \CCC_t^S\right) \leq- \frac{2^a}{\beta \log(1/\eps)}$. 
Thus
	\begin{equation}
	\begin{split}
\mathbb{E} \left[ \sum_{i,j \in S \cup \{t\}} y_i \AA_{i,j} y_j \delta_i \delta_j \; \big| \; \delta_t = 1  \right]& \leq \frac{20 \lambda^2 \eps }{n} +\frac{\lambda }{n}\left(1 -\frac{\zeta }{ \log(1/\eps)}  \right)+ \frac{ 100\zeta^2 \beta^2}{2^a \eps n}\\
\end{split}
	\end{equation}
	Setting $\zeta > 100 \log^2(1/\eps)$, we first note that $\frac{\lambda }{n}\left(1 -\frac{\zeta }{   \log(1/\eps)}  \right) \leq -\frac{99\lambda \zeta }{100n \log(1/\eps)}$. Since $$\lambda = \frac{2000\beta^2 \zeta^2 \log(1/\eps)}{2^a \eps} \leq \frac{2000\beta^2 }{2^6 \zeta \eps } \leq \frac{1}{\ \zeta \eps}$$ it follows that $\frac{20 \lambda^2 \eps}{n} \leq \frac{20 \lambda }{n \zeta} \leq  \frac{20 \lambda  }{n} < \frac{\lambda \zeta }{ 5\log(1/\eps) n}$. Thus $\frac{10 \lambda^2 \eps}{n}  -\frac{99\lambda \zeta }{100n \log(1/\eps)} \leq -\frac{3\lambda \zeta }{4n \log(1/\eps)}$. So we can simply and write 
		\begin{equation}
		\begin{split}
	\mathbb{E} \left[ \sum_{i,j \in S \cup \{t\}} y_i \AA_{i,j} y_j \delta_i \delta_j \; \big| \; \delta_t = 1 \right]  &\leq -\frac{3\lambda \zeta }{4n \log(1/\eps)}+ \frac{100 \zeta^2 \beta^2}{2^a \eps n}\\ 
		& =-\frac{1500\beta^2 \zeta^3 }{2^a\eps n }+ \frac{100 \zeta^2 \beta^2}{2^a \eps n}\\
			& \leq -\frac{1400\beta^2 \zeta^3 }{2^a\eps n }\\
		&	\leq - \frac{50 \zeta \lambda }{n} \\  
\end{split}
	\end{equation}
	as desired.
\end{proof}

  \begin{lemma}\label{lem:Spart}
  	Let $\delta_i$ be an indicator variable with $\ex{\delta_i} = \lambda/n$. Then 
	\[\mathbf{Pr}\left[	\left|\sum_{i , j\in S} y_i \AA_{i,j} y_j \delta_i \delta_j - \bex{\sum_{i , j\in S} y_i \AA_{i,j} y_j \delta_i \delta_j}	\right| \geq C\frac{\lambda }{n} \right] \leq \frac{49}{50}\]
		Where $C>0$ is some constant.
  \end{lemma}
  \begin{proof}

  We can apply Corollary \ref{cor:var}, where we can set the values of $c_1,c_2$ to be bounded by constants by the results of Proposition \ref{prop:Sbound} and \ref{propLl1rowcolumnbound}, and by definition of the set $S$ we can set $\alpha \leq \sqrt{\eps}$ for the $\alpha$ in Corollary \ref{cor:var}, and using that $\lambda \leq O(1/\eps)$, we obtain:
  	\[ \mathbf{Var} \left[		\sum_{i \neq j\in S} y_i \AA_{i,j} y_j \delta_i \delta_j 		\right] \leq	\frac{C\lambda^2}{n^2} \]
  for some constant $C$.	Now note that $\ex{ 	\sum_{i \neq j\in S} y_i \AA_{i,j} y_j \delta_i \delta_j} \leq \frac{30 \lambda^2\eps }{n} \leq O(\frac{\lambda}{n})$, thus by Chebyshev's, with probability $99/100$, we have $|\sum_{i \neq j\in S} y_i \AA_{i,j} y_j \delta_i \delta_j| \leq O( \frac{ \lambda}{n})$. Moreover, note that $\sum_i \AA_{i,i}y_i^2 \delta_i$ can be assumed to be a positive random variable using that $\AA_{i,i} =1$, and note the expectation of this variable is $\frac{\lambda}{n}$, and is at most $100\lambda/n$ with probability $99/100$. Thus $|\sum_{i \in S} y_i^2 \AA_{i,i} \delta_i - \ex{\sum_{i \in S} y_i^2 \AA_{i,i}  }| \leq 100 \lambda/n$. By a union bound, we have:
		  	\[\mathbf{Pr}\left[	\left|\sum_{i , j\in S} y_i \AA_{i,j} y_j \delta_i \delta_j	\right| \geq C\frac{\lambda }{n} \right] \leq \frac{49}{50}\]
		  	Where $C = 150$.
  \end{proof}

  \begin{lemma}\label{lem:tpart}
  	Fix any $t \in T_a^*$. Then 
  	\[\mathbf{Pr}\left[		\left|\sum_{i \in S}y_t(y_i \AA_{i,t} +  \AA_{t,i} y_i) \delta_i - \ex{	\sum_{i \in S}y_t(y_i \AA_{i,t} +  \AA_{t,i} y_i) \delta_i} \right|\geq \frac{10\lambda}{n} \right] \leq \frac{1}{100}\]	
  \end{lemma}
\begin{proof}By independence of the $\delta_i$'s
	\begin{equation}
	\begin{split}
	 \bvar{	\sum_{i \in S} (y_i \AA_{i,t} y_t + y_t \AA_{t,i} y_i) \delta_i } & \leq \frac{\lambda}{n} \sum_{i } (y_i \AA_{i,t} y_t + y_t \AA_{t,i} y_i)^2 \\
	  & \leq   \frac{\lambda}{n} \sum_{i }2 y_t^2 y_i^2	\\ &\leq \frac{2\lambda}{n} (\zeta \beta 2^{-a})^2 \frac{100 2^a}{\eps n} \\
	 & 	\leq \frac{2\lambda}{n} \left(\frac{100 \zeta^2 \beta^2}{2^a\eps n} \right) \\ 
	 & \leq \frac{2\lambda^2 }{5n^2} \\
	\end{split}
	\end{equation}

Now since $\ex{	\sum_{i \in S}(y_i \AA_{i,t} y_t + y_t \AA_{t,i} y_i) \delta_i} \leq - \frac{\lambda}{n}\frac{\zeta}{\log(1/\eps)} \leq -\frac{100\lambda}{n}$, the desired result follows by Chebyshev's inequality.
\end{proof}

\begin{lemma}
	Fix any $t \in T_a^*$. Then 
\[\bpr{	 \sum_{i,j \in S \cup \{t\}} y_i \AA_{i,j} y_j \delta_i \delta_j 	\leq \frac{-25\zeta \lambda}{n}	\; \big| \; \delta_t = 1} \geq 24/25\]
\end{lemma}
\begin{proof}
Conditioned on $\delta_t = 1$, we have 
	\begin{equation}
	\begin{split}
&\left|	\left[ \sum_{i,j \in S \cup \{t\}} y_i \AA_{i,j} y_j \delta_i \delta_j \right] - \bex{\sum_{i,j \in S \cup \{t\}} y_i \AA_{i,j} y_j \delta_i \delta_j \; \big| \; \delta_t = 1| }  \right| \\
& \leq 	\left|\sum_{i \in S}y_t(y_i \AA_{i,t} +  \AA_{t,i} y_i) \delta_i - \ex{	\sum_{i \in S}y_t(y_i \AA_{i,t} +  \AA_{t,i} y_i) \delta_i} \right| \\
& + \left|\sum_{i , j\in S} y_i \AA_{i,j} y_j \delta_i \delta_j - \bex{\sum_{i , j\in S} y_i \AA_{i,j} y_j \delta_i \delta_j}	\right| \\
& \leq C \frac{\lambda}{n}
	\end{split}
	\end{equation}
	for some constant $C \leq 200$, where the last fact follows from Lemmas \ref{lem:Spart} and \ref{lem:tpart} with probability $24/25$. Since $\bex{\sum_{i,j \in S \cup \{t\}} y_i \AA_{i,j} y_j \delta_i \delta_j \; \big| \; \delta_t = 1| } \leq - \frac{50 \zeta \lambda }{n}$ by Lemma \ref{lem:conditionalex}, by scaling $\zeta$ by a sufficiently large constant, the result follows. 
\end{proof}

\begin{theorem}\label{thm:case1}
		Suppose we are in case $1$ with $T_a$ contributing such that $2^a \geq 10^6 \zeta^3$. Then there is an algorithm that queries at most $O(\frac{\log^7(1/\eps)}{\eps^2})$ entries of $\AA$, and finds a principal submatrix of $\AA$ which is not PSD with probability at least $9/10$ in the NO case. The algorithm always returns YES on a YES instance.
\end{theorem}
\begin{proof}
	By the above, we just need to sample a expected size $O(\lambda^2)$ submatrix from the conditional distribution of having sampled at least one entry from $T_a^*$. 
	 Since $|T_a^*| \geq \beta/10 \frac{\eps n}{2^a \log(1/\eps)}$, and since $\lambda = \Theta(\frac{\beta^2 \zeta^2 \log(1/\eps)}{2^a \eps})$, we see that this requires a total of $k$ samples of expected size $O(\lambda^2)$ , where 
	\begin{equation}
	\begin{split}
		k &= (n/|T_a^*|) /\lambda \leq  (\frac{2^a 10 \log(1/\eps)}{ \beta \eps }) (\frac{2^a \eps }{\beta^2 \zeta^2 \log(1/\eps)}	) \\
		&\leq 10 \frac{2^{2a}}{ \beta^3 \zeta^2 }
	\end{split}
	\end{equation}

	Thus the total complexity is $O(k \lambda^2)$, and we have
	\begin{equation}
	\begin{split}
		k \lambda^2 &\leq  10 \frac{2^{2a}}{ \beta^3 \zeta^2 }	(\frac{\beta^4 \zeta^4 \log^2(1/\eps)}{2^{2a} \eps^2}) \\
		&\leq 10 \frac{\beta \zeta^2 \log^2(1/\eps )}{\eps^2} \\
		& = O(\frac{\zeta^2 \log^2(1/\eps)}{\eps^2})
	\end{split}
	\end{equation}
 we use the fact that we can set $\zeta = O(\log^2(1/\eps))$. Finally, note that we do not know $\beta$ or $2^a$, but we can guess the value of $\lambda$ in powers of $2$, which is at most $O(\frac{\zeta^2}{\eps^2})$, and then set $k$ to be the value such that $k \lambda^2$ is within the above allowance. This blows up the complexity by a $\log(1/\eps)$ factor to do the guessing.

\end{proof}

  \subsubsection{Case 2: Spread Negative Mass and Main Theorem}
  In the prior section, we saw that if the quadratic form $x^T \AA x$ satisfies the condition for being in Case 1, we could obtain a $\tilde{O}(1/\eps^2)$ query algorithm for finding a principal submatrix $A_{T \times T}$ such that $y^\top \AA_{T \times T} y < 0$ for some vector $y$. Now recall that  $S = \{ i \in [n] \; : \; |x_i|^2 \leq \frac{1}{\eps n } \}$, and let $T_a  = \{ i \in [n] \; : \; \frac{100 2^{a-1}}{\eps n}  \leq |x_i|^2 \leq \frac{ 100 2^{a}}{\eps n } \}$ for $a \geq 1$. Recall that the definition of Case $1$ was that $x_S^\top \AA x_{T_a} + x_{T_a}^\top \AA x_S  \leq -\eps n /(10\log(1/\eps))$ for some $2^a \geq 10^6 \zeta^3$. In this section, we demonstrate that if this condition does not hold, then we will also obtain a  $\tilde{O}(1/\eps^2)$ query algorithm  for the problem.

  Thus, suppose now that we are in Case $2$; namely that  $x_S \AA x_{T_a} + x_{T_a} \AA x_S  > -\eps n /(10\log(1/\eps))$ for all $2^a \geq 10^6 \zeta^3$. Now let $T^+ = \cup_{2^a > 10^6 \zeta^3} T_a$ and let $T^- = \cup_{2^a \leq 10^6 \zeta^3} T_a$. Let $S^* = S \cup T^-$. We now observe an important fact, which sates that if we are not in Case 1, then  $x_{S^*} \AA x_{S^*}$ contributes a substantial fraction of the negativeness in the quadratic form.
  
  \begin{fact}\label{fact:case2}
  Suppose we are in Case $2$: meaning that $x_S^\top \AA x_{T_a} + x_{T_a} \AA x_S > -\eps n /(10\log(1/\eps))$ for all $2^a \geq 10^6 \zeta^3$.  Then we have $x_{S^*}^\top \AA x_{S^*} \leq - \eps n/2$. 
  \end{fact}
  \begin{proof}
      Notice that this implies that $x_S^\top\AA x_{T^+} + x_{T^+} \AA x_S \geq - \eps n / 10$, since there are at most $\log(1/\eps)$ level sets included in $T^+$ by Proposition \ref{prop:spread}. Note since the contribution of $|x_{T^+}^\top \AA x_{T^+}| \leq - \eps n 10^{-6}/ \zeta^3$ and $|x_{T^-} \AA x_{T^+} + x_{T^+} \AA x_{T^-}| \leq \sqrt{|T^-| |T^+|} \leq \eps n /100$ by Proposition \ref{prop:rectangularcontribution}. Thus if $x^\top\AA x \leq - \eps n$ to begin with, it follows that we must have 
      \begin{equation}
          \begin{split}
            x_{S^*}^\top \AA x_{S^*} &\leq x^\top A x -\left( (x_S^\top\AA x_{T^+} + x_{T^+} \AA x_S)- (x_{T^+}^\top \AA x_{T^+}) - (x_{T^-} \AA x_{T^+} + x_{T^+} \AA x_{T^-} ) \right)\\
            & \leq - \eps n +\eps n/10 + + \eps n 10^{-6}/ \zeta^3 \eps n /100 \\
            &< - \eps n/2       \\  
          \end{split}
      \end{equation}
    \end{proof}

  We now proceed by analyzing the result of sampling a principal submatrix from the quadratic form $x_{S^*}^\top \AA x_{S^*}$, which by the prior fact is already sufficently negative. Specifically, we will demonstrate that the variance of the standard estimator from Lemma \ref{lem:var}, and specifically Corollary \ref{cor:var}, is already sufficiently small to allow for a single randomly chosen $O(1/\eps) \times O(1/\eps)$ principal submatrix of $\AA$ to have negative quadratic form with $x_{S^*}$ with good probability. In order to place a bound on the variance of this estimator and apply Corollary \ref{cor:var}, we will need to bound the row and column contributions of the quadratic form $x_{S^*}^\top \AA_{S^* \times S^*}x_{S^*}$, which we now formally define.
\begin{definition}
For $i \in [n]$, define the row and column contributions of $i$ within $S^*$ as $\RRR_i^*= \sum_{j  \in S^* \setminus i} x_i \AA_{i,j} x_j$ and $\CCC_i^* = \sum_{j \in S^* \setminus i} x_j \AA_{j,i} x_i$ respectively. 
\end{definition}\noindent
Recall that the \textit{total} row and column contributions of $i$ are defined via $\RRR_i = \sum_{j \in [n] \setminus i} x_i \AA_{i,j} x_j$ and $\CCC_i = \sum_{j \in [n] \setminus i} x_j \AA_{j,i} x_i$ respectively, and recall that we have $\RRR_i + \CCC_i \leq 0$ for all $i \in [n]$ by Fact \ref{fact:case1OPT}

     \begin{proposition}\label{propLl1rowcolumnboundCase2}
         We have $\sum_{i \in S^*} (\RRR_i^* + \CCC_i^*)^2 \leq 10^9 \cdot \zeta^3 \eps n$.
     \end{proposition}
     \begin{proof}
     The proof proceeds similarly to Proposition \ref{propLl1rowcolumnbound}.
         Let $z^* ,z,z^-\in \R^{|S^*|}$ be defined for $i \in S^*$ via $z_i^* = \RRR_i^* + \CCC_i^*$, $z_i = \RRR_i   + \RRR_i$, and $z^- = z - z$. Notice that our goal is to bound $\|z^*\|_2^2 $, which by triangle inequality satisfies $\|z^*\|_2^2 \leq 2\left(\|z\|_2^2 + \|z^-\|_2^2\right)$. First note that 
         \begin{equation}
             \begin{split}
                   \|z^-\|_2^2 &= \sum_{i \in S^*} \left(\sum_{j \notin S^*} x_i \AA_{i,j} x_j + \sum_{j \notin S^*}  x_j \AA_{j,i} x_i \right)^2 \\
                   &\leq 2  \sum_{i \in S^*} \left(\sum_{j \notin S^*} x_i \AA_{i,j} x_j \right)^2 + 2  \sum_{i \in S^*} \left(\sum_{j \notin S^*}  x_j \AA_{j,i} x_i \right)^2\\
             \end{split}
         \end{equation}
        Using that $|[n] \setminus S^*| < \eps n/100$, we have by Proposition \ref{prop:rectangularcontribution2} that $\sum_{i \in S^*} \left(\sum_{j \notin S} x_i \AA_{i,j} x_j \right)^2 \leq \eps n / 100$, so $\|z^-\|_2^2 \leq \eps n /25$. 
         
         We now bound $\|z\|_1 = \sum_{i \in S} | \RRR_i + \RRR_i|$. Recall that we have $\RRR_i + \RRR_i \leq 0$ for all $i \in [n]$, which means that $\|z\|_1  \leq \sum_{i \in [n]} | \RRR_i + \RRR_i| = |2 \langle x , \AA x \rangle - 2\sum_{i \in [n]} \AA_{i,i} (x_i)^2| \leq 2 \eps n$.
         Next, we bound $\|z\|_\infty$. Notice that  $|\RRR_i + \RRR_i |= 2|  x_i \AA_{i,*} x|  = 2\eps n (x_i)^2 - 2\AA_{i,i}(x_i)^2 < 4 \eps n  (x_i)^2 $,  using that $\AA x = - \eps n x$, since $x$ is an eigenvectror of $\AA$. Since $ i\in S^*$, by definition we have $(x_i)^2 \leq \frac{100 \cdot 10^6 \cdot \zeta^3}{\eps n}$, thus $\|z\|_\infty \leq 100 \cdot 10^6 \cdot \zeta^3$. It follows that $\|z\|_2^2$ is maximized by having $ 2 \eps n/( 10^8 \cdot \zeta^3)$ coordinates equal to $ 10^8 \cdot \zeta^3$, giving  $\|z\|_2^2 \leq  2 \eps n/( 10^8 \cdot \zeta^3) ( 10^8 \cdot \zeta^3)^2 = 2 \cdot 10^8 \cdot \zeta^3 \eps n$. It follows then that $\|z\|_2^2 \leq  \cdot 10^9 \cdot \zeta^3 \eps n$ as needed.
     \end{proof}

\begin{theorem}\label{thm:inftymain}
	There is an algorithm which, given $\AA$ with $\|\AA\|_\infty \leq 1$ such that either $x^\top \AA x \geq 0$ for all $x \in \R^{n}$ (\texttt{YES} Case), or $x^\top\AA x \leq - \eps n$ for some $x \in \R^n$ with $\|x\|_2 \leq 1$ (\texttt{NO} Case), distinguishes the two cases with probability $3/4$ using at most $\wt{O}(\frac{1}{\eps^2})$ queries, and running in time $\tilde{O}(1/\eps^\omega)$, where $\omega< 2.373$ is the exponent of fast matrix multiplication. Moreover, in the \texttt{YES} case the, the algorithm always outputs \texttt{YES} (with probability $1$), and in the \texttt{NO} case, the algorithm returns a certificate in the form of a principal submatrix which is not PSD. 
\end{theorem}
   \begin{proof}
   	By Theorem \ref{thm:case1} which handles Case 1, we can restrict ourselves to Case 2. Using Fact \ref{fact:case2} as well as Proposition \ref{propLl1rowcolumnboundCase2}, can  apply Corollary \ref{cor:var} with the vector $y = x_{S^*}$, setting $c_1 = \Theta(1)$ and $c_2 = \Theta(\zeta^3)$, and $\alpha = O(\sqrt{\eps} \zeta^3) = O(\sqrt{\eps} \log^6(1/\eps))$, to obtain that $$\mathbf{Var}[\sum_{i \neq j} y_i \AA_{i,j} y_j \delta_i \delta_j] \leq O( \log^{12}(1/\eps) \frac{k^2}{n^2})$$ where $k = \tilde{\Theta}(1/\eps)$. 
   	Since by Proposition \ref{prop:exp} and Fact \ref{fact:case2} , we have $\ex{\sum_{i \neq j} y_i \AA_{i,j} y_j \delta_i \delta_j} \leq \frac{k^2}{4 n^2} \langle x_{S^*},\AA x_{S^*}\rangle \leq -\frac{\eps k^2}{8n}$, it follows that by repeating the sampling procedure $O(\log^{12}(1/\eps))$, by Chebyshev's we will have that at least one sample satisfies  $\sum_{i \neq j} y_i \AA_{i,j} y_j \delta_i \delta_j \leq - \frac{\eps k^2}{4n}$ with probability $99/100$.
   	
   	Now note that this random variable does not take into account the diagonal. Thus, it will suffice to bound the contribution of the random variable $ \sum_{i \in [n]} \delta_i \AA_{i,i}(y_i)^2$ $\tilde{O}((1/\eps)/n)$. First observe that $\ex{ \sum_{i \in [n]} \delta_i \AA_{i,i}(y_i)^2 } = \frac{k}{n}$. The proof proceeds by a simple bucketing argument; let $\Lambda_i = \{i \in S^* \; | \; \frac{2^i}{n} \leq(y_i)^2 \leq   \frac{2^{i+1}}{n} \}$, and for a single $k \times k$ sampled submatrix, let $T \subset [n]$ be the rows and columns that are sample. Note that $\ex{| T \cap \Lambda_i|} \leq k 2^{-i}$, since $|\Lambda_i| \leq 2^{-i}$. Note also that $|\Lambda_i| = 0$ for every $i$ such that $2^i \geq \frac{100^8 \zeta^3 }{\eps}$ by definition of $S^*$ and the fact that $y$ is zero outisde of $S^*$. Then by Chernoff bounds we have that with probability $\pr{| T \cap \Lambda_i| > \log(1/\eps) \max\{ k 2^{-i} ,1\} } \leq  	1-\frac{\eps^{10}}{C}$ for some constant $C$ for our choosing. We can then union bound over all $O(\log(1/\eps))$ sets $\Lambda_i$, to obtain 
   	\[\sum_{i \in [n]} \delta_i \AA_{i,i}(y_i)^2 \leq \sum_{i : 2^i \leq \frac{100^8 \zeta^3 }{\eps} } \frac{2^{i+1}}{n} | T \cap \Lambda_i|  \leq  \sum_{i : 2^i \leq \frac{100^8 \zeta^3 }{\eps} } \frac{2}{n}\log(1/\eps) \max\{ k , 2^i\}\] with probability at least $	1-\frac{\eps^9}{C}$. Setting $k = \Theta(\log^6(1/\eps)/\eps)$, we have that $\sum_{i \in [n]} \delta_i \AA_{i,i}(y_i)^2  \leq  \sum_{i : 2^i \leq \frac{100^8 \zeta^3 }{\eps} }  (2/n)\log(1/\eps) k = O(\log^2(1/\eps)k/n)$ .
 Thus we can condition on $ \sum_{i \in [n]} \delta_i \AA_{i,i}(y_i)^2=  O(\log^2(1/\eps)k/n)$  for all $\tilde{O}(1)$ repetitions of sampling a submatrix. Since at least one sampled submatrix satisfied $\sum_{i \neq j} y_i \AA_{i,j} y_j \delta_i \delta_j \leq - \frac{\eps k^2}{4n}$, and since  $k = \Theta(\log^6(1/\eps)/\eps)$, this demonstrates that at least one sampled submatrix will satisfy  $\sum_{i ,j} y_i \AA_{i,j} y_j \delta_i \delta_j < - \frac{\eps k^2}{8n}$ as needed in the NO instance. The resulting query complexity is then $O(\log^2 (1/\eps)k^2) = O(\frac{\log^{24}(1/\eps)}{\eps^2}) = \wt{O}(\frac{1}{\eps^2})$ as desired.  Finally, for runtime, notice that the main computation is computing the eigenvalues of a $k \times k$ principal submatirx, for $k = \tilde{O}(1/\eps)$, which can be carried out in time $\tilde{O}(1/\eps^\omega)$ \cite{demmel2007fast2,banks2019pseudospectral}.
   \end{proof}

\section{PSD Testing with \texorpdfstring{$\ell_2^2$}{L-2-square}  Gap}
Let $\AA \in \R^{n \times n}$ be a symmetric matrix with eigenvalues $\lambda_{\max} = \lambda_1 \geq \lambda_2 \geq \dots \geq \lambda_n = \lambda_{\min}.$   In this section, we consider the problem of testing positive semi-definiteness with an $\ell_2^2$ gap. Formally, the problem statement is as follows.

\begin{definition}[PSD Testing with $\ell_{2}^2$-Gap.]\label{def:l2}
Fix, $\eps \in (0,1]$ and let $\AA \in \R^{n \times n}$ be a symmetric matrix satisfying $\|\AA\|_\infty \leq 1$, with the promise that either 
\begin{itemize}
	\item \textbf{YES Instance}: $\AA$ is PSD.
	\item \textbf{NO Instance}: $\AA$ is $\eps$-far from PSD in $\ell_2^2$, meaning that $\min_{\BB \succeq 0} \|\AA - \BB\|_F^2 = \sum_{i: \lambda_i < 0} \lambda_i^2 = \eps n^2$.
\end{itemize}
The PSD Testing problem with $\ell_{2}^2$-gap is to design an algorithm which distinguish these two cases with probability at least $2/3$, using the minimum number of queries possible to the entires of $\AA$.
\end{definition}

Our algorithm for this problem will query a principal submatrix $\AA_{S  \times S}$ and return PSD if $\AA_{S  \times S}$ is  PSD, otherwise it will return not PSD. Since all principal submatrices of PSD matrices are PSD, we only need show that if $\AA$ is $\eps$-far from PSD, then we can find a non-PSD principal submatrix with small size. Note again that this implies that our algorithm will be one-sided. Thus, in the following, we can focus on the case where $\AA$ is $\eps$-far from PSD.  We begin by stating two fundamental observations, which, along with an application of our algorithm from Section \ref{sec:Linfty}, will allow us to reduce the problem of PSD testing with $\ell_2$ gap to the problem of testing certain functions of the \textit{singular values} of $\AA$.

\begin{proposition}[PSD matrices are top heavy]\label{prop:topheavy} Fix any $n  \in \mathbb{N}$, $1 \leq k \leq n$, and $\DD \in \R^{n \times n}$. Then if $\DD$ is PSD, we have
	\[\sum_{i > k} \sigma_i(\DD)^2 \leq \frac{1}{k}\left( \Tr(\DD)\right)^2\] 
In particular, if $\DD$ has bounded entries $\|\DD\|_\infty \leq 1$, we have $\sum_{i > k} \sigma_i(\DD)^2 \leq \frac{1}{k}n^2$.
\end{proposition}
\begin{proof}
	We first show that $\sigma_k(\DD) \leq k^{-1} \Tr(\DD)$. To see this, suppose $\sigma_{k}(\DD) > k^{-1} \Tr(\DD)$. Then because $\DD$ is PSD, we would have $\sum_i \sigma_i = \sum_i \lambda_i  = \text{Tr}(\AA) > k \cdot k^{-1} \Tr(\DD)$, a contradiction. Thus, $\sigma_{i}(\DD) \leq k^{-1} \Tr(\DD)$ for all $i \geq k$. Using this and the bound  $\sum_{i > k} \sigma_i(\DD) \leq \Tr(\DD)$, it follows that the quantity $\sum_{i > k} \sigma_i(\DD)^2$ is maximimized by having $k$ singular values equal to $\Tr(\DD)/k$, yielding $\sum_{i > t} \sigma_i(\DD)^2 \leq k \cdot (\Tr(\DD)/k)^2 =  k^{-1} (\Tr(\DD))^2 $ as needed. 
\end{proof}

\begin{proposition}\label{prop:nottopheavy} 
	Let $\DD \in \R^{n \times n}$ be a symmetric matrix such that $\|\DD\|_\infty \leq 1$, and let $\sigma_1 \geq \sigma_2 \geq \dots \geq \sigma_n$ be its singular values. Suppose $\DD$ is at least $\eps$-far in $L_2$ from PSD, so that $\sum_{i : \lambda_i(\DD) < 0} \lambda_i^2(\DD) \geq   \eps n^2$, and suppose further that $\min_i \lambda_i(\DD) > - \frac{1}{2k} n$ for any $k \geq \frac{2}{\eps}$. Then we have
	\[\sum_{i >k}\sigma_{i}^2(\DD) > \frac{\eps}{2} n^2\] 
\end{proposition}
\begin{proof}
	Let $W \subseteq [n]$ be the set of values $i \in [n]$ such that $\lambda_{i} < 0$.  Let $W' \subseteq [n]$ be the set of values $i \in [n]$ such that $\sigma_{i} < \frac{1}{2k} n$. By assumption: $\sum_{i \in W'} \sigma_{i}^2 \geq  \sum_{i \in W} \lambda_{i}^2 \geq  \eps n^2$.  Now  $\sum_{i \in W'} \sigma_{i}^2 = \sum_{i \in W', i \leq k} \sigma_{i}^2 + \sum_{i \in W', i > k} \sigma_{i}^2$, so the fact that $|\sigma_i| \leq (1/2k)n$ for every $i \in W'$, we have that  $ \sum_{i \in W', i \leq k} \sigma_{i}^2 \leq k (n/(2k))^2 = n^2/4k < \eps n^2/2$. Thus we must have $\sum_{i \in W, i > t} \sigma_{i}^2 > \eps n^2/2 $, giving
	
	\begin{equation}
	\begin{split}
	\sum_{i > k} \sigma_{i}^2  & \geq \sum_{i \in W', i > k} \sigma_{i}^2 \\
	& > \eps  n^2/2 \\
	\end{split}
	\end{equation}
	as required.
\end{proof}
\noindent

\subsection{Analysis of the Algorithm}
Our analysis will require several tools, beginning with the following interlacing lemma. 



\begin{lemma}[Dual Lidskii Inequality,  \cite{tao2011topics} Chapter 1.3]\label{lem:duallidskii}
	Let $\M_1,\M_2$ be $t \times t$ symmetric Matrices, and fix $1 \leq i_1 < i_2 < \dots < i_k \leq n$. Then we have
	\[\sum_{j=1}^{k} \lambda_{i_j} (\M_1 + \M_2) \geq \sum_{j=1}^{k}\lambda_{i_j} (\M_1 )  +  \sum_{j=1}^{k} \lambda_{n-j+1}(\M_2)		\]
	
\end{lemma}\noindent
 We will also need the following result of Rudelson and Vershynin \cite{rudelson2007sampling} on the decay of spectral norms of random submatrices. 
\begin{proposition}[\cite{rudelson2007sampling}]\label{prop:vershy}
	Let $\AA \in \R^{n \times m}$ be a rank $r$ matrix with maximum Euclidean row norm bounded by $M$, in other words $\max_i |(\AA \AA^\top)_{i,i}| \leq M$. Let $Q \subset [n]$ be a random subset of rows of $\AA$ with expected cardinality $q$. Then there is a fixed constant $\kappa \geq 1$ such that 
	
	\[  \ex{\|\AA_{Q \times [m] }\|_2 } \leq \kappa ( \sqrt{\delta} \|\AA\|_2 +  \sqrt{\log q} M)  \]
\end{proposition}

\noindent
Finally, we will need a generalized Matrix Chernoff bound for the interior eigenvalues of sums of random matrices, which was derived by Gittens and Tropp \cite{gittens2011tail}.

\begin{theorem}[Interior Eigenvalue Matrix Chernoff, Theorem 4.1 of \cite{gittens2011tail}]\label{thm:matcher}
Consider a finite sequence $\{\X_j\}$ of independent, random, positive-semidefinite matrices with dimension $m$, and assume that $\|\X_j\|_2 \leq L$ for some value $L$ almost surely. Given an integer $k \leq n$, define 
\[	\mu_k = \lambda_k\left(\sum_j \ex{\X_j}\right)	\] 
then we have the tail inequalities

\begin{equation}
\begin{cases}
\; \; \; \; \bpr{ \lambda_k( \sum_j \X_j) \geq (1 +\delta ) \mu_k } \leq (n-k+1) \cdot \left[ 	\frac{e^\delta}{(1+\delta)^{1+\delta}}\right]^{\mu_k/L} & \text{ for } \delta > 0 \\[12pt]
\; \; \; \;  \bpr{ \lambda_k( \sum_j \X_j) \leq (1 -\delta ) \mu_k } \leq k \cdot \left[ 	\frac{e^{-\delta}}{(1-\delta)^{1-\delta}}\right]^{\mu_k/L} & \text{ for } \delta \in [0,1) \\
\end{cases}
\end{equation}
\end{theorem}
\vspace{.2in}

\paragraph{The Algorithm.} Our first step is to  run the $\ell_\infty$-gap algorithm of Section \ref{sec:Linfty} with $\eps_0 = \frac{2}{k}$, where we set  $k = \frac{2 \cdot 400^2 \kappa^4}{\eps}$, 
 where $\kappa \geq 1$ is the constant in Proposition \ref{prop:vershy}. 
 This allows us to assume that $\lambda_i \geq - \eps_0 n / 1000 \geq - \frac{1}{2k}n$ for all $i$, otherwise we have a $\wt{O}(1/\eps^2)$-query algorithm from the  Section \ref{sec:Linfty}, and since our target complexity is $\wt{O}(1/\eps^4)$, we can safely disregard the cost of running this algorithm in parallel. We begin by demonstrating that the Frobenius norm of $\S \AA$ is preserved (up to scaling), where $\S$ is a random row sampling matrix with sufficiently many rows. 

\begin{proposition}\label{prop:trace}
	Let $\M \in \R^{m \times m}$. Fix $t \geq 1$ and let $\S$ be a row sampling matrix which samples each row of $\M$ with probability $p = \frac{t}{m}$, and let $\S \in \R^{t_0 \times m}$ be a row sampling matrix drawn from this distribution, where $\ex{t_0} = t$.
	Then we have
	\[ \ex{\frac{1}{p}\text{Tr}(\S \M \S^\top)} = \sum_i \lambda_i(\M)  = \text{Tr}(\M)\]
	and 
	\[ \text{Var}\left(\frac{1}{p}\text{Tr}(\S \M \S^\top)\right)\leq \frac{m}{t}\sum_i \M_{i,i}^2   \]
	
\end{proposition}
\begin{proof}
	For $i \in [m]$, let $\delta_i \in \{0,1\}$ indicate that we sample row $i$.	We have $\ex{\text{Tr}(\S \M\S^\top)} = \frac{1}{p} \ex{\sum_{i=1}^n \delta_i \M_{i,i} } = \text{Tr}(\M)$. Moreover, 
	\begin{equation}
	\begin{split}
	\text{Var}\left(\frac{1}{p}\text{Tr}(\S \M \S^\top)\right)& \leq \frac{1}{p^2}  \sum_{i=1}^n \delta_i \M_{i,i} - \left(\text{Tr}(\M)\right)^2 \\
	& \leq \sum_{i\neq j} \M_{i,i}\M_{j,j} + \frac{1}{p} \sum_i \M_{i,i}^2  - \left(\text{Tr}(\M)\right)^2\\
	& \leq  \frac{1}{p}\sum_i \M_{i,i}^2   \\
	\end{split}
	\end{equation}
	as stated.
\end{proof}


We now fix $t = \Theta(\log(1/\eps)/\eps^2)$, and draw row independent sampling matrices $\S, \T$ with an expected $t$ rows. Let $S ,T \subset [n]$ be the rows and columns sampled by $\S,\T^\top$ respectively. We then compute $\ZZ = \S \AA \T^\top$ with an expected $O(t^2)$ queries. Finally, we query the principal submatrix $\AA_{(S \cup T) \times (S \cup T)}$, and test whether $\AA_{(S \cup T) \times (S \cup T)}$ is PSD. Clearly if $\AA$ is PSD, so is $\AA_{(S \cup T) \times (S \cup T)}$, so it suffices to anaylzie the \texttt{NO} case, which we do in the remainder.

\begin{lemma}\label{lem:smallk}
	Let $\AA \in \R^{n \times n}$ be $\eps$-far from PSD with  $\|\AA\|_\infty \leq 1$.  Then let $\ZZ = \S \AA \T^\top$ be samples as described above, so that $\ZZ$ has an expected $t = \Theta(\log(1/\eps)/\eps^2)$ rows and columns, where $t$ is scaled by a larger enough constant, and let $k = \frac{2 \cdot 400^2 \kappa^4}{\eps}$, where $\kappa \geq 1$ is the constant in Proposition \ref{prop:vershy}. Suppose further that $\sigma_{k+1}(\AA) \leq  10 n/k$. Then  with probability $19/20$, we have 
	\[  \frac{n^2}{t^2}\sum_{i > k} \sigma_i^2(\ZZ) > \eps n^2/16  \]
\end{lemma}
\begin{proof} 
	
	Now write $\AA = \U \Lambda \V^\top$, $\AA_{k} = \U \Lambda_{k} \V^\top, \AA_{-k} = \U \Lambda_{-k} \V^\top$. Then $\AA = \AA_{k} + \AA_{-k}$, and the rows of $\AA_{k}$ are orthogonal to the rows of $\AA_{-k}$. Note that this implies that $\|\AA_{i,*}\|_2^2 = \|(\AA_{k})_{i,*}\|_2^2 + \|(\AA_{-k})_{i,*}\|_2^2$ for each $i \in [n]$ by the Pythagorean theorem, and since $\|\AA\|_\infty \leq 1$ we have $\|(\AA_{-k})_{i,*}\|_2^2 \leq n$. 
	
	Now set $\M_1 = \S \AA_{k}  \AA_{k}^\top \S^\top$, and   $\M_2 = \S \AA_{-k}  \AA_{-k}^\top \S^\top$. Notice that $\M_1 + \M_2 = \S (\AA_{k}  \AA_{k}^\top + \AA_{-k}  \AA_{-k}^\top )\S^\top = \S \AA \AA^\top \S^\top$, using the fact that the rows and columns of $\AA_{k}$ are orthogonal to the rows and columns (respectively) of $\AA_{-k}$.  Let $p = \frac{t}{n}$ be the row sampling probability. Now suppose $\|(\AA_{-k})\|_F^2 = \alpha n^2$. Note that we have shown that $\alpha > \eps /2$. By Proposition \ref{prop:trace}, we have $\ex{\text{Tr}(\M_2)/p} = \sum_{i > k}= \alpha n^2 > \eps n^2/2$ for some $\alpha \geq  \eps/2 $, where the last inequality follows from Proposition \ref{prop:nottopheavy}. Moreover, we have
	
	\begin{equation}
	\begin{split}
	\text{Var}\left(\frac{1}{p}\text{Tr}(\M_2)\right)&\leq \frac{1}{p}\sum_i (\M_2)_{i,i}^2  \\
	&= \frac{1}{p}\sum_i \|(\AA_{-k})_{i,*}\|_2^4\\
	\end{split}
	\end{equation} 
	It follows that since each row satisfies $\|(\AA_{-k})_{i,*}\|_2^2 \leq n$ and $\|(\AA_{-k})\|_F^2 = \alpha n^2$., the quantity $\sum_i \|(\AA_{-k})_{i,*}\|_2^4$ is maximized having $\alpha n$ rows with squared norm equal to $n$.  This yields 
	
	\begin{equation}
	\begin{split}
	\text{Var}\left(\frac{1}{p}\text{Tr}(\M_2)\right)& \leq  \frac{1}{p}\sum_i 2 \alpha n \cdot n^2\\
	& \leq 2 \frac{\alpha n^4}{t} \\
	& \leq \frac{ \alpha^2}{100^2} n^4 \\
	\end{split}
	\end{equation} 
	Where in the last line, we used that $t > \frac{4\cdot 100^2}{\eps_0} \geq \frac{2 \cdot 100^2}{\alpha}$. Then by Chebyshev's inequality, with probability $99/100$, we have $\frac{1}{p}\text{Tr}(\M_2) >  \alpha n^2  -  (\alpha/10) n^2 = (9/10)  \alpha n^2 \geq (9/20) \eps n^2$. Call this event $\mathcal{E}_1$, and condition on it now. 
	Next, by Proposition \ref{prop:vershy}, since $\sigma_{k+1}(\AA) \leq  10 n/k$ we have $\ex{\|\S\AA_{-k}\|_2} \leq  \kappa (10\sqrt{tn}/k + \sqrt{2\log(1/\eps)}\sqrt{n}) < 20 \kappa \sqrt{tn}/k$. Then by Markovs, we have $\|\S\AA_{-k}\|_2^2 = \|\M_2\|_2 \leq 200^2 \kappa^2   t n/k^2$ with probability $99/100$, which we condition on now, and call this event $\mathcal{E}_2$.
	Then by the Dual Lidskii inequality \ref{lem:duallidskii}, we have 
	\begin{equation}
	\begin{split}
	\frac{1}{p}\sum_{j> k} \lambda_j(\M_1 + \M_2) &\geq \frac{1}{p} \left(\sum_{j> k}\lambda_j( \M_2) \right)	\\
	&\geq \frac{1}{p}( \text{Tr}(\M_2) - k \|\M_2\|_2 )	\\
	&\geq  (9/20) \eps n^2 - 200^2 \kappa^2   n^2/k\\
	&\geq   \eps n^2/4 \\
	\end{split}
	\end{equation}
	using that $k > \frac{2 \cdot 400^2 \kappa^4}{ \eps}$. Now let $\W = \frac{1}{\sqrt{p}}(\S\AA)^\top$, and note that we took the transpose, so $\W$ has $n$ rows and $t_1$ columns, where $\ex{t_1} = t$. Now by Chernoff bounds, with probability $99/100$ we have $t_2 \leq 2t$; call this event $\mathcal{E}_3$ and condition on it now. The above demonstrates that $\frac{1}{p} \sum_{j > k} \lambda_j(\M_1 + \M_2) = \sum_{j > k}\lambda_j^2(\W) \leq \eps n^2 / 4$. Now note that $\sigma_{k+1}(\W) = \frac{1}{\sqrt{p}}(\sigma_{k+1}(\S\AA_{k} + \S \AA_{-k}) <  \frac{1}{\sqrt{p}} \|\S\AA_{-k}\|_2 \leq 200 \kappa n/k$, where we used the Weyl inequality for singular values: namely that for any two matrices $\AA,\BB$ and value $i$,  $|\sigma_i(\AA + \BB) - \sigma_i(\AA)| \leq \|\BB\|_2$, and using that $\S\AA_k$ is rank at most $k$, so $\sigma_{k+1}(\S\AA_{k}) = 0$.

	 Now draw a random row sampling matrix $\T$ with an expected $t$ rows, and write $\NN_1 = \T \W_k \W_k^\top \T$ and $\NN_2 = \T \W_{-k} \W_{-k}^\top \T$, and note again that $\NN_1 + \NN_2 = \T \W \W^\top \T$. Moreover, the rows of $\W_k$ live in a subspace orthogonal to the rows of $\W_{-k}$, so again by the Pythagorean theorem and boundedness of the entries in $\AA$, we have $\|(\W_{-k})_{i,*}\|_2^2 \leq \frac{1}{p} t_1 \leq 2 n$ for all $i \in [n]$.  
 Then by Proposition \ref{prop:trace}, we have $\ex{\text{Tr}(\NN_2)/p} = \|\W_{-k}\|_F^2 = \alpha n^2 \geq \eps n^2 / 4$, and 
\begin{equation}
\begin{split}
\text{Var}\left(\frac{1}{p}\text{Tr}(\NN_2)\right)& \leq \frac{1}{p}\sum_{i=1}^n \| (\W_{-k})_{i,*}\|_2^4 \\
& \leq \frac{1}{p} n^3 \\
& \leq \frac{1}{t}  n^4 \\
& \leq \frac{\eps^2}{100^2} n^4 \\
\end{split}
\end{equation}   Then by Chebyshev's inequality, with probability $99/100$, we have $\frac{1}{p}\text{Tr}(\NN_2) >  \eps n^2/4  -  (\eps /10) n^2 = \eps n^2/8$. Call this event $\mathcal{E}_4$, and condition on it now. Now as shown above, we have $\|W_{-k}\|_2 \leq 200 \kappa n/k$, thus by Proposition \ref{prop:vershy} we have $\ex{\|\T \W_{-k}\|_2 } \leq \kappa(200 \kappa \sqrt{tn} /k + 4\sqrt{\log(1/\eps)}\sqrt{n} ) \leq 400 \kappa^2 \sqrt{tn}$, again where we take $t = \Theta(\frac{\log(1/\eps)}{\eps^2})$ with a large enough constant. Then by Markov's inequality, with probability $99/100$ we have $\|\NN_2\|_2 \leq 400^2 \kappa^4 n^2/k^2$, and again by the Dual Lidskii inequality \ref{lem:duallidskii}, we have 
	\begin{equation}
	\begin{split}
	\frac{1}{p}\sum_{j> k} \lambda_j(\NN_1 + \NN_2) &\geq \frac{1}{p} \left(\sum_{j> k}\lambda_j( \NN_2) \right)	\\
	&\geq \frac{1}{p}( \text{Tr}(\NN_2) - k \|\NN_2\|_2 )	\\
	&\geq  \eps n^2/8 - 400^2 \kappa^4 n^2/k\\
	&\geq   \eps n^2/16 \\
	\end{split}
	\end{equation}
	Using that $k \geq \frac{2 \cdot 400^2 \kappa^4}{\eps}$. 	Note moreover that $$\frac{1}{p}\sum_{j> k} \lambda_j(\NN_1 + \NN_2) = \frac{1}{p}\sum_{j> k} \sigma_j^2(\T \W) = \frac{1}{p^2}\sum_{j> k} \sigma_j^2(\S \AA^\top \T^\top)$$ Using that $\AA = \AA^\top$ so that $ \ZZ = \S \AA^\top \T^\top$ we conclude that $ \frac{1}{p^2}\sum_{i > k} \sigma_i^2(\ZZ) =  \frac{n^2}{t^2}\sum_{i > k} \sigma_i^2(\ZZ) > \eps n^2/16  $ as desired. Note that we conditioned on $\mathcal{E}_i$ for $i=1,2,3,4,5$, each of which held with probability $99/100$, thus the result holds with probability $19/20$ by a union bound. 
		
\end{proof}

We will now address the case where $\sigma_k(\AA) > 10 n/k$.

\begin{lemma}\label{lem:bigk}
	Let $\AA \in \R^{n \times n}$ be $\eps$-far from PSD with  $\|\AA\|_\infty \leq 1$.  Then let $\ZZ = \S \AA \T^\top$ be samples as described above, so that $\ZZ$ has an expected $t = \Theta(\log(1/\eps)/\eps^2)$ rows and columns, where $t$ is scaled by a larger enough constant, and let $k = \frac{2 \cdot 400^2 \kappa^4}{\eps}$, where $\kappa \geq 1$ is the constant in Proposition \ref{prop:vershy}. Suppose further that $\sigma_{k}(\AA) > 10 n/k$. Then  with probability $49/50$,  we have 
	\[  \frac{n}{t} \sigma_k(\ZZ) \geq 8 n/k  \]
\end{lemma}
\begin{proof}
	The proof is by application of Theorem \ref{thm:matcher} twice. We first generate a random row  sampling matrix $\S$ with an expected $t$ rows, and bound $\lambda_k( (\S\AA)^\top \S \AA ) = \sigma_k^2(\S \AA)$. Let $\X_j \in \R^{n \times n}$ be a random variable such that $\X_j = \AA_{(j)}^\top \AA_{(j)}$, where $\AA_{(j)}$ is the $j$-th row of $\AA$ that was sampled in $\S$. Then $\sum_{j} \X_j = (\S\AA)^\top \S \AA $, and $\ex{\X_j} = \frac{t}{n}\sum_{j=1}^n \AA_j^\top \AA_j = \frac{t}{n} \AA \AA^\top$, where $ \AA_j$ is the $j$-th row of $\AA$.  Moreover, note that $\|\X_j\|_2 \leq \max_i \|\AA_{i,*}\|_2^2 \leq n$ for all $j$, by the boundedness of $\AA$. Thus note that $\mu_k = \lambda_k((t/n) \AA^\top \AA) \geq (t/n)  100 n^2 / k^2 = \frac{100 t n }{k^2}$. 
	Thus by the Interior Matrix Chernoff Bound \ref{thm:matcher}, we have that for some constant $c$:
	
	\begin{equation}
	\begin{split}
	\bpr{ \lambda_k((\S\AA)^\top \S \AA) \leq .9 \mu_k } &\leq k \cdot c^{\mu_k/L} \\
	 &\leq k \cdot c^{\frac{100 t n }{k^2 } \cdot \frac{1}{n}} \\  
	 	 &\leq k \cdot e^{-100 \log(k)} \\
	 	 & \leq   1/1000
	\end{split} 
	\end{equation}
Where we use $t = \Theta(\frac{\log(1/\eps)}{\eps^2})$ with a large enough constant.  Also condition on the fact that $\S$ has at most $2t$ rows, which holds with probability $999/1000$. Call the union of the above two event $\mathcal{E}_1$, which holds with probability $99/100$, and condition on it now. Given this, we have $\sigma_k^2(\S\AA) \geq  \frac{90 tn}{k^2}$. Now again, let $\Y_j = (\S\AA)_{(j)} (\S\AA)_{(j)}^\top$, where $(\S\AA)_{(j)}$ is the $j$-th column of $\S\AA$ sampled by the column sampling matrix $\T$. Let $\M = (\S \AA)^\top$. Then again we have $\|\Y_j\|_2 \leq 2t$, using that $\S\AA$ has at most $2t$ rows, and each entry is bounded by $1$. Moreover, $\sum_j \Y_j =   \T \M \M^\top \T^\top  $ We also have  $\lambda_k(\ex{ \sum_j \Y_j}) =  \lambda_k(\frac{t}{n} \M \M^\top ) > \frac{90 t^2}{k^2}$. Applying the Interior Matrix Chernoff Bound again, we have that for some constant $c$:
	
\begin{equation}
\begin{split}
\bpr{ \lambda_k(  \T (\S \AA)^\top(\S \AA) \T^\top ) \leq .9 \mu_k } &\leq k \cdot c^{\mu_k/L} \\
&\leq k \cdot c^{\frac{90 t^2  }{k^2 } \cdot \frac{1}{2t} }\\  
&\leq k \cdot e^{-100 \log(k)} \\
& \leq   1/1000 \\
\end{split} 
\end{equation}
Call the above event $\mathcal{E}_2$. Conditioned on $\mathcal{E_1} \cup \mathcal{E}_2$, which hold together with probability $49/50$, we have that 
 $\sigma_k(  \S \AA \T^\top ) \leq .9 \sqrt{\frac{90 t^2}{k^2}} > 8 t/k$. Since $\ZZ = \S \AA \T^\top $, we have $\frac{n}{t} \sigma_k( \ZZ ) > 8n/k$ as needed. 

\end{proof}

\begin{theorem}\label{thm:l2premain}
	Let $\AA \in \R^{n \times n}$ be $\eps$-far from PSD with $\|\AA\|_\infty \leq 1$. Then if $S,T \subset [n]$ are random subsets with expected each size $t = O(\log(1/\eps)/\eps^2)$, then with probability $9/10$ the principal submatrixx $\AA_{(S \cup T) \times (S \cup T)}$ is not PSD.
\end{theorem}
\begin{proof}
First, by Chernoff bounds, with probability $99/100$ we have $|S \cup T| \leq |S| + |T| \leq 4t$, which we call $\mathcal{E}_1$ and condition on now. 
	First, consider the case that $\sigma_k(\AA) \leq 10 n/k$, where $k = \frac{2 \cdot 400^2 \kappa^4}{\eps}$. Then by Lemma \ref{lem:smallk}, with probability $19/20$, we have that $\sum_{i >k}\sigma_i^2(\AA_{S \times T}) > \eps t^2/16$. 
	Now  we first prove the following claim:
	
	\begin{claim}
		Let $\ZZ \in \R^{n \times m}$ be any matrix, and let $\tilde{\ZZ}$ be a rectangular submatrix of $\ZZ$. for any Let $\ZZ_k,\tilde{\ZZ}_k$ be the truncated SVD of $\ZZ,\tilde{\ZZ}$ respectively, for any $1 \leq k \leq \min\{n,m\}$. Then we have 
		\[\|\ZZ - \ZZ_k\|_F^2 \geq \|\tilde{\ZZ}- \tilde{\ZZ}_k\|_F^2		\]
	\end{claim}
\begin{proof}
	Note that $\|\ZZ - \ZZ_k\|_F^2 \geq \|\tilde{\ZZ}_k - \ZZ_k'\|_F^2$, where $\ZZ_k'$ is the matrix $\ZZ_k$ restricted to the submatrix containing $\tilde{\ZZ}$. But $\tilde{\ZZ}_k$ is the \textit{best} rank-$k$ approximation to $\tilde{\ZZ}$, so $ \|\tilde{\ZZ}- \tilde{\ZZ}_k\|_F^2	 = \min_{\BB \text{rank-k}}  \|\tilde{\ZZ}- \BB\|_F^2 \leq  \|\tilde{\ZZ}_k - \ZZ_k'\|_F^2$, using the fact that a submatrix of a rank-k matrix is at most rank $k$.
\end{proof}
It follows that $\|\AA_{(S \cup T) \times (S \cup T)} - (\AA_{(S \cup T) \times (S \cup T)})_k\|_F^2 = \sum_{j > k} \sigma_j^2(\AA_{(S \cup T) \times (S \cup T)}) \geq  \sum_{j > k} \sigma_j^2(\AA_{ S \times T}) > \eps t^2/16 > \eps |S \cup T|^2/256$. But note that if $\AA_{(S \cup T) \times (S \cup T)}$ was PSD, then we would have  $\sum_{j > k} \sigma_j^2(\AA_{ S \times T}) < \leq \frac{16}{k}t^2$, which is a contradiction since $k = \frac{2 \cdot 400^2 \kappa^4}{\eps} > \frac{100^2}{\eps}$.

Now consider the case that $\sigma_k(\AA) > 10 n/k$. Then by Lemma \ref{lem:bigk}, we have $\sigma_k((\AA_{ S \times T}) \geq 8 t/k$ with probability at least $49/50$. Then $\|\AA_{ S \times T}\|_{\mathcal{S}_1} \geq \sum_{i=1}^k \sigma_i((\AA_{ S \times T}) \geq 8t$. Using the fact that the Schatten norm of a matrix is always at least as large as the Schatten norm of any submatrix (this follows from the fact that the singular values of the submatrix are point-wise dominated by the larger matrix, see Theorem 1 \cite{thompson1972principal}), we have  $\|\AA_{(S \cup T) \times (S \cup T)}\|_{\mathcal{S}_1} \geq 8t$. But note that if $\AA_{(S \cup T) \times (S \cup T)}$ was PSD, then we would have $\|\AA_{(S \cup T) \times (S \cup T)}\|_{\mathcal{S}_1}  = \text{Tr}(\AA_{(S \cup T) \times (S \cup T)}) \leq |S \cup T| \leq 4t$, which is a contradiction. This completes the proof of the theorem.
\end{proof}

\begin{theorem}\label{thm:l2Main}
Fix  $\AA \in \R^{n \times n}$ with $\|A\|_\infty \leq 1$.
	There is a non-adaptive sampling algorithm that, with probability $9/10$, correctly distinguishes the case that $\AA$ is PSD from the case that $\AA$ is $\eps$-far from PSD in $\ell_2$, namely that $\sum_{i : \lambda_i(\AA) < 0} \frac{ \lambda_i^2(\AA)}{n^2} \geq \eps$. The algorithm queries a total of $O(\frac{\log^2(1/\eps)}{\eps^4})$ entries of $\AA$, and always correctly classifies $\AA$ as PSD if $\AA$ is indeed PSD. Moreover, the algorithm runs in time $\tilde{O}(1/\eps^{2\omega})$, where $\omega< 2.373$ is the exponent of fast matrix multiplication.
\end{theorem}

\begin{proof}
We first apply the algorithm of Section \ref{sec:Linfty} with $\eps_0 = \frac{2}{k}$, which as discussed allows us to assume that $\lambda_i \geq - \eps_0 n / 1000 \geq - \frac{1}{2k}n$ for all $i$. The cost of doing so is  $\wt{\Theta}(1/\eps^2)$ queries, and this algorithm also yields one-sided error as desired. The remainder of the theorem follows directly from Theorem \ref{thm:l2premain}, using that all principal submatrices of PSD matrices are PSD. Finally, for runtime, notice that the main computation is computing the eigenvalues of a $k \times k$ principal submatirx, for $k = \tilde{O}(1/\eps^2)$, which can be carried out in time $\tilde{O}(1/\eps^{2\omega})$ \cite{demmel2007fast2,banks2019pseudospectral}.
\end{proof}

	\section{Lower bounds}\label{sec:lb}

\subsection{Lower Bound for PSD Testing with \texorpdfstring{$\ell_\infty$}{L-infinity} Gap}
	We begin by demonstrating a $O(1/\eps^2)$ lower bound for the problem of testing postive semi-definiteness with an $\ell_\infty$ gap. Our lower bound holds even when the algorithm is allowed to adaptively sample entryies of $\AA$.

	\begin{theorem}\label{thm:linftyLB}
		Any adaptive or non-adaptive algorithm which receives query access to $A \in \R^{n \times n}$ with $\|\AA\|_\infty \leq 1$, and distinguishes with probability at least $2/3$ whether 
		\begin{itemize}
			\item $\AA$ is PSD.
			\item $x^T\AA x< - \eps n$ for some unit vector $x \in \R^n$ and $\eps \in (0,1)$
		\end{itemize}
		must make $\Omega(1/\eps^2)$ queries to $A$.
	\end{theorem}
\begin{proof}
	We construct two distributions $\mathcal{D}_1,\mathcal{D}_2$ over matrices, and draw the input $A$ from the mixture $(\mathcal{D}_1 + \mathcal{D}_2 )/2$. $\mathcal{D}_1$ is supported on one matrix: the zero matrix $\mathbf{0}^{n \times n}$, which is PSD. Now set $t = 2 \eps^2 n$ and let $B \in \R^{n \times n}$ be the matrix given by
	
	\[	\BB = \begin{bmatrix}
	0 & -\mathbf{1}^{ n  - t \times t} \\
-	\mathbf{1}^{ t \times n  -t} &  -\mathbf{1}^{ t\times t} \\
	\end{bmatrix}	\]	
	Where $-	\mathbf{1}^{ n\times m} $ is the $n \times m$ matrix consisting of a $-1$ in each entry. Now let $x \in \R^{n \times n}$ be defined by $x_i = 1$ for $i=1,2,\dots,n-t$, and let  $x_j = 1/\eps$ for $j > n - t$. Then note that $x^T \BB x < - \frac{1}{\eps} \cdot 2 \eps^2 n^2 < -  \eps n \|x\|_2^2 $, thus $\BB$ is $\eps$-far from PSD in $\ell_\infty$ gap. To sample $A \sim \mathcal{D}_1$, we set $\AA = \P_\Sigma \BB \P_{\sigma}^T$, where $\P_{\sigma}$ is a randomly drawn permutation matrix, namely $\sigma \sim S_n$ uniformly at random. Notice that to distinguish $A \sim \mathcal{D}_1$ from $A \sim \mathcal{D}_2$, the algorithm must read a non-zero entry. By Yao's min-max principle, we can assume that there is a deterministic algorithm that solves the problem with probability $2/3$ over the randomness of the distribution. Fix any $k < 1/(100\eps^2)$, and let $s_1,s_2,\dots s_k$ be the adaptive sequence of entries it would sample if $A_{s_i} = 0$ for each $i=1,2,\dots,k$. Then then the probability that any of the the $s_i$'s land in a row or a column of  $\AA = \P_\Sigma \BB \P_{\sigma}^T$ with non-zero entries is at most $1/50$. Thus with probability $49/50$ under input from $A \sim \mathcal{D}_2$, the algorithm will output the same value had $A$ been the all zero matrix. Thus the algorithm succeeds with probability at most $51/100$ when $A$ is drawn from the mixture, demonstrating that $\Omega(1/\eps^2)$ samples are required for probability $2/3$ of success.
\end{proof}

\subsection{Lower Bound for PSD Testing with \texorpdfstring{$\ell_2$}{L-2}  Gap}
We now present our main lower bound for PSD testing. Our result relies on the construction of explicit graphs with gaps in their spectrum, which have the property that they are indistinguishable given only a small number of queries to their adjacency matrices. 
Our lower bound is in fact a general construction, which will also result in lower bounds for testing the Schatten $1$ norm, Ky-Fan norm, and cost of the best rank $k$ approximation.

\paragraph{Roadmap}

In the following, we will first introduce the notation and theory required for the section, beginning with the notion of subgraph equivalence of matrices. We then construct our hard distributions $\mathcal{D}_1,\mathcal{D}_2$, and prove our main conditional results, Lemma \ref{lem:ifBDthenLB}, which demonstrates a lower bound for these hard distributions conditioned on the existence of certain pairs of subgraph equivalent matrices. Finally, we prove the existence of such matrices, which is carried out in the following Section \ref{sec:CnLemma}. Putting these pieces together, we obtain our main lower bound in Theorem \ref{thm:lbmain}. 

\paragraph{Preliminaries and Notation}
In the following, it will be useful to consider \textit{signed graphs}. 
A signed graph $\Sigma$ is a pair $(|\Sigma|, s)$, where $|\Sigma| = (V,E)$ is a simple graph, called the \textit{underlying graph}, and $s:E \to \{1,-1\}$ is the \textit{sign function}. We will sometimes abbreviate the signs equivalently as $\{+,-\}$. We will write $E^+,E^-$ to denote the set of positive and negative edges. If $\Sigma$ is a signed graph, we will often write $\Sigma = (V(\Sigma),E(\Sigma))$, where $E(\Sigma)$ is a set of \textit{signed} edges, so $E(\Sigma) \subset \binom{|V(\Sigma)|}{2} \times \{+,-\}$ with the property that for each $e \in \binom{|V(\Sigma)|}{2}$, at most one of  $(e,+,),(e,-)$ is contained in $E(\Sigma)$. For a signed graph $G$ on $n$ vertices, let $\AA_G \in \{1,0,-1\}^{n \times n}$ be its adjacency matrix, where $(\AA_G)_{i,j}$ is the sign of the edge $e=(v_i,v_j)$ if $e \in E(G)$, and is $0$ otherwise.

For a graph $H$, let $\|H\|$ denote the number of vertices in $H$. For any simple (unsigned) graph $G$, let $\overline{G}$ be the signed graph obtained by having $E^+(\overline{G}) = E(G)$, and $E^-(\overline{G}) = \binom{|V|}{2} \setminus E(G)$. In other words, $\overline{G}$ is the complete signed graph obtained by adding all the edges in the complement of $G$ with a negative sign, and giving a positive sign the edges originally in $G$. We remark that the negation of the adjacency matrix of $\overline{G}$ is known as the \textit{Seidel matrix} of $G$. In what follows, we will often not differentiate between a signed graph $G$ and its signed adjacency matrix $\AA_G$. For graphs $G,H$, let $G \oplus H$ denote the disjoint union of two graphs $G,H$
 We will assume familiarity with basic group theory. For groups $G,H$, we write $H \leq G$ if $H$ is a subgroup of $G$. For a set $T$, let $2^T$ denote the power set of $T$.  Throughout, let $S_n$ denote the symmetric group on $n$ letters. 
  For two signed graphs $\Sigma,H$, let $\mathcal{F}_H(\Sigma) = \{G  = (V(\Sigma),E(G)) \; | \; E(G) \subseteq E(\Sigma), G \cong H  \}$ be the set of signed subgraphs of $\Sigma$ isomorphic to $H$. For a permutation $\sigma \in S_n$, we write $\P_\sigma \in \R^{n \times n}$ to denote the row permutation matrix associated with $\sigma$. For $k \geq 3$, let $C_k$ denote the cycle graph on $k$ vertices. 
  
 For signed graphs $G,H$, a signed graph isomorphism (or just isomorphism) is a graph isomorphism that preserve the signs of the edges. For any set $U  \subset [n] \times [n]$ and matrix $\AA \in \R^{n \times n}$, we write $\AA_U$ to denote the matrix obtained by setting the entries $(A_U)_{i,j} = A_{i,j}$ for $(i,j) \in U$, and  $(\AA_U)_{i,j} = 0$ otherwise. A set $U  \subset [n] \times [n]$ is called symmetric if $(i,j) \in U \iff (j , i) \in U$. We call $U$ simple if it does not contain any elements of the form $(i,i)$. We will sometimes refer to a simple symmetric $U$ by the underlying simple undirected graph induced $U$.
 
 \paragraph{Subgraph Equivalence} 
We now formalize the indistinguishably property which we will require. For matrices $\AA,\BB$, when thought of as adjacency matrices of graphs, this property can be thought of as a more general version of ``locally indistinguishability'', in the sense that, for any small subgraph $H$ of $\AA$, there is a \textit{unique} subgraph of $\BB$ that is isomorphic to $H$. The following definition is more general, in the sense that a subgraph can also have ``zero valued edges'', corresponding to the fact that an algorithm can learn of the non-existence of edges, as well as their existence.

 \begin{definition}[Sub-graph Equivalence]\label{def:subgraphequiv}
 Fix any family $\mathcal{U}$ of symmetric subsets $\mathcal{U} = \{U_i\}_i \in 2^{[n] \times [n]}$, and let $\Gamma \leq S_n$ be a subgroup of the symmetric group on $n$ letters. Let $\AA,\BB \in \R^{n \times n}$. Then we say that $\AA$ is $(\mathcal{U},\Gamma)$-subgraph isomorphic to $\BB$, and write $\AA \cong_{\mathcal{U},\Gamma} \BB$, if for every $U_i \in \mathcal{U}$ there is a bijection $\psi_i: \Gamma \to \Gamma$ such that 
 \[  \left(\P_{\sigma} \AA \P_{\sigma}^T \right)_{U_i}    = \left(\P_{\psi_i(\sigma)} \BB  \P_{\psi_i(\sigma)}^T \right)_{U_i}   \]
 for all $\sigma \in \Gamma$. If $G,H$ are  two signed graphs on $n$ vertices with adjacency matrices $\AA_G,\AA_H$, then we say that $G$ is $(\mathcal{U},\Gamma)$-subgraph  equivalent to $H$, and write  $G \cong_{\mathcal{U},\Gamma} H$, if $\AA_G  \cong_{\mathcal{U},\Gamma}  \AA_H$. 
 \end{definition}
 \noindent  Note we do not require the $U_i$'s to be simple in the above definition. At times, if $\Gamma = S_n$, then we may omit $\Gamma$ and just write $G \cong_{\mathcal{U}} H$ or $\AA  \cong_{\mathcal{U}}  \BB$.

\begin{example}
Let $G,H$ be arbitrary graphs on $n$ vertices, and let each $\mathcal{U} = \{U_i\}$ be a simple graph consisting of a single edge. Then $G \cong_{\mathcal{U},S_n} H$ if and only if $|E(G)| = |E(H)|$. 
\end{example}

\begin{example}
Let $G,H$ be arbitrary graphs on $n$ vertices, and let $\mathcal{U} = \{U_i\}$  be a single graph, where $U_i$ is a triangle on any three vertices. Then $G \cong_{\mathcal{U},S_n} H$ if and only if the number of induced subgraphs on three vertices that are triangles, wedges, and single edges, are each the the same in $G$ as in $H$.
\end{example}

In what follows, we will consider graphs that are $\mathcal{U}$ subgraph isomorphic, for a certain family of classes $\mathcal{U}$, which we now define. In what follows, recall that the matching number $\nu(G)$ of a graph $G$ is the size of a maximum matching in $G$, or equivalently the maximum size of any subset of pairwise vertex disjoint edges in $G$.
\begin{definition}
For $1 \leq t \leq n$, let $\mathcal{U}^t_n$ be the set of all undirected, possibly non-simple graphs $U_i$ on $n$ vertices, with the property that after removing all self-loops, $U_i$ does not contains any set of $t$ vertex disjoint edges. Equivalently, after removing all self-loops from $U_i$, the matching number $\nu(U_i)$ of $U_i$ is less than $t$.
\end{definition}
 In other words, $\mathcal{U}^t_n$ is the set of graphs with no set of $t$ pair-wise non-adjacent edges $e_1,\dots,e_t$ such that each $e_i$ is not a self loop. 
Notice by the above definition that $\mathcal{U}^{t}_n \subset \mathcal{U}^{t+1}_n$. We will also need the following definition.
\begin{definition}
For any $n,m \leq 1$, let $\Gamma_{n,m} \leq S_{nm}$ be the subgroup defined $\Gamma_{n,m} = \{\sigma \in S_{nm} \; | \; \sigma(i,j) = (\pi(i),j), \; \pi \in S_{n} \}$, where the tuple $(i,j) \in [n] \times [m]$ indexes into $[nm]$ in the natural way. 
\end{definition}
\noindent
Notice in particular, if $\AA \in \R^{n \times n}$ and $\DD \in\R^{m \times m}$, then we have \[\{ \P_{\sigma}(\AA \otimes \DD) \P_{\sigma}^T \; | \; \sigma \in \Gamma_{n,m}\} =\{ (\P_{\pi} \otimes \mathbb{I}_m) (\AA \otimes \DD) (\P_{\pi} \otimes \mathbb{I}_m)^T \; | \; \pi \in S_n\}\]
Note also by elementary properties of Kronecker products, we have $(\P_{\pi} \otimes \mathbb{I}_m) (\AA \otimes \DD) (\P_{\pi} \otimes \mathbb{I}_m)^T = ( \P_{\pi} \AA \P_{\pi}^T) \otimes \DD$. For such a $\sigma \in \Gamma_{n,m}$, we write $\sigma = \pi \otimes \text{id}$, where $\pi \in S_n$

\begin{lemma}\label{len:kron}
    Fix any $t ,m \geq 1$, and let $\AA,\BB \in \R^{n \times n}$ be matrices with  $\AA \cong_{\mathcal{U}^t_n,S_n} \BB$, where $\mathcal{U}^t_n$ is defined as above, and let $\T \in \R^{m \times m}$ be any matrix. Then $\AA \otimes \T \cong_{\mathcal{U}^t_{nm},\Gamma_{n,m}} \BB \otimes \T$, where $\Gamma \leq S_{nm}$ is as defined above.\footnote{Note that this fact extends naturally to tensoring with rectangular matrices $\T$.}
\end{lemma}
\begin{proof}
    Fix any $U_i' \in \mathcal{U}^t_{nm}$. Note that every edge of $U_i'$ corresponds to a unique edge of a graph on $n$ vertices. This can be seen as every edge of $U_i'$ is of the form $((i_1,j_1),(i_2,j_2))$ where $i_1,i_2 \in [n], j_1 ,j_2 \in [m]$, which corresponds to the edge $(i_1,i_2) \in [n] \times [n]$. So let $U_i \subset [n] \times [n]$ be the set of all such edges induced by the edges of $U_i'$. Observe, of course, that many distinct edges of $U_i'$ could result in the same edge of $U_i$. We claim that $U_i \in \mathcal{U}^t_n$. Suppose this was not the case, and let $e_1,\dots,e_{t} \in U_i$ be vertex disjoint non-self loop edges, where $e_\ell = (i_\ell,j_\ell)$, $i_\ell \neq j_\ell$. Then for each $\ell \in [t]$, there must be at least one edge $e_\ell' \in U_i'$ such that $e_{\ell}' = ( (i_\ell, a_\ell) (j_\ell,b_\ell)) \in U_i'$, and we can fix $e_{\ell}'$ to be any such edge. Then since each vertex $i_{\ell} \in [n]$ occured in at most one edge of $e_1,\dots,e_{t}$ by assumption, it follows that each vertex $(i_\ell,j_\ell) \in [n] \times [m]$ occurs at most once in $e_1',\dots,e_{t}'$, which contradictts the fact that the $U_i' \in \mathcal{U}^t_{nm}$. 
    
    Now that we have $U_i \in \mathcal{U}^t_n$, since $\AA \cong_{\mathcal{U}^t_n,S_n} \BB$ we have a bijection function $\psi_i:S_n \to S_n$ such that $\left(\P_{\pi} \AA \P_{\pi}^T \right)_{U_i}    = \left(\P_{\psi_i(\pi)} \BB \P_{\psi_i(\pi)}^T \right)_{U_i}$. We now define the mapping $\hat{\psi}_i : \Gamma_{n,m} \to \Gamma_{n,m}$  by $\hat{\psi}_i(\pi \otimes \text{id}) = \psi_i(\pi) \otimes \text{id}$, and show that it satisfies the conditions of Definition \ref{def:subgraphequiv}. Now note that each $\sigma = \pi \otimes \text{id} \in \Gamma_{n,m}$ satisfies $\P_\sigma = \P_{\pi} \otimes \mathbb{I}$, and so $\P_\sigma (\AA \otimes \T )\P_\sigma^T = \P_{\pi}\AA \P_{\pi}^T \otimes\T$. 
 
    We now claim that for any $U_i' \in \mathcal{U}^t_{nm}$, if we construct $U_i \in \mathcal{U}^t_n$ as above, we have that for any matrix $\ZZ \in \R^{n \times n}$ the non-zero entries of  $(\ZZ)_{U_i} \otimes \T$ contain the non-zero entries of $(\ZZ\otimes \T )_{U_i'}$. As a consequence, if $(\ZZ)_{U_i} \otimes \T = (\Y)_{U_i} \otimes \T$ for some other matrix $\Y \in \R^{n \times n}$, we also have $(\ZZ \otimes \T)_{U_i'} = (\Y \otimes \T)_{U_i'}$. But the claim in question just follows from the construction of $U_i'$, since for every entry $((i_1,j_1),(i_2,j_2)) \in U_i'$ we added the entry $(i_1,i_2) \in U_i$.    Now since we have  that$\left(\P_{\pi} \AA \P_{\pi}^T \right)_{U_i}    = \left(\P_{\psi_i(\pi)} \BB  \P_{\psi_i(\pi)}^T \right)_{U_i}$, we also obtain \[\left(\P_{\pi} \AA \P_{\pi}^T \right)_{U_i} \otimes  \T    = \left(\P_{\psi_i(\pi)}\BB  \P_{\psi_i(\pi)}^T \right)_{U_i}  \otimes \T\] which as just argued implies that 
    
    \[      \left(\P_{\pi} \AA \P_{\pi}^T\otimes \T\right)_{U_i'}    = \left(\P_{\psi_i(\pi)} \BB  \P_{\psi_i(\pi)}^T   \otimes  \T\right)_{U_i'} \] 
    Since $\left(\P_{\pi} \AA \P_{\pi}^T\otimes\T \right)_{U_i'} = \left(\P_{\sigma}( \AA \otimes  \T) \P_{\sigma}^T \right)_{U_i'}$ and $(\P_{\psi_i(\pi)} \BB  \P_{\psi_i(\pi)}^T   \otimes \T)_{U_i'} =  (\P_{\hat{\psi}(\sigma)}( \BB \otimes   \T) \P_{\hat{\psi}(\sigma)}^T )_{U_i'}$, it follows that $\AA \otimes \T \cong_{\mathcal{U}^t_{nm},\Gamma_{n,m}} \BB \otimes \T$ as required.
    
\end{proof}


\paragraph{The Hard Instance.} We now describe now distributions, $\mathcal{D}_1,\mathcal{D}_2$, supported on $n \times n$ matrices $\AA$ and paramterized by a value $k \geq 1$, such that distinguishing $\mathcal{D}_1$ from $\mathcal{D}_2$ requires $\Omega(k^2)$ samples. The distributions are parameterized by \textit{three} matrices, $(\BB,\DD,\ZZ)$, which are promised to satisfy the properties that $\BB,\DD \in \R^{d \times d}$ with $\BB \cong_{\mathcal{U}^t_{d} , S_d} \DD$ for some $t \leq d$, and $\ZZ \in \R^{m \times m}$, where $m=n/(dk)$. Also define $\wt{\BB} = \BB \otimes \ZZ$, $\wt{\DD} = \DD \otimes \ZZ$. We now define the distribution. We first define $\mathcal{D}_1$. In $\mathcal{D}_1$, we select a random partition of $[n]$ into $L_1,\dots,L_k$, where each $|L_i| = n/k$ exactly. Then for each $i \in [k]$, we select a uniformly random $\sigma_i \in \Gamma_{d,m}$ and set $\AA_{L_i \times L_i} = \P_{\sigma_i} \wt{\BB} \P_{\sigma_i}^T$, and the remaining elements of $\AA$ are set to $0$. In $\mathcal{D}_2$, we perform the same procedure, but set $\AA_{L_i \times L_i} = \P_{\sigma_i} \wt{\DD} \P_{\sigma_i}^T$. So if $\AA \sim \frac{\mathcal{D}_1 + \mathcal{D}_2}{2}$, then $\AA$ is block-diagonal, with each block having size $n/k$. We first demonstrate that for any matrices $(\BB,\DD,\ZZ)$ satisfying the above properties, distinguishing these distributions requires $\Omega(k^2)$ samples. We assume in the following that $dk$ divides $n$, which will be without loss of generality since we can always embed a small instance of the lower bound with size $n'$ such that $ n/2<n-dk\leq  n' \leq n$, and such that $dk$ divides $n$. 

\begin{lemma}\label{lem:ifBDthenLB}
Fix any $1 \leq k ,d\leq n$. Let $(\BB,\DD,\ZZ)$ be any three matrices such that $\BB,\DD \in \R^{d \times d}$,  $\BB \cong_{\mathcal{U}^t_{d} , S_d} \DD$ where $t = \log k$, and $\ZZ \in \R^{m \times m}$, where $m = n/(dk)$. Then any non-adaptive sampling algorithm which receives $\AA \sim \frac{\mathcal{D}_1 + \mathcal{D}_2}{2}$ where the distributions are defined by the tuple $(\BB,\DD,\ZZ)$ as above, and distinguishes with probability at least $2/3$ whether $\AA$ was drawn from $\mathcal{D}_1$ or $\mathcal{D}_2$ must sample $\Omega(k^2)$ entries of $\AA$. 

\end{lemma}
\begin{proof}

We show that any algorithm cannot distinguish $\mathcal{D}_1$ from $\mathcal{D}_2$ with probability greater than $2/3$ unless it makes at least $\ell > C\cdot k^2$ queries, for some constant $C>0$. So suppose the algorithm makes at most $C\cdot k^2/100$ queries in expectation and is correct with probability $2/3$. Then by Markov's there is a algorithm that always makes at most $\ell = C k^2$ queries which is correct with probability $3/5$. By Yao's min-max principle, there is a determinstic algorithm making this many queries which is correct with probability $3/5$ over the distribution $\frac{\mathcal{D}_1 + \mathcal{D}_2}{2}$. So fix this algorithm, which consists of a single subset $U \subset[n] \times [n]$ with $|U| = \ell$. 

We now generate the randomness used to choose the partition $L_1,\dots,L_k$ of $[n]$. Let $U_i = U \cap L_i \times  L_i = \{(i,j) \in U \; | \; i,j \in L_i\}$. Let $\mathcal{E}_i$ be the event that $U_i \in \mathcal{U}_{md}^t$. We first bound $\pr{\neg \mathcal{E}_i}$, where the probability is over the choice of the partition $\{L_i\}_{i \in [k]}$. For $\neg \mathcal{E}_i$, there must be $t$ pairwise vertex disjoint non-self loop edges $e_1,\dots,e_t \in U$ such that $e_{j} = (a_j,b_j)$ and $a_j, b_j \in L_i$. In other words, we must have $2t$ distinct vertices $a_1,b_1,\dots,a_t,b_t \in L_i$. For a fixed vertex $v \in [n]$, this occurs with probability $1/k$, and the probability that another $u \in [n] \cap L_i$ conditioned on $v \in [n]$ is strictly less than $1/k$ as have have $|L_i| = n/k$ exactly. Thus, the probability that all $2t$ vertices are contained in $L_i$ can then be bounded $\frac{1}{k^{2t}}$. Now there are $\binom{|U|}{t} \leq \ell^t$ possible choices of vertex disjoint edges $e_1,\dots,e_t \in U$ which could result in $\mathcal{E}_i$ failing to hold, thus $\pr{ \mathcal{E}_i } \geq 1 - \frac{\ell^t}{k^{2t}}$
and by a union bound
\begin{equation}
    \begin{split}
         \pr{\cap_{i=1}^k \mathcal{E}_i } & \geq 1 - \frac{\ell^t}{k^{2t-1}} \\
         & \geq   1 - \frac{C^t k^{2t}}{k^{2t-1}} \\
           & \geq   1 - C^t k \\
            & \geq   \frac{99}{100} \\
    \end{split}
\end{equation}
Where in the last line, we took $C \leq 1/10$ and used the fact that $t = \log(k)$. Then if $\mathcal{E} = \cap_{i=1}^k \mathcal{E}_i$, we have $\pr{\mathcal{E}} > 99/100$, which we condition on now, along with any fixing of the $L_i$'s that satisfies $\mathcal{E}$. Conditioned on this, it follows that  $U_i \in \mathcal{U}_{md}$ for each $i \in [k]$.  Using that $B\cong_{\mathcal{U}^t_{d} , S_n} D$, we can and apply Lemma \ref{len:kron} to obtain  $\BB \otimes \ZZ \cong_{\mathcal{U}^t_{dm},\Gamma_{d,m}} \DD \otimes \ZZ$. Thus, for each $i \in [k]$ we can obtain a bijection function $\psi_i:\Gamma_{d,m} \to \Gamma_{d,m}$ such that $(\P_{\sigma_i} \wt{\BB} \P_{\sigma_i}^T)_{U_i} = (\P_{\psi_i(\sigma_i)} \wt{\DD} \P_{\psi(\sigma_i)}^T)_{U_i}$ for each $\sigma_i \in \Gamma_{d,m}$. Thus we can create a coupling of draws from $\mathcal{D}_1$ with those from $\mathcal{D}_2$ conditioned on $\mathcal{E}$, so for any possible draw from the remaining randomness of $\mathcal{D}_1$, which consists only of drawing some $(\sigma_1,\dots,\sigma_k) \in S_d^k$ generating a matrix $\AA_1$, we have a unique corresponding draw  $(\psi_1(\sigma_1),\dots,\psi_k(\sigma_k)) \in S_d^k$ of the randomness in $\mathcal{D}_2$ which generates a matrix $\AA_2$, such that $(\AA_1)_U = (\AA_2)_U$. Thus conditioned on $\mathcal{E}$, any algorithm is correct on $\frac{\mathcal{D}_1 + \mathcal{D}_2}{2}$ with probability exactly $1/2$. Since $\mathcal{E}$ occured with probability $99/100$, it follows than the algorithm is correct with probability $51/100 < 3/5$, which is a contradiction. Thus we must have $\ell \geq  Ck^2 = \Omega(k^2)$ as needed.

\end{proof}

We are now ready to introduce our construction of the matrices as required in the prior lemma. Recall that $k \geq 3$, let $C_k$ denote the cycle graph on $k$ vertices. 

\begin{fact}\label{fact:cycle}
Fix any $n \geq 3$.
We have $\lambda_{\min}(C_{2n + 1}) = -2 + \Theta(1/n^2)$ and $\lambda_{\min}(C_{n} \oplus C_{n+1}) = -2$. 
\end{fact}
\begin{proof}
    The eigenvalues of the cycle $C_\ell$ are given by $2\cos(\frac{2 \pi t}{\ell})$ \cite{chung1996lectures} for $t=0,\dots,\ell-1$, which yields the result using the fact that $\cos(\pi (1 + \eps)) = 1 + \Theta(\eps^2)$ for small $\eps$
\end{proof}

\begin{proposition} \label{prop:cong}
Fix any $n = n_1 + n_2$. 
   For any $t \leq \min \{n_1,n_2\}/4$, we have $C_{n} \cong_{\mathcal{U}_{n}^t, S_{n}}
   C_{n_1} \oplus C_{n_2}$
\end{proposition}
\begin{proof}
  We begin by fixing any set $U_i \in \mathcal{U}_n^t$. For any signed graph $\Sigma$ and graph $G$ on $n$ vertices such that the maximum set of vertex disjoint edges in $\Sigma$ is $t \leq \min \{n_1,n_2\}/4$, let $\mathcal{H}_{\Sigma}(G) = \{\sigma  \in S_n \; | \; \P_{\sigma} \AA_{\Sigma} \P_{\sigma}^T  = \AA_{H} , H \subset G \} $ and let $\mathcal{H}^{-1}_{\Sigma}(G) = \{\sigma  \in S_n \; | \; \AA_{\Sigma}  = \P_{\sigma} \AA_{H}\P_{\sigma}^T , H \subset G \} $.  By Corollary \ref{cor:CycleEquiv}, we have $|\mathcal{H}_{\Sigma}(\overline{C}_n)| =| \mathcal{H}_{\Sigma}(\overline{C_{n_1} \oplus C_{n_2} } )|$ whenever $|\Sigma|$ has no set of at least $\min \{n_1,n_2\}/4$ vertex disjoint edges. Since $S_n$ is a group and has unique inverses, we also have$|\mathcal{H}^{-1}_{\Sigma}(\overline{C}_n)| =|\mathcal{H}_{\Sigma}(\overline{C}_n)| = | \mathcal{H}_{\Sigma}(\overline{C_{n_1} \oplus C_{n_2} } )| = | \mathcal{H}_{\Sigma}^{-1}(\overline{C_{n_1} \oplus C_{n_2} } )|$.
  
   We now define a function $\psi_i:S_n \to S_n$ such that $(\P_\sigma \AA_{\overline{C}_n} \P_{\sigma}^T)_{U_i} =(\P_{\psi_i(\sigma)} \AA_{\overline{C_{n_1} \oplus c_{n_2}}} \P_{\psi_i(\sigma)}^T)_{U_i} $ for every $\sigma \in S_n$. Now fix any signed graph $\Sigma$ such that $\Sigma = (\P_\sigma \AA_{\overline{C}_n} \P_{\sigma}^T)_{U_i}$ for some $\sigma \in S_n$. Note that the set of $\pi \in S_n$ such that $\Sigma = (\P_\pi A_{\overline{C}_n} \P_{\pi}^T)_{U_i}$ is precisely $\mathcal{H}^{-1}_{\Sigma}(\overline{C}_n)$. Similarly, the set $\pi \in S_n$ such that $\Sigma = (\P_\pi A_{\overline{C_{n_1} \oplus C_{n_2} }} \P_{\pi}^T)_{U_i}$ is precisely $\mathcal{H}^{-1}_{\Sigma}(\overline{C_{n_1} \oplus C_{n_2} })$. Also, by construction of $\mathcal{U}_{n}^t$, we know that the maximum set of vertex disjoint edges in $U_i$, and therefore in $\Sigma$ is $t \leq \min \{n_1,n_2\}/4$, So by the above, we know there is a bijection $\psi_{i}^{\Sigma} : \mathcal{H}^{-1}_{\Sigma}(\overline{C}_n) \to \mathcal{H}^{-1}_{\Sigma}(\overline{C_{n_1} \oplus C_{n_2} })$ for every such realizable matrix $\Sigma$. Taking $\psi_i(\sigma) = \psi_i^{ (\P_\sigma \AA_{\overline{C}_n} \P_{\sigma}^T)_{U_i}}(\sigma)$ satisfies the desired properties for $\overline{C_{n}} \cong_{\mathcal{U}_{n}^t, S_{n}}
   \overline{C_{n_1} \oplus C_{n_2}}$. Notice that this implies that $C_{n} \cong_{\mathcal{U}_{n}^t, S_{n}}
   C_{n_1} \oplus C_{n_2}$, since $C_{n}$ and $C_{n_1} \oplus C_{n_2}$ are both obtained from obtained  $\overline{C_{n}}$ and $ \overline{C_{n_1} \oplus C_{n_2}}$ by changing every entry with the value $-1$ to $0$. 
   
\end{proof}

\begin{proposition}\label{prop:LBfinal}
Fix any $t > 1$, and set either $d_0 = 4t$, and $d = 2d_0+1$. Set $\BB = 1/2(\AA_{C_{d}} + \lambda \mathbb{I}_{d})$ and $\DD = 1/2(\AA_{C_{d_0} \oplus C_{d_0 + 1}} + \lambda \mathbb{I}_{d})$, where $\lambda = - 2 \cos( \frac{2 \pi d_0 }{2d_0+1})$. Then we have that $\BB$ is PSD, $\lambda_{\min}(\DD) < -\delta$ where $\delta = \Theta(1/d^2)$, $\|\BB\|_\infty , \|\DD\|_\infty \leq 1$, and  $\BB\cong_{\mathcal{U}^t_{d} , S_{d}} \DD$.
\end{proposition}
\begin{proof}
    By Proposition \ref{prop:cong}, we know $C_{2d_0+1} \cong_{\mathcal{U}_{2d_0+1}^t, S_{2d_0+1}}
   C_{d_0} \oplus C_{d_0+1}$, so to show subgraph equivalence suffices to show that adding $\lambda \mathbb{I}_{2d_0+1}$ to both $C_{2d_0+1}$ and $ C_{d_0} \oplus C_{d_0+1}$ does not effect the fact that they are $\mathcal{U}^t_{d_0} , S_{d_0}$ subgraph-equivalent. But note that this fact is clear, since we have only changed the diagonal which is still equal to $\lambda$ everywhere for both $\BB,\DD$. Namely, for any $\sigma,\pi \in S_{2d_0+1}$ and $i \in [2d_0+1]$ we have $\left(\P_{\sigma} \BB \P_{\sigma}^T \right)_{(i,i)}    = \left(\P_{\pi} \DD  \P_{\pi}^T \right)_{(i,i)}  = \lambda$, thus the subgraph equivalence between $C_{2d_0+1}$ and $ C_{d_0} \oplus C_{d_0+1}$ still holds using the same functions $\psi_i$ as required for $C_{2d_0+1} \cong_{\mathcal{U}_{2d_0+1}^t, S_{2d_0+1}}
   C_{d_0} \oplus C_{d_0+1}$. Note that the $L_\infty$ bound on the entries follows from the fact that adjacency matrices are bounded by $1$ and zero on the diagonal, $\lambda \leq 2$, and we scale each matrix down by $1/2$. Next, by Fact \ref{fact:cycle}, we know that $\BB$ is PSD and $\lambda_{\min}(\DD) =- \Theta(\frac{1}{d^2})$, which holds still after scaling by $1/2$, and completes the proof. 
\end{proof}

  We now state our main theorem, which is direct result of instantiating the general lower bound of  Lemma \ref{lem:ifBDthenLB} with the matrices as described above in Proposition \ref{prop:LBfinal}.

\begin{theorem}\label{thm:lbmain}
Any non-adaptive sampling algorithm which solves with probability at least $2/3$ the PSD testing problem with $\eps$-$\ell_2^2$ gap must query at least $\wt{\Omega}( \frac{1}{\eps^2})$ entries of the input matrix. 
\end{theorem}
\begin{proof}

Set $k =  C \frac{1}{\eps \log^6(1/\eps)}$ for a small enough constant $C>0$. Also set $t = \log k$, $d_0 = 4t$. $d = 2d_0 + 1$, and as before set $m = n/(dk)$.
We first apply Lemma \ref{lem:ifBDthenLB} with $\ZZ = \mathbf{1}^{m \times m}$, and the matrices $\BB = 1/2(\AA_{C_{d}} + \lambda \mathbb{I}_{d})$ and $\DD = 1/2(\AA_{C_{d_0} \oplus C_{d_0+1}} + \lambda \mathbb{I}_{d})$ from Proposition \ref{prop:LBfinal}, where $\lambda = - 2 \cos( \frac{2 \pi d_0 }{2d_0+1})$. Then by Lemma \ref{lem:ifBDthenLB}, using that $\BB\cong_{\mathcal{U}^t_{2d_0+1} , S_{2d_0+1}} \DD$ via Proposition \ref{prop:LBfinal}, it follows that any non-adaptive sampling algorithm that distinguishes $\mathcal{D}_1$ from $\mathcal{D}_2$ requires $\Omega(k^2)$ samples. 

We now demonstrate every instance of $\mathcal{D}_1$ and $\mathcal{D}_2$ satisfy the desired $\ell_2^2$-gap as defined in Problem \ref{prob:l2}. First, since the eigenvalues of the Kronecker product $\Y \otimes \ZZ$ of any matrices $\Y,\ZZ$ are all pairwise eigenvalues of the matrices $\Y,\ZZ$, it follows that $\wt{\BB}$ is PSD as $\BB$ is PSD by Proposition \ref{prop:LBfinal} and and $\mathbf{1}^{m \times m} = \mathbf{1}^m (\mathbf{1}^m)^T$ is PSD. By the same fact and Proposition \ref{prop:LBfinal}, since $\lambda_1(\mathbf{1}^{m \times m}) = m$, we have that $\lambda_{\min}(\wt{\DD}) = -\Theta(m/(d^2))  = -\Theta(\frac{n}{d^3 k})$. Now note that if $\AA_1 \sim \mathcal{D}_1$, then $\AA_1$ is a block-diagonal matrix where each block is PSD, thus $\AA_1$ is PSD. Note also that if $\AA_2 \sim \mathcal{D}_2$, then $\AA_2$ is a block-diagonal matrix where each block is has an eigenvalue smaller than $-C'\frac{n}{d^3 k}$ for some constant $C' > 0$. Since the eigenvalues of a block diagonal matrix are the union of the eigenvalues of the blocks, it follows that

\begin{equation}
    \begin{split}
        \sum_{i: \lambda_i(\AA_2) < 0} (\lambda_i(\AA_2))^2 &= \sum_{i=1}^k \lambda_{\min}(\wt{\DD})^2 \\
        &\geq k (C' \frac{n}{d^3 k})^2 \\
        &= \left( \frac{(C')^2 \log^6 (1/\eps) }{C (8\log(\frac{C}{\eps \log^6(1/\eps)}) +1)^6 }\right) \cdot \eps n^2 \\ 
        &\geq \eps n^2 \\
    \end{split}
\end{equation} 
Where the last inequality follows from setting the constant $C = \frac{(C')^2}{100^6}$ so that 
\begin{equation}
    \begin{split}
     \frac{(C')^2 \log^6 (1/\eps) }{C\left(8 \log(\frac{C}{\eps \log^6(1/\eps)}) +1 \right)^6 } &=     \frac{100^6 \log^6 (1/\eps) }{\left(8 \log(\frac{(C')^2}{ 100^6\eps \log^6(1/\eps)}) +1 \right)^6 }  \\ 
    &\geq  \frac{100^6 \log^6 (1/\eps) }{\left(16 \log(\frac{1}{\eps })\right)^6 }  \\ 
        & >1  \\ 
    \end{split}
\end{equation}
and using that the first inequality above holds whenever $\left(\frac{1}{\eps }\right)^{16} \leq  \left(\frac{ (C')^2}{3 \cdot 100^6\eps \log^6(1/\eps)}\right)^{8}$, which is true so long as $\eps < C_0$ for some constant $C_0$. Note that if $\eps > C_0$, then a lower bound of $\Omega(1) = \Omega(1/\eps^2)$ follows from the one heavy eigenvalue $\ell_\infty$ gap lower bound. Thus $\AA_1,\AA_2$ satisfies the $\eps$-$L_2$ gap property as needed, which completes the proof. 

    \end{proof}

 \subsubsection{\texorpdfstring{$C_{n_1+ n_2}$}{Cycle n1 plus n2} is Subgraph Equivalent to \texorpdfstring{$C_{n_1} \oplus C_{n_2}$}{Cycle n1 plus Cycle n2} }\label{sec:CnLemma}
 In this section, we demonstrate the subgraph equivalence of the the cycle $C_{n_1+ n_2}$ and union of cycles $C_{n_1} \oplus C_{n_2}$. In order to refer to edges which are not in the cycles $C_{n_1+ n_2}$ and $C_{n_1} \oplus C_{n_2}$, it will actually be convenient to show that $\overline{C_{n_1+ n_2}}$ is subgraph equivelant to  $\overline{C_{n_1} \oplus C_{n_2}}$, where recall that $\overline{G}$ for a simple graph $G$ is the result of adding negative edges to $G$ for each edge $e = (u,v) \notin E(G)$. Equivalently, the adjacency matrix of $\overline{G}$ is the result of replacing the $0$'s on the off-diagonal of $A_{G}$ with $-1$'s. Notice that, by the definition of subgraph equivalence, it does not matter whether these values are set to $0$ or to $-1$. 
 
 \paragraph{Overview of the bijection. }
 We now intuitively describe the bijection of Lemma \ref{lem:cyclesmain}, which demonstrates that for any singed graph $\Sigma$ such that any set of pairwise vertex disjoint edges $\{e_1,\dots,e_k\}$ (i.e. any matching) in $\Sigma$ has size at most $k \leq \min\{n_1,n_2\}/4$, the number of subgraphs of $\overline{C_{n_1 + n_2}}$ isomorphic to $\Sigma$ is the same as the number of subgraphs of $\overline{C_{n_1} \oplus C_{n_2}}$ isomorphic to $\Sigma$. So let $H$ be any subgraph of $\overline{C_{n_1 + n_2}}$ that is isomorphic to $\Sigma$. For simplicity, let $n_1 = n_2$, and suppose $H$ contains only positive edges, so that $H$ is actually a subgraph of the unsigned cycle $C_{2n}$. Since $\Sigma$ has at most $n/4$ edges, $\Sigma \cong H$ must be a collection of disjoint paths. So the problem can be described as an arrangement problem: for each arrangement $H$ of $\Sigma$ in $C_{2n}$, map it to a unique arrangement $H'$ of $\Sigma$ in $C_{n} \oplus C_n$. 
\begin{figure}
\centering
\includegraphics[scale = 1.1 ]{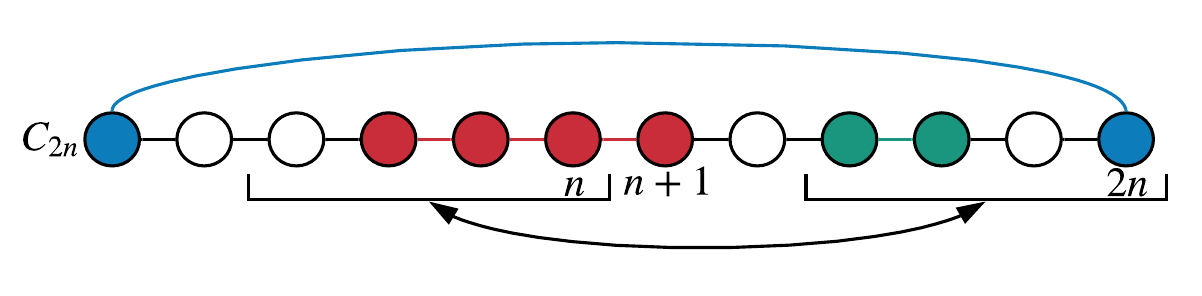}
	\includegraphics[scale = 1.1 ]{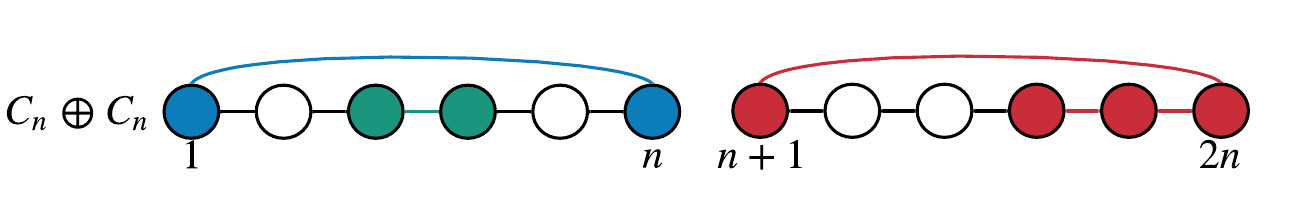}
	\caption{An illustration of the bijection in Lemma \ref{lem:cyclesmain}, when $H$ only contains positive edges. The three colored paths represent the graph $H$, which must be mapped from $C_{2n}$ to $C_n \oplus C_n$. Since the paths intersect the edges $(2n,1)$ and $(n,n+1)$ to be cut, we must first swap the last four vertices $\{n-4,\dots,n\}$ and $\{2n - 4, \dots, 2n\}$ of $C_{2n}$ before the two splitting points $n,2n$, and then cut the cycle. Note that four is the smallest number of vertices which can be swapped, without swapping in the middle of a path of $H$. 
	}
	\label{fig:graph1}
\end{figure}
 
We would like to construct such a mapping by ``splitting'' the big cycle $C_{2n}$ into two smaller cycles, see Figure \ref{fig:graph1} for an example. Specifically, we could split the cycle $C_{2n}$ down the middle, cutting the edges $(n,n+1)$, and $(n,1)$, and instead connecting the first vertex to the $n$-th and the $n+1$-st to the $2n$-th. Now if $H$ does not contain either of the cut edges, then the resulting collection of paths will be an isomorphic copy of $H$ living inside of $C_{n} \oplus C_n$. However, if $H$ does contain such an edge, we cannot cut the cycle here, as the resulting paths inside of $C_n \oplus C_n$ would not be isomorphic. For example see Figure \ref{fig:graph1}, where if we just cut the edge between $(n,n+1)$ and rerouted it to $(n,1)$, then the red cycle with $4$ vertices would be disconnected into a cycle of length three, and an isolated vertex. To handle this, before cutting and rerouting the edges $(n,n+1)$ and $(2n,1)$, we first \textit{swap} the last $i$ vertices before the cutting points, for some $i$. Namely, we swap the vertices $(n-i,n-i+1,\dots,n)$ with $(2n - i, 2n - i + 1, \dots ,2n)$ and \textit{then} split the graph at the edges $(n,n+1)$ and $(2n,1)$. For the resulting graphs to be isomorphic, we cannot swap in the middle of a path, thus the value $i$ is chosen as the \textit{smallest} $i \geq 0$ such that the edges $(n-i-1,n-i)$ and $(2n-i-1,2n - i)$ do not exist in any path of $H$. Moreover, such an $i$ must exist, so long as $H$ has fewer than $\min\{n_1,n_2\}$ edges (the stronger bound of $\min\{n_1,n_2\}/4$ is only needed for the more general case, where negative edges are included).
 
One can show that this mapping is actually an involution; namely, given the collection of paths $H'$ in $C_{n} \oplus C_{n}$ which are obtained from applying the function on $H$, one can similarly find the smallest $i \geq 0$ such that the edges $(n-i-1,n-i)$ and $(2n-i-1,2n - i)$ are not in $H'$, which must in fact be the same value of $i$ used when mapping $H$! Then, by swapping the last $i$ vertices before $n$ and $2n$, and then reconnecting $C_n \oplus C_n$ into a single cycle, one obtains the original graph $H$. From this, demonstrating bijectivity becomes relatively straightforward.
Extending this to the case where $H$ is allowed to contain negative edges of $\overline{C_{2n}}$ follows similar steps, albiet with a stronger condition on the choice of $i$. 
The full proof is now presented below.


 \begin{lemma}\label{lem:cyclesmain}
 Fix any $n = n_1 + n_2$. Fix any simple graph $|\Sigma|$, such that any set of vertex disjoint edges $\{e_1,\dots,e_k\}$ in $|\Sigma|$ has size at most $k \leq \min\{n_1,n_2\}/4$, and let $\Sigma = (|\Sigma|,\sigma)$ be any signing of $|\Sigma|$. Let $\mathcal{F}_{\Sigma}(\overline{C_n})$ denote the set of subgraphs of $\overline{C_n}$ isomorphic to $|\Sigma|$, and similarly define $\mathcal{F}_{\Sigma}(\overline{C_{n_1} \oplus C_{n_2}})$. Then we have 
 \[			\left|\mathcal{F}_{\Sigma}(\overline{C_n})\right| = \left|\mathcal{F}_{\Sigma}(\overline{C_{n_1} \oplus C_{n_2}})\right|	\] 	
 \end{lemma}
\begin{proof}
		Order the vertices's of the cycle $C_n = \{1,2,\dots,n\}$, which we will describe as the same vertex set for $C_{n_1} \oplus C_{n_2}$, where $\{1,\dots,n_1\}$ are the vertices of the first cycle $C_{n_1}$ and $\{n_1 + 1, \dots, n\}$ are the vertices of $C_{n_2}$. 
	We derive a bijection $\varphi:  \mathcal{F}_{\Sigma}(\overline{C_n})\to  \mathcal{F}_{\Sigma}(\overline{C_{n_1} \oplus C_{n_2}})$. 
		We describe a point $\X \in \mathcal{F}_H(\overline{C_n}) \cup  \mathcal{F}_{\Sigma}(\overline{C_{n_1} \oplus C_{n_2}})$ by its (signed) adjacency matrix $\X \in \{-1,0,1\}^{n \times n}$. Namely, $\X \in \{-1,0,1\}^{n \times n}$ is any matrix obtained by setting a subset of the entries of $\AA_{\overline{C_{n_1 + n_2}}}$ or $\AA_{\overline{C_{n_1} \oplus C_{n_2}}} $ equal to $0$, such that the signed graph represented by $\X$ is isomorphic to $\Sigma$. In this following, we will always modularly interpret the vertex $v_{n +i} =v_i$ for $i \geq 1$.
		
		Thus, we can now think of $\varphi$ as being defined on the subset of the matrices $\{-1,0,1\}^{n \times n}$ given by the adjacency matrices of signed graphs in $\mathcal{F}_{\Sigma}(\overline{C_n})$. In fact, it will useful to define $\varphi$ on a larger domain. Let $\mathcal{D} \subset \{-1,0,1\}^{n \times n}$ be the set of all adjacency matrices  for signed graphs $G$ with the property that any set of vertex disjoint edges $\{e_1,\dots,e_k\}$ in $G$ size at most $k \leq \min\{n_1,n_2\}/4$. Notice that $\mathcal{D}$ contains both $ \mathcal{F}_{\Sigma}(\overline{C_n})$ and $\mathcal{F}_{\Sigma}(\overline{C_{n_1} \oplus C_{n_2}})$.  For a given $\X \in \mathcal{D}$, we will define $\varphi(\X) = \P_{\sigma_\X} \X \P_{\sigma_\X}^T$ for some permutation $\sigma_\X$. Since the graph of $ \P_{\sigma_\X} \X \P_{\sigma_\X}^T$ is by definition isomorphic to $\X$, it follows that $ \P_{\sigma_\X} \X \P_{\sigma_\X}^T \in \mathcal{D}$, thus $\varphi$ maps  $\mathcal{D}$ into $\mathcal{D}$. So in order to define the mapping $\varphi(\X)$, it suffices to define a function $\phi: \mathcal{D} \to S_n$ mapping into the symmetric group so that $\varphi(\X) = \P_{\phi(\X)}\X\P_{\phi(\X)}^T$.

		For $i=0,1,2,\dots,\min_{n_1,n_2} -1$, define the permutation $\sigma_i \in S_n$ as follows. 
 For $j \in \{0,1,\dots,n_1 - i\} \cup \{n_1 + 1,\dots,n - i\}$, we set $\sigma_i(j) = j$.
	If $i > 0$, then for each $0 \leq j < i$, we set $\sigma_i(n_1 - j) = n - j$ and $\sigma_i(n - j) = n_1 - j$.  In other words, the function $\sigma_i$  swaps the last $\max\{0,i-1\}$ vertices before the spliting points $n_1,n$ of the cycle. Notice that $\sigma_i$ is an involution, so $\sigma_i(\sigma_i) = \text{id}$ and $\sigma_i = \sigma^{-1}_i$.

	We now define our bijection $\varphi$. For $\X \in \mathcal{D}$, let $i(\X)$ be the smallest value of $i \geq 0$ such that $\X_{n_1 - i, n_1 - i + 1} = \X_{n - i, n - i +1} =\X_{n - i, n_1 - i + 1} = \X_{n_1 - i, n - i +1} = 0$. Equivalently, $i(\X)$ is the smallest value of $i \geq 0$ such that none of the four edges of the cycle $c_i = (v_{n_1 -i}, v_{n_1 -i+1}, v_{n -i}, v_{n - i + 1})$ exist in $\X$.  We then define $\varphi(\X) = \sigma_{i(\X)} = \sigma_\X$, so that $\varphi(\X) =  \P_{\sigma_{i(\X)}}\X\P_{\sigma_{i(\X)}}^T$. 
	Note that if the maximum number of vertex disjoint edges in $\X$ is at most $\min\{n_1,n_2\}/4$, then $i(\X)$ must always exist and is at most $\min\{n_1,n_2\}/2+1$. This can be seen by the fact that for each $i$ such that $i(\X) > i+1$, there must be at least one edge with endpoints in the set $\{v_{n_1 - i}, v_{n_1-i+1}, v_{n - i}, v_{n-i+1}\}$, thus for each $i \geq 0$ with $i< i(\X)$ we can assign an edge $e_{i,}$ such that $e_{0}, e_{2},e_4, \dots, e_{i(\X)-1}$ are vertex disjoint. 
	
	We must first argue that if $\X \in \mathcal{F}_{\Sigma}(\overline{C_n})$, then $\varphi(\X) \in\mathcal{F}_{\Sigma}(\overline{C_{n_1} \oplus C_{n_2}})$, namely that the function maps into the desired co-domain. To do this, we must show that for every  $(i,j)$ with $(\P_{\sigma_\X}\X\P_{\sigma_\X}^T)_{i,j} \neq 0$, we have  $(\P_{\sigma_\X}\X\P_{\sigma_\X}^T)_{i,j} = (\AA_{\overline{C_{n_1} \oplus C_{n_2} }})_{i,j}$. This is equivalent to showing that for any signed edge $e = (v_i,v_j) \in \X$, $v_i,v_j$ are connected in $C_n$ if and only if $v_{\sigma_\X(i)},v_{\sigma_\X(j)}$ are connected in $C_{n_1} \oplus C_{n_2}$.  In the proof of this fact, we will only use that $\X \in \mathcal{D}$. 
		
	So suppose $v_i,v_j$ were connected in $C_n$, and wlog $j > i$. First suppose that $i \notin \{n,n_1\}$. Then we have $j=i+1$. Since  $e = (v_i,v_j) \in \X$ is an edge of the subgraph, we know $i \notin \{ n - i(\X), n_1 - i(\X)\}$ by construction of $i(\X)$. Thus $(v_{\sigma_\X(i)},v_{\sigma_\X(j)}) = (v_{i'} , v_{i'+1})$ for some $i' \notin \{n_1,n\}$, which is always an edge of $C_{n_1} \oplus C_{n_2}$. If $i =  n$, then $j = 1$, and we have $i(\X) > 0$, so $\sigma(i) = n_1$ and $\sigma(j) = 1$, and $(v_{n_1}, v_{1}) $ is an edge of $C_{n_1} \oplus C_{n_2}$. Similarly, if $i =  n_1$, then $j = n_1 + 1$, and  since again necessarily $i(\X)>0$ we have $\sigma(i) = n, \sigma(j) = n_1 + 1$, and $(v_{n},v_{n_1 + 1})$ is an edge of $C_{n_1} \oplus C_{n_2}$.	
	We now consider the case where $(v_i,v_j) \in \X$ is not an edge in $C_n$. Suppose for the sake of contradiction that $(v_{\sigma_\X(i)} ,v_{\sigma_\X(j)})$ is an edge in $C_{n_1} \oplus C_{n_2}$.  WLOG, $i,j$ are in the first cycle $C_{n_1}$. We can write $\sigma_\X(i) = i', \sigma_\X(j) = i' + 1$ for some $i' \in \{1,2,\dots,n_1 \}$, where $i'+1$ is interpreted as $1$ if $i' = n_1$. If $i' \leq i(\X) -1$, then both $ i' = i$ and $ i'+1 = i+1 = j$, but $(v_{i}, v_{i+1})$ is also connected in $C_n$.  If $i'\geq i(\X) +1$, then $i' = i + n_2$ and $i'+1 = i + n_2 +1$ (where $i + n_2 + 1$ is interpreted modularly as $1$ if $i =n_1$), and again $v_{i+n_2}$ and $v_{i+n_2 + 1}$ are connected in $C_n$. Finally, if $i' = i(\X)$, then $i= i'$ and $j = i+  n_2 + 1$, but then we cannot have  $(v_i,v_j) \in \X$ by construction of $i(\X)$, which completes the of the claim that $\varphi$ maps $\mathcal{F}_{\Sigma}(\overline{C_n})$ into $\mathcal{F}_{\Sigma}(\overline{C_{n_1} \oplus C_{n_2}})$.

	We now show that $\varphi$ is injective. 
To do this, we show that $\varphi(\varphi(\X)) = \X$ for any $\X \in \mathcal{D}$ -- namely that $\varphi$ is an involution on $\mathcal{D}$. This can be seen by showing that we always have $i(\X) = i(\varphi(\X))$. To see this, observe that $i(\X)$ is defined as the first $i \geq 0$ such that none of the four edges of the cycle $c_i = (v_{n_1 -i}, v_{n_1 -i+1}, v_{n -i}, v_{n - i + 1})$ exist in $\X$. Thus it suffices to show that for each $\min \{n_1,n_2\}-1 > i \geq 0$, the number of edges in $c_i$ is preserved after permuting the vertices by $\sigma_{i(\X)}$. To see this, note that if $i(x) > i$, then $(\sigma_{i(\X)}(v_{n_1 -i}),\sigma_{i(\X)}( v_{n_1 -i + 1}),\sigma_{i(\X)} (v_{n -i}),\sigma_{i(\X)}( v_{n - i + 1})) = (v_{n -i}, v_{n -i+1}, v_{n_1 -i}, v_{n_1 - i + 1}) $, which is the same cycle. If $i(\X) < i$, then $\sigma_{i(\X)}$ does not move any of the vertices in $c_i$. Finally, if $i(\X) = i$, then $(\sigma_{i(\X)}(v_{n_1 -i}),\sigma_{i(\X)}( v_{n_1 -i + 1}),\sigma_{i(\X)} (v_{n -i}),\sigma_{i(\X)}( v_{n - i + 1})) = (v_{n_1 -i}, v_{n -i+1}, v_{n -i}, v_{n_1 - i + 1})$, which again is the same cycle $c_i$ (just with the ordering of the vertices reversed). So $\varphi(\varphi(\X)) = \X$ for any $\X \in\mathcal{D}$, so in particular $\varphi:\mathcal{F}_{\Sigma}(\overline{C_n}) \to \mathcal{F}_{\Sigma}(\overline{C_{n_1} \oplus C_{n_2} }) $ is injective.

To show surjectivity, it suffices to show that if $\X  \in \mathcal{F}_{\Sigma}(\overline{C_{n_1} \oplus C_{n_2}})$ then $\varphi(\X) \in\mathcal{F}_{\Sigma}(\overline{C_{n}})$. Namely, that $\varphi$ can also be defined as a valid function $\varphi: \mathcal{F}_{\Sigma}(\overline{C_{n_1} \oplus C_{n_2}}) \to \mathcal{F}_{\Sigma}(\overline{C_{n}})$. Again, this is equivalent to showing that for any signed edge $e = (v_i,v_j) \in \X$, $v_i,v_j$ are connected in $C_{n_1} \oplus C_{n_2}$ if and only if $v_{\sigma_\X(i)},v_{\sigma_\X(j)}$ are connected in $C_{n}$. Since $\sigma_\X$ is an involution, this is the same as asking that for any signed edge $e = (v_i,v_j) \in \X$, $v_{\sigma_\X(\sigma_\X(i))},v_{\sigma_\X(\sigma_\X(j))}$ are connected in $C_{n_1} \oplus C_{n_2}$ if and only if $v_{\sigma_\X(i)},v_{\sigma_\X(j)}$ are connected in $C_{n}$. Setting $i' = \sigma_\X(i), j' =\sigma_\X(j)$, this states that for all signed edges $(v_{i'},v_{j'}) \in \P_{\sigma_\X} \X \P_{\sigma_\X} = Y \in \mathcal{D}$, we have that $v_{i'},v_{j'}$ are connected in $C_{n}$ if and only if $v_{\sigma_\X(i')},v_{\sigma_\X(j')}$ are connected in $C_{n_1} \oplus C_{n_2}$. But as shown above, we have that $i(\X) = i(\varphi(\X))$, so $\sigma_\X = \sigma_Y$,  and then this fact was already proven above for any $Y \in \mathcal{D}$, which completes the proof.

\end{proof}

Now for any signed graph $\Sigma$ on $n$ vertices, let $\AA_{\Sigma}$ be its adjacency matrix. Note that we can equivalently define via $\mathcal{F}_{\Sigma}(\overline{C_n}) = \{ H \subseteq \overline{C_n}, \; | \; \P_{\sigma} \AA_{\Sigma} \P_{\sigma}^T  = \AA_{H} , \sigma \in S_n \}$. Here $ H \subseteq \overline{C_n}$ means $H$ is a subgraph of $\overline{C}_n$.  On the other hand, we may be interested in the potentially much larger set of all possible permutations $\sigma$ such that  $\P_{\sigma} \AA_{\Sigma} \P_{\sigma}^T = \AA_{H}$ for some $H \subset \overline{C_n}$. So define $\mathcal{H}_{\Sigma}(\overline{C_n}) = \{\sigma  \; | \; \P_{\sigma} \AA_{\Sigma} \P_{\sigma}^T  = \AA_{H} , H \subset \overline{C_n}, \sigma \in S_n \} $. It is not difficult to show that $|\mathcal{H}_{\Sigma}(\overline{C_n})| = |\text{Aut}(\Sigma)| |\mathcal{F}_{\Sigma}(\overline{C_n})|$, where $\text{Aut}(\Sigma)$ is the set of (signed) graph automorphisms of $\Sigma$.

\begin{fact}
	We have $|\mathcal{H}_{\Sigma}(\overline{C_n})| = |\text{Aut}(\Sigma)| |\mathcal{F}_{\Sigma}(\overline{C_n})|$.
\end{fact}
\begin{proof}
	Fix any $H \subset \overline{C_n}$ such that  $\P_{\sigma} \AA_{\Sigma} \P_{\sigma}^T  = \AA_{H}$ for some $\sigma \in S_n$. We show that there are exactly $|\text{Aut}(\Sigma)|$ elements $\sigma' \in S_n$ such that $\P_{\sigma'} \AA_{\Sigma} \P_{\sigma'}^T  = \AA_{H}$. By definition, $\text{Aut}(\Sigma)$ is the set of permutations $\pi \in S_n$ with $\P_{\sigma} \AA_{\Sigma} \P_{\sigma}^T = \AA_{\Sigma}$. For every $\pi \in \text{Aut}(\Sigma)$, we have $\P_{\sigma} \P_{\pi} \AA_{\Sigma} \P_{\pi}\P_{\sigma}^T  =\P_{\sigma \circ \pi } \AA_{\Sigma} \P_{\sigma \circ \pi}^T = \AA_{H}$, and moreover the set of elements $|\{ \sigma \circ \pi \; | \; \pi \in \text{Aut}(\Sigma) \}| = |\text{Aut}(\Sigma)|$ since $S_n$ is a group. Now suppose we have some $\lambda \in S_n$ such that  $\P_{\lambda} \AA_{\Sigma} \P_{\lambda}^T  = \AA_{H}$  and $\lambda \notin \{ \sigma \circ \pi \; | \; \pi \in \text{Aut}(\Sigma) \}$. Then  $\P_{\sigma} \AA_{\Sigma} \P_{\sigma}^T = \P_{\lambda} \AA_{\Sigma} \P_{\lambda}^T$, so $\P_{\sigma^{-1} \circ \lambda} \AA_{\Sigma} \P_{\sigma^{-1} \circ \lambda}^T = \AA_{\Sigma} $, which by definition implies that $\sigma^{-1} \circ \lambda =x$ for some $x \in \text{Aut}(\Sigma)$. Thus  $ \lambda =  \sigma \circ x \in \{ \sigma \circ \pi \; | \; \pi \in \text{Aut}(\Sigma) \}$, which is a contradiction. 
\end{proof}

\begin{corollary}\label{cor:CycleEquiv}
	 Fix any $n = n_1 + n_2$. Fix any simple graph $|\Sigma|$, such that any set of vertex disjoint edges $\{e_1,\dots,e_k\}$ in $|\Sigma|$ has size at most $k \leq \min\{n_1,n_2\}/4$, and let $\Sigma = (|\Sigma|,\sigma)$ be any signing of $|\Sigma|$. Let $\mathcal{F}_{\Sigma}(\overline{C_n})$ denote the set of subgraphs of $C_n$ isomorphic to $|\Sigma|$, and similarly define $\mathcal{F}_{\Sigma}(\overline{C_{n_1} \oplus C_{n_2}})$. Then we have 
	\[			\left|\mathcal{H}_{\Sigma}(\overline{C_n})\right| = \left|\mathcal{H}_{\Sigma}(\overline{C_{n_1} \oplus C_{n_2}})\right|	\] 	
\end{corollary}

 \subsection{Lower Bounds for Schatten, Ky-Fan, and Tail Error Testing}
 In this section, we demonstrate how our construction of subgraph equivalent matrices with gaps in their spectrum result in lower bounds for a number of other spectral testing problems via Lemma \ref{lem:ifBDthenLB}. We begin by proving a lower bound for testing Schatten norms. To do this, we must first demonstrate that there is a gap in the Schatten $1$ norm between a cycle and the union of two disjoint cycles.

 \begin{fact}[Theorem 1 of \cite{knapp2009sines}]\label{fact:cos}
     Fix any $a ,b,n \in \R$ with $\sin(b/2) \neq 0$. Then we have
     \[\sum_{k=0}^{n-1} \cos(a + kb )= 
     \frac{\sin(\frac{n b}{2})}{ \sin(\frac{b}{2}) }\cos\left(a + \frac{(n-1)b}{2} \right) \]
 \end{fact}

\begin{proposition}
Fix any $d \geq 6$ be any integer divisible by $4$. Then
\[ \|C_d\|_{\mathcal{S}_1} = 4 \cdot \frac{\cos \left( \pi/d \right)}{\sin(\pi/d)}\]
\end{proposition}
\begin{proof}
    By \cite{chung1996lectures}, for any $d \geq 3$ the eigenvalues of $C_d$ are given by $2 \cdot \cos( \frac{2 \pi j}{d})$ for $j=0,1,\dots,d-1$. Let $a_1 = \lfloor d/4 \rfloor, a_2 = \lfloor 3d/4 \rfloor, a_3 = d - a_2 -1$. 
    
    \begin{equation}
        \begin{split}
            \|C_{d}\|_1 & =2 \sum_{j=0}^{d-1} \left|\cos\left( \frac{2 \pi j}{d}\right) \right| \\ 
          & =2 \left( \sum_{j=0}^{a_1  } \cos\left( \frac{2 \pi j}{d}\right)  - \sum_{j=a_1 + 1}^{a_2 }  \cos\left( \frac{2 \pi j}{d} \right)+ \sum_{j=a_2 + 1 }^{d-1} \cos\left( \frac{2 \pi j}{d} \right) \right)    \\ 
                & =2 \left( \sum_{j=-a_3}^{a_1  } \cos\left( \frac{2 \pi j}{d}\right)  - \sum_{j=a_1 + 1}^{a_2 }  \cos\left( \frac{2 \pi j}{d} \right) \right)    \\
        \end{split}
    \end{equation}
We analyze each term in the above via Fact \ref{fact:cos}. Firstly:

\begin{equation}
    \begin{split}
             \sum_{j=-a_3}^{a_1  } \cos\left( \frac{2 \pi j}{d}\right) &=  \sum_{j=0}^{a_1 + a_3  } \cos\left( \frac{2 \pi j}{d}  -\frac{2 \pi a_3 }{d}\right) \\
             &= \frac{\sin((a_1 + a_3 + 1) \pi/d )}{\sin(\pi/d)} \cos\left( \frac{(a_1 + a_3) \pi}{d} -\frac{2 \pi a_3 }{d}\right)\\ 
    \end{split}
\end{equation}

Note that if $d$ is divisible by $4$, the above becomes  $2 \cos(\pi/d)/\sin(\pi/d)$.  
   Next, for the second term, we have
   \begin{equation}
     \begin{split}
                 \sum_{j=a_1 + 1}^{a_2 }  \cos\left( \frac{2 \pi j}{d} \right) &=   \sum_{j=0}^{a_2 - a_1 - 1 }  \cos\left( \frac{2 \pi j}{d} - \frac{2\pi(a_1 + 1) }{d} \right)\\
                 &= \frac{\sin((a_2 - a_1 ) \pi/d )}{\sin(\pi/d)} \cos\left( \frac{(a_2 - a_1 - 1) \pi }{d} - \frac{2\pi(a_1 + 1) }{d}  \right)\\
       \end{split}
   \end{equation}

Again, note that if $d$ is divisible by $4$, the above becomes  $2 \cos(\pi/d)/\sin(\pi/d)$. 
 Putting these two equations together, we have that
   \[  \|C_d\|_1 = 4 \cdot \frac{\cos \left( \pi/d \right)}{\sin(\pi/d)} \]

      \end{proof}
 
 \begin{proposition}\label{prop:schattengap}
     Fix any $d$ larger than some constant. Then we have 
     
     \[ \left| \| \CC_{8d}\|_{\mathcal{S}_1} - \| \CC_{4d} \oplus \CC_{4d} \|_{\mathcal{S}_1}  \right| \gtrsim \frac{1}{d^3} \]
 \end{proposition}
 \begin{proof}
     By the prior Lemma, we have $\| \CC_{d}\|_{\mathcal{S}_1} = 4  \cot( \pi/d)$ for any $d$ divisible by $4$. Thus using the Taylor expansion of cotangent, we have
     
     \begin{equation}
         \begin{split}
          \| \CC_{8d}\|_{\mathcal{S}_1} &= 4 \left( \frac{8d}{\pi } + \frac{ \pi}{24d} + \frac{\pi^3}{ 45 \cdot 512 \cdot d^3} + O(1/d^5) \right)
         \end{split}
     \end{equation}
     and
     \begin{equation}
         \begin{split}
          \| \| \CC_{4d} \oplus \CC_{4d} \|_{\mathcal{S}_1} &=2  \|\CC_{4d} \|_{\mathcal{S}_1} \\
          & = 4 \left( \frac{8d}{\pi } + \frac{ \pi}{24d} + \frac{\pi^3}{ 45 \cdot 128 \cdot d^3} + O(1/d^5) \right)
         \end{split}
     \end{equation}
     Thus 
     \begin{equation}
         \begin{split}
              \left| \| \CC_{8d}\|_{\mathcal{S}_1} - \| \CC_{4d} \oplus \CC_{4d} \|_{\mathcal{S}_1}  \right|& \gtrsim \frac{1}{d^3}
         \end{split}
     \end{equation}
     
     \end{proof}
 
 \begin{theorem}\label{thm:schattenlb}
 Fix any $\frac{1}{\sqrt{n}} \leq \eps \leq 1$. Then given $\AA \in \R^{n \times n}$ with $\|\AA\|_\infty \leq 1$, any non-adaptive sampling algorithm which distinguishes between the cases
 \begin{enumerate}
     \item  $\|\AA\|_{\mathcal{S}_1} > \eps_0 n^{1.5}$
       \item  $\|\AA\|_{\mathcal{S}_1} <\eps_0 n^{1.5} - \eps n^{1.5} $
 \end{enumerate}
with probability at least $3/4$, where $\eps_0 = \wt{\Theta}(\eps)$, must query at least $\tilde{\Omega}(1/\eps^4)$ entries of $\AA$.
 \end{theorem}

 \begin{proof}
     We use the hard instance $\mathcal{D}_1,\mathcal{D}_2$ as earlier. Set $k = C \frac{1}{\eps^2 \log^9(1/\eps)}$, $t = \log k$, and $d = 4k$, and $m=n/(dk)$.  
     We instantiate the matrices $(\BB,\DD,\ZZ)$ in the hard instance via $\BB = C_{2d}, \DD = \C_{d} \oplus C_d$, and let $\ZZ = \delta_{i,j}$ for $i \leq j$, where $\delta_{i,j} \in \{-1,1\}$ are i.i.d. Bernoulli random variables, so that $\ZZ \in \R^{m \times m}$ is a symmetric random Bernoulli matrix. Using the fact that $\|\ZZ\|_2 \leq O(\sqrt{n})$ with high probability \cite{vershynin2010introduction}, along with the fact that $\|\ZZ\|_F^2 = n^2$ deterministically, we have that $\|\ZZ\|_{\mathcal{S}_1} > C_1 m^{1.5}$ with non-zero probability for some constant $C_1 > 0$, as the former two facts imply that $\ZZ$ has $\Omega(n)$ eigenvalues with magnitude $\Theta(\sqrt{n})$. 
     Thus, we can deterministically fix $\ZZ$ to be such a matrix with $\{1,-1\}$ entries such that $\|\ZZ\|_{\mathcal{S}_1} \geq C_1 m^{1.5}$. Given this, we have $\|\wt{\BB}\|_{\mathcal{S}_1} = \|\BB \otimes \ZZ \|_{\mathcal{S}_1} = \|\BB\|_{\mathcal{S}_1}  \cdot \|\ZZ\|_{\mathcal{S}_1}$, and so by Proposition \ref{prop:schattengap}, we have
     \[ \left|  \|\wt{\BB}\|_{\mathcal{S}_1} - \|\wt{\DD}\|_{\mathcal{S}_1}     \right|  \geq C_0\frac{m^{1.5}}{d^3}   \]
     for some absolute constant $C_0 \geq 0$. Note also that we have  $\|\BB \|_{\mathcal{S}_1} > \Omega(d)$, where we use the fact that a constant fraction of the eigenvalues $2 \cdot \cos( \frac{2 \pi j}{d})$ for $j=0,1,\dots,d-1$ of $\BB$ are $\Omega(1)$. Thus we have $\|\wt{\BB}\|_{\mathcal{S}_1} = dm^{1.5}$.

   Now by Proposition \ref{prop:cong}, we obtain that $\BB \cong_{\mathcal{U}_{2d}^t, S_{2d} } \DD$, and thus $\wt{\BB} = \BB \otimes \ZZ \cong_{\mathcal{U}_{2d}^t, \Gamma_{2d,2dm}} \DD \otimes \ZZ = \wt{\DD}$ by Lemma \ref{len:kron}. Thus by Lemma \ref{lem:ifBDthenLB}, we have that distinguishing $\mathcal{D}_1$ from $\mathcal{D}_1$ requires $\Omega(k^2) = \tilde{O}(1/\eps^4)$ samples for any non-adaptive algorithm. It suffices then to show that if $\AA_1 \sim \mathcal{D}_1$ and $\AA_2 \sim \mathcal{D}_2$, then we have the desired gap in Schatten norms. We have
   
   \begin{equation}
       \begin{split}
           \left| \|\AA_1\|_{\mathcal{S}_1} -\| \AA_2\|_{\mathcal{S}_1}    \right| &\geq \sum_{i=1}^k C_0\frac{m^{1.5}}{d^3}   \\
           & \geq C_0 \frac{n^{1.5}}{d^{4.5} k^{1/2}} \\
             & \geq \eps n^{1.5} \\
       \end{split}
   \end{equation} 
   Where the last inequality follows setting $C$ large enough, and assuming that $1/\eps$ is larger than some constant as in Theorem \ref{thm:lbmain}. Again, if $1/\eps$ is not larger than some constant, a $\Omega(1)$ lower bound always applies, since an algorithm must read at least one entry of the matrix to have any advantage.  Now note that we also have  $\|\AA_1\|_{\mathcal{S}_1}  = k \|\wt{\BB}\|_{\mathcal{S}_1} = kdm^{1.5} = n^{1.5} /\sqrt{dk} = \wt{\Theta}(\eps n^{1.5})$ as desired.   To complete the proof, we can scale down all the entries of the input matrix by $1/2$, which results in the required bounded entry property, and only changes the gap by a constant factor.  
     \end{proof}
     We now present our lower bound for testing Ky-Fan norms. Recall that for a matrix $\AA \in \R^{n \times n}$ and $1 \leq s \geq n$, the Ky-Fan $s$ norm is defined as $\|\AA\|_{KF(s)} = \sum_{i=1}^k \sigma_i(\AA)$, where $\sigma_i(\AA)$ is the $i$-th singular value of $\AA$.
     
 \begin{theorem}\label{thm:kflb}
  Fix any $1 \leq s \leq n/(\poly \log n)$. Then there exists a fixed constant $c > 0$ such that given $\AA \in \R^{n \times n}$ with $\|\AA\|_\infty \leq 1$, any non-adaptive sampling algorithm which distinguishes between the cases
 \begin{enumerate}
     \item  $\|\AA\|_{KF(s)} >  \frac{c}{\log(s)}n $
       \item  $\|\AA\|_{KF(s)} < (1-\eps_0) \frac{c}{\log(s)} n $
 \end{enumerate}
with probability at least $3/4$, where $\eps_0 = \Theta(1/\log^2(s))$, must query at least $\tilde{\Omega}(s^2)$ entries of $\AA$.\footnote{$\wt{\Omega}$ hides $\log(s)$ factors here.}
 \end{theorem}
\begin{proof}
The proof is nearly the same as the usage of the hard instance in Theorem \ref{thm:lbmain}. 
Set $k=s$, and let $d_0 = \Theta(\log s)$ and $d = 2d_0 + 1$.
We apply Lemma \ref{lem:ifBDthenLB} with the hard instance as instantiated with $\ZZ = \mathbf{1}^{m \times m}$, and the matrices $\BB = 1/4(\AA_{C_{d}} - 2\mathbb{I}_{d})$ and $\DD = 1/4(\AA_{C_{d_0} \oplus C_{d_0+1}} -2\mathbb{I}_{d})$. Notice that since the eigenvalues of $C_d$ are given by $2 \cdot \cos( \frac{2 \pi j}{d})$ for $j=0,1,\dots,d-1$ \cite{chung1996lectures}, we have  $\lambda_{\min}(\AA_{C_{d}}) = -2 \cos(\frac{2 \pi d_0}{2d_0 + 1}) = -2 + \Theta(1/\log^2(1/\eps))$, $\lambda_{\min}(\AA_{C_{d_0} \oplus C_{d_0+1}} ) = -2$,  and  $\lambda_{\max}(\AA_{C_{d}}) =  \lambda_{\max}(\AA_{C_{d_0} \oplus C_{d_0+1}} ) = 2$. Thus $\|\DD\|_2 = 4$ and $\|\BB\|_2 = 4 - \Theta(1/d^2)$, and moreover $\|\DD \otimes \ZZ\|_2 = 4m$ $\|\BB \otimes \ZZ\|_2 = 4m(1 - \Theta(1/\log^2(1/\eps)))$. Thus if $\AA_1 \sim \mathcal{D}_1$, we have $\|\AA_1\|_{KF(s)} > \sum_{i=1}^k 4m  = 4km$, and  $\|\AA_2\|_{KF(s)} < 4km(1 - \Theta(1/\log^2(1/\eps)))$. The proof then follows from the $\Omega(k^2)$ lower bound for this hard instance via  Lemma \ref{lem:ifBDthenLB}.

\end{proof} 

      We now present our lower bound for testing the magnitude of the $s$-tail $\|\AA - \AA_s \|_{F}^2$, where $\AA_s = \U \Sigma_s \V^T$ is the truncated SVD (the best rank-$s$ approximation to $\AA$). Note that $\|\AA - \AA_s \|_{F}^2 = \sum_{j > s}\sigma_j^2(\AA)$. 
 \begin{theorem}\label{thm:taillb}
  Fix any $1 \leq s \leq n/(\poly \log n)$. Then there exists a fixed constant $c> 0$ (independent of $\eps)$, such that given $\AA \in \R^{n \times n}$ with $\|\AA\|_\infty \leq 1$, any non-adaptive sampling algorithm which distinguishes between the cases
 \begin{enumerate}
     \item  $\|\AA - \AA_s \|_{F}^2 >   \frac{c}{\log(s)} \cdot \frac{n^2}{s} $
       \item  $\|\AA - \AA_s \|_{F}^2 <  (1-\eps_0)\cdot  \frac{c}{\log(s)}\cdot  \frac{n^2}{s}$
 \end{enumerate}
with probability at least $3/4$,  where $\eps_0 = \wt{\Theta}(1)$,  must query at least $\tilde{\Omega}(s^2)$ entries of $\AA$.
 \end{theorem}
\begin{proof}
    We set $s = k$, and use the same hard instance as in Theorem \ref{thm:kflb} above. Note that if $\mathcal{D}_1,\mathcal{D}_2$ are defined as in Theorem \ref{thm:kflb}, if $\AA_1 \sim \mathcal{D}_1$, $\AA_s \sim \mathcal{D}_s$, we have $\sum_{i=1}^s \lambda_i(\AA_1) = s(4m)^2 = 16n^2/(sd^2)$ and $\sum_{i=1}^s \lambda_i(\AA_2) = 16 n^2/(sd^2)(1-\Theta(1/\log^2 s))$.   Now note that $\|\AA_1\|_F^2 = \|\AA_2\|_F^2 = k d m^2 = n^2 / (dk) = n^2/(ds)$, using that each of the single cycle and union of two smaller cycles has $d$ edges, so the Frobenius norm of each block is $d m^2$ in both cases. Using that $d = \Theta(\log s)$, we have that if  $\|(\AA_1) - (\AA_1)s \|_{F}^2 > n^2/(ds) - 16n^2/(sd^2) = c \frac{n^2}{s \log(s)}$ for some constant $c> 0$, and $\|(\AA_2) - (\AA_2)s \|_{F}^2 > n^2/(ds) - 16 n^2/(sd^2)(1-\Theta(1/\log^2 s)) = c \frac{n^2}{s \log(s)} + \wt{\Theta}(\frac{n^2}{s})$, which completes the proof after applying Lemma \ref{lem:ifBDthenLB}. 
\end{proof}


\subsection{Lower Bound For Estimating Ky-Fan of \texorpdfstring{$\AA \AA^T$ }{A A-transpose} via Submatrices}
In this section, we demonstrate a $\Omega(1/\eps^4)$ query lower bound for algorithms which estimate the quantity $\sum_{i=1}^k \sigma_i^2(\AA) = \|\AA\AA^T \|_{KF(k)}$ for any $k \geq 1$ by querying a sub-matrix.  The following lemma as a special case states that for $\eps = \Theta(1/\sqrt{n})$, additive $\eps n^2$ approximation of $\|\AA\AA^T \|_{KF(k)} $ requires one to read the entire matrix $\AA$. 

\begin{lemma}\label{lem:lowerboundAA}
Fix any $1 \leq k \leq n$, and fix any $\frac{100}{\sqrt{n}} \leq \eps \leq 1/4$. 
Any algorithm that queries a submatrix $\AA_{S \times T}$ of $\AA \in \R^{n \times n}$ with $\|\AA\|_\infty \leq 1$ and distinguishes with probability at least $4/5$ between the case that either:
\begin{itemize}
    \item  $\sum_{i=1}^k \sigma_i^2(\AA) > n^2 / 2 + \eps n^2$.
     \item  $\sum_{i=1}^k \sigma_i^2(\AA)  \leq n^2 / 2 $
\end{itemize}
 must make $|S|\cdot |T| = \Omega(1/\eps^4)$ queries to the matrix $\AA$. 
\end{lemma}
\begin{proof}
We design two distributions $\mathcal{D}_1,\mathcal{D}_2$.  If $\AA_1 \sim \mathcal{D}_1$, we independently set each row of $\AA_1$ equal to the all $1's$ vector with probability $p_1 = 1/2 + 2  \eps$, and then return either $\AA = \AA_1$ or $\AA = \AA_1^T$ with equal probability.  If $\AA_2 \sim \mathcal{D}_2$, we independently set each row of $\AA_2$ equal to the all $1's$ vector with probability $p_2 = 1/2 -  2\eps$, and then return either $\AA = \AA_2$ or $\AA = \AA_2^T$ with equal probability. Our hard instance then draws $\AA \sim \frac{\mathcal{D}_1 + \mathcal{D}_2}{2}$ from the mixture. First note that in both cases, we have $\|\AA\|_2^2 = \|\AA\|_F^2 = \sum_{i=1}^k \sigma_i^2(\AA)$, since the matrix is rank $1$. Since $\frac{100}{\sqrt{n}} \geq \eps$, by Chernoff bounds, we have that if  $\AA_1 \sim \mathcal{D}_1$ then $\sum_{i=1}^k \sigma_i^2(\AA) > n^2 / 2 + \eps n^2$ with probability at least $99/100$. Similarly, we have that if  $\AA_2 \sim \mathcal{D}_2$ then $\sum_{i=1}^k \sigma_i^2(\AA) \leq n^2$ with probability at least $99/100$.  

Now suppose that such an algorithm sampling $|S|\cdot |T|  < \frac{c^2}{ \eps^4}$ entries exists, for some constant $c>0$. Then by Yao's min-max principle, there is a fixed submatrix $S,T \subset [n]$ such that, with probability $9/10$ over the distribution $\frac{\mathcal{D}_1 + \mathcal{D}_2}{2}$, the algorithm correctly distinguishes $\mathcal{D}_1$ from $\mathcal{D}_2$ given only $A_{S \times T}$. Suppose WLOG that $|S| \leq \frac{c}{\eps^2}$. Then consider the case only when $\AA_1$ or $\AA_2$ is returned by either of the distributions, and not their transpose, which occurs with probability at least $1/2$. Then $A_{S \times T}$ is just a set of $|S|$ rows, each of which are either all $0$'s or all $1$'s. Moreover, each row is set to being the all $1$'s row independently with probability $p_1$ in the case of $\mathcal{D}_2$, and $p_2$ in the case of $\mathcal{D}_2$. Thus, by Independence across rows, the behavior of the algorithm can be assumed to depend only on the number of rows which are set to $1$. Thus, in the case of $\mathcal{D}_1$ the algorithm receives $X_1 \sim \texttt{Bin}(|S|,p_1)$ and in  $\mathcal{D}_2$ the algorithm receives $X_2 \sim \texttt{Bin}(|S|,p_2)$. Then if $d_{TV}(X_1,X_2)$ is the total variational distance between $X_1,X_2$, then by Equation 2.15 of \cite{adell2006exact}, assuming that $\eps \sqrt{|S|} $ is smaller than some constant (which can be obtained by setting $c$ small enough), we have
\[  d_{TV}(X_1,X_2) \leq O( \eps \sqrt{|S|} ) \]    
Which is at most $1/100$ for $c$ a small enough constant. Thus any algorithm can correctly distinguish these two distributions with advantage at most $1/100$. Since we restricted our attention to the event when rows were set and not columns, and since we conditioned on the gap between the norms which held with probability $99/100$, it follows that the algorithm distinguishes $\mathcal{D}_1$ from $\mathcal{D}_2$ with probability at most $1/2 + 1/4 + (2/100) < 4/5$, which completes the proof. 
\end{proof}

\section{Conclusion}\label{sec:conclusion}

In this work, we gave an optimal (up to $\log(1/\eps)$ factors) algorithm for testing if a matrix was PSD, or was far in spectral norm distance from the PSD cone. In addition, we gave a query efficient algorithm for testing if a matrix was PSD, or was $\eps n^2$ far from the PSD-cone in $\ell_2^2$ distance. Furthermore, we established a new technique for proving lower bounds based on designing ``subgraph-equivelant'' matrices. We believe that this technique is quite general, as shown by its immediate application to lower bounds for the Schatten-$1$ norm, Ky-Fan norm, and tail error testing. Our construction could also likely be useful for proving lower bounds against testing of \textit{graph properties}, which is a well studied area \cite{goldreich2010introduction}. We pose the open problem to design (or demonstrate the non-existence of) additional subgraph-equivalent matrices beyond the cycle graph construction utilized in this work,  which have gaps in their spectral or graph-theoretic properties.

Additionally, we pose the open problem of determining the exact non-adaptive query complexity of PSD testing with $\ell_2^2$ gap. As discussed in Section \ref{sec:contri}, there appear to be several key barriers to improving the complexity beyond $O(1/\eps^4)$. Indeed, it seems that perhaps the main tool that is lacking is a concentration inequality for the eigenvalues of random principal submatrices. Since most such decay results apply only to norms \cite{tropp2008norms,rudelson2007sampling}, progress in this direction would likely result in important insights into eigenvalues of random matrices.

Finally, we note that the complexity of the testing problems for several matrix norms, specifically the Schatten $p$ and Ky-Fan norms, are still open in the bounded entry model. In particular, for the Schatten $1$ norm, to the best of our knowledge no non-trivial algorithms exist even for estimation with additive error $\Theta(n^{1.5})$, thus any improvements would be quite interesting.

\section*{Acknowledgement}

We thank Erik Waingarten for many useful suggestions,  being closely involved in the early stages of this project, and for feedback on early drafts of this manuscript.  We also thank Roie Levin, Ryan O'Donnell, Pedro Paredes, Nicolas Resch, and Goran Zuzic for illuminating discussions related to this project.

\bibliographystyle{alpha}
\bibliography{cluster}

\newcommand{\etalchar}[1]{$^{#1}$}
\begin{thebibliography}{IVWW19}

\bibitem[ACK{\etalchar{+}}16]{andoni2016sketching}
Alexandr Andoni, Jiecao Chen, Robert Krauthgamer, Bo~Qin, David~P Woodruff, and
  Qin Zhang.
\newblock On sketching quadratic forms.
\newblock In {\em Proceedings of the 2016 ACM Conference on Innovations in
  Theoretical Computer Science}, pages 311--319. ACM, 2016.

\bibitem[AHK05]{arora2005fast}
Sanjeev Arora, Elad Hazan, and Satyen Kale.
\newblock Fast algorithms for approximate semidefinite programming using the
  multiplicative weights update method.
\newblock In {\em 46th Annual IEEE Symposium on Foundations of Computer Science
  (FOCS'05)}, pages 339--348. IEEE, 2005.

\bibitem[AJ06]{adell2006exact}
Jos{\'e}~A Adell and Pedro Jodr{\'a}.
\newblock Exact kolmogorov and total variation distances between some familiar
  discrete distributions.
\newblock {\em Journal of Inequalities and Applications}, 2006(1):64307, 2006.

\bibitem[AMS96]{alon1996space}
Noga Alon, Yossi Matias, and Mario Szegedy.
\newblock The space complexity of approximating the frequency moments.
\newblock In {\em Proceedings of the twenty-eighth annual ACM symposium on
  Theory of computing}, pages 20--29. ACM, 1996.

\bibitem[AN13]{andoni2013eigenvalues}
Alexandr Andoni and Huy~L. Nguyen.
\newblock Eigenvalues of a matrix in the streaming model.
\newblock In {\em Proceedings of the twenty-fourth annual ACM-SIAM symposium on
  Discrete algorithms}, pages 1729--1737. Society for Industrial and Applied
  Mathematics, 2013.

\bibitem[Aro98]{arora1998polynomial}
Sanjeev Arora.
\newblock Polynomial time approximation schemes for euclidean traveling
  salesman and other geometric problems.
\newblock {\em Journal of the ACM (JACM)}, 45(5):753--782, 1998.

\bibitem[ARV09]{arora2009expander}
Sanjeev Arora, Satish Rao, and Umesh Vazirani.
\newblock Expander flows, geometric embeddings and graph partitioning.
\newblock {\em Journal of the ACM (JACM)}, 56(2):1--37, 2009.

\bibitem[{\~A}z{\"O}{\"O}93]{a1993heat}
M~Necati {\~A}-zisik, M~Necati {\"O}z{\i}s{\i}k, and M~Necati
  {\"O}z{\i}{\c{s}}{\i}k.
\newblock {\em Heat conduction}.
\newblock John Wiley \& Sons, 1993.

\bibitem[BBG18]{barman2016testing}
Siddharth Barman, Arnab Bhattacharyya, and Suprovat Ghoshal.
\newblock Testing sparsity over known and unknown bases.
\newblock In {\em International Conference on Machine Learning}, pages
  491--500, 2018.

\bibitem[BCK{\etalchar{+}}18]{braverman2018matrix}
Vladimir Braverman, Stephen Chestnut, Robert Krauthgamer, Yi~Li, David
  Woodruff, and Lin Yang.
\newblock Matrix norms in data streams: Faster, multi-pass and row-order.
\newblock In {\em International Conference on Machine Learning}, pages
  649--658, 2018.

\bibitem[BCW19]{bakshi2019robust}
Ainesh Bakshi, Nadiia Chepurko, and David~P Woodruff.
\newblock Robust and sample optimal algorithms for psd low-rank approximation.
\newblock {\em arXiv preprint arXiv:1912.04177}, 2019.

\bibitem[BH11]{balcan2011learning}
Maria-Florina Balcan and Nicholas~JA Harvey.
\newblock Learning submodular functions.
\newblock In {\em Proceedings of the forty-third annual ACM symposium on Theory
  of computing}, pages 793--802, 2011.

\bibitem[BJWY20]{ben2020framework}
Omri Ben{-}Eliezer, Rajesh Jayaram, David~P Woodruff, and Eylon Yogev.
\newblock A framework for adversarially robust streaming algorithms.
\newblock In {\em Proceedings of the 39th ACM SIGMOD-SIGACT-SIGAI Symposium on
  Principles of Database Systems}, pages 63--80, 2020.

\bibitem[BKKS19]{braverman2019schatten}
Vladimir Braverman, Robert Krauthgamer, Aditya Krishnan, and Roi Sinoff.
\newblock Schatten norms in matrix streams: Hello sparsity, goodbye dimension.
\newblock {\em arXiv preprint arXiv:1907.05457}, 2019.

\bibitem[BLWZ19]{BalcanTesting}
Maria-Florina Balcan, Yi~Li, David~P Woodruff, and Hongyang Zhang.
\newblock Testing matrix rank, optimally.
\newblock In {\em Proceedings of the Thirtieth Annual ACM-SIAM Symposium on
  Discrete Algorithms}, pages 727--746. SIAM, 2019.

\bibitem[BSST13]{batson2013spectral}
Joshua Batson, Daniel~A Spielman, Nikhil Srivastava, and Shang-Hua Teng.
\newblock Spectral sparsification of graphs: theory and algorithms.
\newblock {\em Communications of the ACM}, 56(8):87--94, 2013.

\bibitem[BVKS19]{banks2019pseudospectral}
Jess Banks, Jorge~Garza Vargas, Archit Kulkarni, and Nikhil Srivastava.
\newblock Pseudospectral shattering, the sign function, and diagonalization in
  nearly matrix multiplication time.
\newblock {\em arXiv preprint arXiv:1912.08805}, 2019.

\bibitem[BW18]{bakshi2018sublinear}
Ainesh Bakshi and David Woodruff.
\newblock Sublinear time low-rank approximation of distance matrices.
\newblock In {\em Advances in Neural Information Processing Systems}, pages
  3782--3792, 2018.

\bibitem[BZ16]{balcan2016noise}
Maria-Florina~F Balcan and Hongyang Zhang.
\newblock Noise-tolerant life-long matrix completion via adaptive sampling.
\newblock In {\em Advances in Neural Information Processing Systems}, pages
  2955--2963, 2016.

\bibitem[Chu96]{chung1996lectures}
Fan~RK Chung.
\newblock Lectures on spectral graph theory.
\newblock {\em Lecture Notes}, 1996.

\bibitem[CRT06]{candes2006stable}
Emmanuel~J Candes, Justin~K Romberg, and Terence Tao.
\newblock Stable signal recovery from incomplete and inaccurate measurements.
\newblock {\em Communications on Pure and Applied Mathematics: A Journal Issued
  by the Courant Institute of Mathematical Sciences}, 59(8):1207--1223, 2006.

\bibitem[CW17]{clarkson2017low}
Kenneth~L Clarkson and David~P Woodruff.
\newblock Low-rank approximation and regression in input sparsity time.
\newblock {\em Journal of the ACM (JACM)}, 63(6):54, 2017.

\bibitem[Dat10]{dattorro2010convex}
Jon Dattorro.
\newblock {\em Convex optimization \& Euclidean distance geometry}.
\newblock Lulu. com, 2010.

\bibitem[DDHK07]{demmel2007fast2}
James Demmel, Ioana Dumitriu, Olga Holtz, and Robert Kleinberg.
\newblock Fast matrix multiplication is stable.
\newblock {\em Numerische Mathematik}, 106(2):199--224, 2007.

\bibitem[DJS{\etalchar{+}}19]{diao2019optimal}
Huaian Diao, Rajesh Jayaram, Zhao Song, Wen Sun, and David Woodruff.
\newblock Optimal sketching for kronecker product regression and low rank
  approximation.
\newblock In {\em Advances in Neural Information Processing Systems}, pages
  4737--4748, 2019.

\bibitem[DK19]{diakonikolas2019recent}
Ilias Diakonikolas and Daniel~M Kane.
\newblock Recent advances in algorithmic high-dimensional robust statistics.
\newblock {\em arXiv preprint arXiv:1911.05911}, 2019.

\bibitem[DL09]{deza2009geometry}
Michel~Marie Deza and Monique Laurent.
\newblock {\em Geometry of cuts and metrics}, volume~15.
\newblock Springer, 2009.

\bibitem[EK12]{eldar2012compressed}
Yonina~C Eldar and Gitta Kutyniok.
\newblock {\em Compressed sensing: theory and applications}.
\newblock Cambridge university press, 2012.

\bibitem[GGR98]{goldreich1998property}
Oded Goldreich, Shari Goldwasser, and Dana Ron.
\newblock Property testing and its connection to learning and approximation.
\newblock {\em Journal of the ACM (JACM)}, 45(4):653--750, 1998.

\bibitem[GI10]{gilbert2010sparse}
Anna Gilbert and Piotr Indyk.
\newblock Sparse recovery using sparse matrices.
\newblock {\em Proceedings of the IEEE}, 98(6):937--947, 2010.

\bibitem[Gle94]{glendinning1994stability}
Paul Glendinning.
\newblock {\em Stability, instability and chaos: an introduction to the theory
  of nonlinear differential equations}, volume~11.
\newblock Cambridge university press, 1994.

\bibitem[GLSS18]{garg2018matrix}
Ankit Garg, Yin~Tat Lee, Zhao Song, and Nikhil Srivastava.
\newblock A matrix expander chernoff bound.
\newblock In {\em Proceedings of the 50th Annual ACM SIGACT Symposium on Theory
  of Computing}, pages 1102--1114, 2018.

\bibitem[Gol10]{goldreich2010introduction}
Oded Goldreich.
\newblock Introduction to testing graph properties.
\newblock In {\em Property testing}, pages 105--141. Springer, 2010.

\bibitem[Gol17]{goldreich2017introduction}
Oded Goldreich.
\newblock {\em Introduction to property testing}.
\newblock Cambridge University Press, 2017.

\bibitem[GT11]{gittens2011tail}
Alex Gittens and Joel~A Tropp.
\newblock Tail bounds for all eigenvalues of a sum of random matrices.
\newblock {\em arXiv preprint arXiv:1104.4513}, 2011.

\bibitem[GW95]{goemans1995improved}
Michel~X Goemans and David~P Williamson.
\newblock Improved approximation algorithms for maximum cut and satisfiability
  problems using semidefinite programming.
\newblock {\em Journal of the ACM (JACM)}, 42(6):1115--1145, 1995.

\bibitem[HMAS17]{han2017approximating}
Insu Han, Dmitry Malioutov, Haim Avron, and Jinwoo Shin.
\newblock Approximating spectral sums of large-scale matrices using stochastic
  chebyshev approximations.
\newblock {\em SIAM Journal on Scientific Computing}, 39(4):A1558--A1585, 2017.

\bibitem[Ind06]{indyk2006stable}
Piotr Indyk.
\newblock Stable distributions, pseudorandom generators, embeddings, and data
  stream computation.
\newblock {\em Journal of the ACM (JACM)}, 53(3):307--323, 2006.

\bibitem[IVWW19]{indyk2019sample}
Piotr Indyk, Ali Vakilian, Tal Wagner, and David Woodruff.
\newblock Sample-optimal low-rank approximation of distance matrices.
\newblock {\em arXiv preprint arXiv:1906.00339}, 2019.

\bibitem[JST11]{Jowhari:2011}
Hossein Jowhari, Mert Sa\u{g}lam, and G\'{a}bor Tardos.
\newblock Tight bounds for lp samplers, finding duplicates in streams, and
  related problems.
\newblock In {\em Proceedings of the Thirtieth ACM SIGMOD-SIGACT-SIGART
  Symposium on Principles of Database Systems}, PODS '11, pages 49--58, New
  York, NY, USA, 2011. ACM.

\bibitem[JSTW19]{jayaram2019weighted}
Rajesh Jayaram, Gokarna Sharma, Srikanta Tirthapura, and David~P Woodruff.
\newblock Weighted reservoir sampling from distributed streams.
\newblock In {\em Proceedings of the 38th ACM SIGMOD-SIGACT-SIGAI Symposium on
  Principles of Database Systems}, pages 218--235, 2019.

\bibitem[JW18]{jayaram2018perfect}
Rajesh Jayaram and David~P Woodruff.
\newblock Perfect lp sampling in a data stream.
\newblock In {\em 2018 IEEE 59th Annual Symposium on Foundations of Computer
  Science (FOCS)}, pages 544--555. IEEE, 2018.

\bibitem[JW19]{jayaram2019towards}
Rajesh Jayaram and David~P Woodruff.
\newblock Towards optimal moment estimation in streaming and distributed
  models.
\newblock In {\em Approximation, Randomization, and Combinatorial Optimization.
  Algorithms and Techniques (APPROX/RANDOM 2019)}. Schloss
  Dagstuhl-Leibniz-Zentrum fuer Informatik, 2019.

\bibitem[KBV09]{koren2009matrix}
Yehuda Koren, Robert Bell, and Chris Volinsky.
\newblock Matrix factorization techniques for recommender systems.
\newblock {\em Computer}, 42(8):30--37, 2009.

\bibitem[KIDP16]{kannan2016bounded}
Ramakrishnan Kannan, Mariya Ishteva, Barry Drake, and Haesun Park.
\newblock Bounded matrix low rank approximation.
\newblock In {\em Non-negative Matrix Factorization Techniques}, pages 89--118.
  Springer, 2016.

\bibitem[KLS20]{kyng2020four}
Rasmus Kyng, Kyle Luh, and Zhao Song.
\newblock Four deviations suffice for rank 1 matrices.
\newblock {\em Advances in Mathematics}, 375:107366, 2020.

\bibitem[Kna09]{knapp2009sines}
Michael~P Knapp.
\newblock Sines and cosines of angles in arithmetic progression.
\newblock {\em Mathematics Magazine}, 82(5):371--372, 2009.

\bibitem[KNW10]{kane2010exact}
Daniel~M Kane, Jelani Nelson, and David~P Woodruff.
\newblock On the exact space complexity of sketching and streaming small norms.
\newblock In {\em Proceedings of the twenty-first annual ACM-SIAM symposium on
  Discrete Algorithms}, pages 1161--1178. SIAM, 2010.

\bibitem[KOSZ13]{kelner2013simple}
Jonathan~A Kelner, Lorenzo Orecchia, Aaron Sidford, and Zeyuan~Allen Zhu.
\newblock A simple, combinatorial algorithm for solving sdd systems in
  nearly-linear time.
\newblock In {\em Proceedings of the forty-fifth annual ACM symposium on Theory
  of computing}, pages 911--920, 2013.

\bibitem[KS03]{krauthgamer2003property}
Robert Krauthgamer and Ori Sasson.
\newblock Property testing of data dimensionality.
\newblock In {\em Proceedings of the fourteenth annual ACM-SIAM symposium on
  Discrete algorithms}, pages 18--27. Society for Industrial and Applied
  Mathematics, 2003.

\bibitem[KS18]{kyng2018matrix}
Rasmus Kyng and Zhao Song.
\newblock A matrix chernoff bound for strongly rayleigh distributions and
  spectral sparsifiers from a few random spanning trees.
\newblock In {\em 2018 IEEE 59th Annual Symposium on Foundations of Computer
  Science (FOCS)}, pages 373--384. IEEE, 2018.

\bibitem[LNW14]{li2014sketching}
Yi~Li, Huy~L Nguyen, and David~P Woodruff.
\newblock On sketching matrix norms and the top singular vector.
\newblock In {\em Proceedings of the twenty-fifth annual ACM-SIAM symposium on
  Discrete algorithms}, pages 1562--1581. SIAM, 2014.

\bibitem[LNW19]{li2019approximating}
Yi~Li, Huy~L Nguyen, and David~P Woodruff.
\newblock On approximating matrix norms in data streams.
\newblock {\em SIAM Journal on Computing}, 48(6):1643--1697, 2019.

\bibitem[LPP91]{lust1991non}
Fran{\c{c}}oise Lust-Piquard and Gilles Pisier.
\newblock Non commutative khintchine and paley inequalities.
\newblock {\em Arkiv f{\"o}r matematik}, 29(1-2):241--260, 1991.

\bibitem[LW16a]{li2016approximating}
Yi~Li and David~P Woodruff.
\newblock On approximating functions of the singular values in a stream.
\newblock In {\em Proceedings of the forty-eighth annual ACM symposium on
  Theory of Computing}, pages 726--739, 2016.

\bibitem[LW16b]{li2016tight}
Yi~Li and David~P Woodruff.
\newblock Tight bounds for sketching the operator norm, schatten norms, and
  subspace embeddings.
\newblock In {\em Approximation, Randomization, and Combinatorial Optimization.
  Algorithms and Techniques (APPROX/RANDOM 2016)}. Schloss
  Dagstuhl-Leibniz-Zentrum fuer Informatik, 2016.

\bibitem[LW17]{li2017embeddings}
Yi~Li and David~P Woodruff.
\newblock Embeddings of schatten norms with applications to data streams.
\newblock In {\em 44th International Colloquium on Automata, Languages, and
  Programming (ICALP 2017)}. Schloss Dagstuhl-Leibniz-Zentrum fuer Informatik,
  2017.

\bibitem[LWW14]{li2014improved}
Yi~Li, Zhengyu Wang, and David~P Woodruff.
\newblock Improved testing of low rank matrices.
\newblock In {\em Proceedings of the 20th ACM SIGKDD international conference
  on Knowledge discovery and data mining}, pages 691--700, 2014.

\bibitem[MSS15]{marcus2015interlacing}
Adam~W Marcus, Daniel~A Spielman, and Nikhil Srivastava.
\newblock Interlacing families ii: Mixed characteristic polynomials and the
  kadison—singer problem.
\newblock {\em Annals of Mathematics}, pages 327--350, 2015.

\bibitem[MW10]{monemizadeh20101}
Morteza Monemizadeh and David~P Woodruff.
\newblock 1-pass relative-error lp-sampling with applications.
\newblock In {\em Proceedings of the twenty-first annual ACM-SIAM symposium on
  Discrete Algorithms}, pages 1143--1160. SIAM, 2010.

\bibitem[MW17]{musco2017sublinear}
Cameron Musco and David~P Woodruff.
\newblock Sublinear time low-rank approximation of positive semidefinite
  matrices.
\newblock In {\em 2017 IEEE 58th Annual Symposium on Foundations of Computer
  Science (FOCS)}, pages 672--683. IEEE, 2017.

\bibitem[Pis09]{pisier2009remarks}
Gilles Pisier.
\newblock Remarks on the non-commutative khintchine inequalities for 0< p< 2.
\newblock {\em Journal of Functional Analysis}, 256(12):4128--4161, 2009.

\bibitem[PR03]{parnas2003testing}
Michal Parnas and Dana Ron.
\newblock Testing metric properties.
\newblock {\em Information and Computation}, 187(2):155--195, 2003.

\bibitem[PR17]{pisier2017non}
Gilles Pisier and {\'E}ric Ricard.
\newblock The non-commutative khintchine inequalities for $0< p< 1$.
\newblock {\em Journal of the Institute of Mathematics of Jussieu},
  16(5):1103--1123, 2017.

\bibitem[RV07]{rudelson2007sampling}
Mark Rudelson and Roman Vershynin.
\newblock Sampling from large matrices: An approach through geometric
  functional analysis.
\newblock {\em Journal of the ACM (JACM)}, 54(4):21--es, 2007.

\bibitem[Sch35]{schoenberg1935remarks}
Isaac~J Schoenberg.
\newblock Remarks to maurice frechet's article``sur la definition axiomatique
  d'une classe d'espace distances vectoriellement applicable sur l'espace de
  hilbert.
\newblock {\em Annals of Mathematics}, pages 724--732, 1935.

\bibitem[SKZ14]{saade2014spectral}
Alaa Saade, Florent Krzakala, and Lenka Zdeborov{\'a}.
\newblock Spectral clustering of graphs with the bethe hessian.
\newblock In {\em Advances in Neural Information Processing Systems}, pages
  406--414, 2014.

\bibitem[SL{\etalchar{+}}91]{slotine1991applied}
Jean-Jacques~E Slotine, Weiping Li, et~al.
\newblock {\em Applied nonlinear control}, volume 199.
\newblock Prentice hall Englewood Cliffs, NJ, 1991.

\bibitem[ST04]{spielman2004nearly}
Daniel~A Spielman and Shang-Hua Teng.
\newblock Nearly-linear time algorithms for graph partitioning, graph
  sparsification, and solving linear systems.
\newblock In {\em Proceedings of the thirty-sixth annual ACM symposium on
  Theory of computing}, pages 81--90, 2004.

\bibitem[ST11]{spielman2011spectral}
Daniel~A Spielman and Shang-Hua Teng.
\newblock Spectral sparsification of graphs.
\newblock {\em SIAM Journal on Computing}, 40(4):981--1025, 2011.

\bibitem[Ste10]{steurer2010fast}
David Steurer.
\newblock Fast sdp algorithms for constraint satisfaction problems.
\newblock In {\em Proceedings of the twenty-first annual ACM-SIAM symposium on
  Discrete Algorithms}, pages 684--697. SIAM, 2010.

\bibitem[SWYZ19]{sun2019querying}
Xiaoming Sun, David~P Woodruff, Guang Yang, and Jialin Zhang.
\newblock Querying a matrix through matrix-vector products.
\newblock {\em arXiv preprint arXiv:1906.05736}, 2019.

\bibitem[SZ20]{song2020hyperbolic}
Zhao Song and Ruizhe Zhang.
\newblock Hyperbolic polynomials i: Concentration and discrepancy.
\newblock {\em arXiv preprint arXiv:2008.09593}, 2020.

\bibitem[Tao11]{tao2011topics}
Terence Tao.
\newblock Topics in random matrix theory.
\newblock {\em Lecture Notes}, 2011.

\bibitem[Tho72]{thompson1972principal}
Robert~C Thompson.
\newblock Principal submatrices ix: Interlacing inequalities for singular
  values of submatrices.
\newblock {\em Linear Algebra and its Applications}, 5(1):1--12, 1972.

\bibitem[Tro08]{tropp2008norms}
Joel~A Tropp.
\newblock Norms of random submatrices and sparse approximation.
\newblock {\em Comptes Rendus Mathematique}, 346(23-24):1271--1274, 2008.

\bibitem[Tro15]{tropp2015introduction}
Joel Tropp.
\newblock An introduction to matrix concentration inequalities.
\newblock {\em Foundations and Trends{\textregistered} in Machine Learning},
  8(1-2):1--230, 2015.

\bibitem[Tur48]{turing1948rounding}
Alan~M Turing.
\newblock Rounding-off errors in matrix processes.
\newblock {\em The Quarterly Journal of Mechanics and Applied Mathematics},
  1(1):287--308, 1948.

\bibitem[VB96]{vandenberghe1996semidefinite}
Lieven Vandenberghe and Stephen Boyd.
\newblock Semidefinite programming.
\newblock {\em SIAM review}, 38(1):49--95, 1996.

\bibitem[Ver10]{vershynin2010introduction}
Roman Vershynin.
\newblock Introduction to the non-asymptotic analysis of random matrices.
\newblock {\em arXiv preprint arXiv:1011.3027}, 2010.

\bibitem[Wai19]{wainwright2019high}
Martin~J Wainwright.
\newblock {\em High-dimensional statistics: A non-asymptotic viewpoint},
  volume~48.
\newblock Cambridge University Press, 2019.

\bibitem[WSV12]{wolkowicz2012handbook}
Henry Wolkowicz, Romesh Saigal, and Lieven Vandenberghe.
\newblock {\em Handbook of semidefinite programming: theory, algorithms, and
  applications}, volume~27.
\newblock Springer Science \& Business Media, 2012.

\end{thebibliography}

\appendix

\section{Proof of Eigenvalue Identity}
\label{sec:appendixA}

\begin{proposition}
	Let $\AA \in \R^{n \times n}$ be any real symmetric matrix. 
	Then $\min_{\BB \succeq 0} \|\AA - \BB\|_F^2 = \sum_{i : \lambda_i(\AA) < 0 } \lambda_i^2(\AA)$.
\end{proposition}
\begin{proof}
	 Let $g_i$ be the eigenvector associated with $\lambda_i = \lambda_i(\AA)$.
	 First, setting $\BB = \sum_{i : \lambda_i(\AA) \geq  0 } \lambda_i g_i g_i^\top$, which is a PSD matrix, we have $\|\AA - \BB\|_F^2 = \| \sum_{i : \lambda_i(\AA) < 0 } \lambda_i(g_i g_i^\top \AA) \|_2^2 = \sum_{i : \lambda_i(\AA) < 0 } \lambda_i^2(\AA)$, where the second equality follows from the Pythagorean Theorem, which proves that $\min_{\BB \succeq 0} \|\AA - \BB\|_F^2 \leq \sum_{i : \lambda_i(\AA) < 0 } \lambda_i^2(\AA)$. To see the other direction, fix any PSD matrix $\BB$, and let $\ZZ =  \BB - \AA$.  Then $\ZZ + \AA \succeq 0$, where $\succeq$ is the Lowner ordering, thus $\ZZ \succeq -\AA$, which by definition implies that $x^\top \ZZ x \geq - x^\top \AA x$ for all $x \in \R^n$. Then by the Courant-Fischer variational characterization of eigenvalues, we have that $\lambda_i(\ZZ) \geq -\lambda_i(\AA)$ for all $i$. In particular, $|\lambda_i(\ZZ)| \geq |\lambda_i(\AA)|$ for all $i$ such that $\lambda_i(\AA) < 0$.
Thus $\|\ZZ\|_F^2 = \sum_i \lambda_i^2(\ZZ) \geq   \sum_{i : \lambda_i(\AA) < 0 } \lambda_i^2(\ZZ)  \geq \sum_{i : \lambda_i(\AA) < 0 } \lambda_i^2(\AA)$, which completes the proof.

\end{proof}

\end{document}